\def\llncs{0}
\def\fullpage{1}
\def\anonymous{0}
\def\authnote{1}
\def\notxfont{0}
\def\submission{0}%when submit 30page version to conference, it is 1. For full version, it is 0.

\def\extendedabstract{0}

\ifnum\submission=1
\def\llncs{1}
\fi

\ifnum\llncs=1
\documentclass[envcountsect,a4paper,runningheads,10pt]{llncs}%default is 10pt
\else
	\documentclass[letterpaper,hmargin=1.05in,vmargin=1.05in,11pt]{article}
			\ifnum\fullpage=1
		\usepackage{fullpage}
		\fi
\fi

\usepackage{CJKutf8}
\usepackage[%
  colorlinks=true,
  citecolor=blue,
  pagebackref=true
]{hyperref}

\usepackage{amsmath, amsfonts, amssymb, mathtools,amscd}

\usepackage{amsthm}

\usepackage{lmodern}
\usepackage[T1]{fontenc}
\usepackage[utf8]{inputenc}

\usepackage{etex}

\usepackage{arydshln} % In order to use \hdashline
\usepackage{url}
\usepackage{ifthen}
\usepackage{bm}
\usepackage{multirow}
\usepackage[dvips]{graphicx}

\usepackage[usenames]{color}
\usepackage{xcolor,colortbl} % Leave out in case the usepackage cannot be found. Not critical
\usepackage{threeparttable}
\usepackage{comment}
\usepackage{paralist,verbatim}
\usepackage{cases}
\usepackage{booktabs}
\usepackage{braket}
\usepackage{cancel} 
\usepackage{ascmac} 
\usepackage{framed}
\usepackage{authblk}
\usepackage{pifont}
\usepackage{qcircuit}
\usepackage{tikz}
\usetikzlibrary{cd}
\definecolor{darkblue}{rgb}{0,0,0.6}
\definecolor{darkgreen}{rgb}{0,0.5,0}
\definecolor{maroon}{rgb}{0.5,0.1,0.1}
\definecolor{dpurple}{rgb}{0.2,0,0.65}

\usepackage[capitalise,noabbrev]{cleveref}
\usepackage[absolute]{textpos}
\usepackage[final]{microtype}
\usepackage[absolute]{textpos}
\usepackage{everypage}
\DeclareMathAlphabet{\mathpzc}{OT1}{pzc}{m}{it}

\usepackage{algorithmic}
\usepackage{algorithm}
\usepackage{here}

\usepackage{thm-restate}
\usepackage[normalem]{ulem}

\newtheoremstyle{thicktheorem}%
{\topsep}
{\topsep}
{\itshape}{}%
{\bfseries}%
{.}
{ }%
{\thmname{#1}\thmnumber{ #2}%
		\thmnote{ (#3)}%
}

\newtheoremstyle{remark}%name
{\topsep}
{\topsep}
	{}%body font
	{}%indent amount
	{}%theorem head font
	{.}%punctuation after theorem head
	{ }%space after theorem head
	{\textit{\thmname{#1}}\thmnumber{ #2}%theorem head specs
			\thmnote{ (#3)}%
	}

\ifnum\llncs=0
	\theoremstyle{thicktheorem}
	\newtheorem{theorem}{Theorem}[section]
	\newtheorem{lemma}[theorem]{Lemma}
	\newtheorem{corollary}[theorem]{Corollary}
	\newtheorem{proposition}[theorem]{Proposition}
	\newtheorem{definition}[theorem]{Definition}

        \newtheorem{construction}[theorem]{Construction}

	\theoremstyle{remark}
	\newtheorem{claim}[theorem]{Claim}
	\newtheorem{remark}[theorem]{Remark}

\else
\fi
	\crefname{theorem}{Theorem}{Theorems}
	\crefname{assumption}{Assumption}{Assumptions}
	\crefname{construction}{Construction}{Constructions}
	\crefname{corollary}{Corollary}{Corollaries}
	\crefname{conjecture}{Conjecture}{Conjectures}
	\crefname{definition}{Definition}{Definitions}
	\crefname{exmaple}{Example}{Examples}
	\crefname{experiment}{Experiment}{Experiments}
	\crefname{counterexample}{Counterexample}{Counterexamples}
	\crefname{lemma}{Lemma}{Lemmata}
	\crefname{observation}{Observation}{Observations}
	\crefname{proposition}{Proposition}{Propositions}
	\crefname{remark}{Remark}{Remarks}
	\crefname{claim}{Claim}{Claims}
	\crefname{fact}{Fact}{Facts}
	\crefname{note}{Note}{Notes}

\ifnum\llncs=1
 \crefname{appendix}{App.}{Appendices}
 \crefname{section}{Sec.}{Sections}
\else
\fi

\ifnum\llncs=1
\renewcommand*{\backref}[1]{}
\else
	\renewcommand*{\backref}[1]{(Cited on page~#1.)}
	\ifnum\notxfont=1
	\else
		\usepackage{newtxtext}
	\fi
\fi

\usepackage{fancyhdr}

\ifnum\authnote=0  %%%%% Remove comments %%%%%
\newcommand{\mor}[1]{}
\newcommand{\shogo}[1]{}
\newcommand{\takashi}[1]{}
\newcommand{\fuyuki}[1]{}
\newcommand{\minki}[1]{}

\else
\newcommand{\shogo}[1]{$\ll$\textsf{\color{darkgreen} Shogo: { #1}}$\gg$}

\DeclareRobustCommand{\Erase}{\bgroup\markoverwith{\textcolor{red}{\rlap{\rule[0.7ex]{2pt}{0.4pt}}\rule[0.3ex]{2pt}{0.4pt}}}\ULon}

\fi

\newcommand{\CZ}{\mathrm{CZ}}

\newcommand{\Support}{\mathrm{supp}}

\newcommand{\Tr}{\mathrm{Tr}}

\newcommand{\good}{\mathsf{Good}}

\newcommand{\cA}{\mathcal{A}}
\newcommand{\cB}{\mathcal{B}}

\newcommand{\cD}{\mathcal{D}}
\newcommand{\cE}{\mathcal{E}}
\newcommand{\cF}{\mathcal{F}}
\newcommand{\cG}{\mathcal{G}}
\newcommand{\cH}{\mathcal{H}}
\newcommand{\cI}{\mathcal{I}}
\newcommand{\cK}{\mathcal{K}}

\newcommand{\cN}{\mathcal{N}}
\newcommand{\cO}{\mathcal{O}}
\newcommand{\cP}{\mathcal{P}}

\newcommand{\cS}{\mathcal{S}}
\newcommand{\cT}{\mathcal{T}}
\newcommand{\cU}{\mathcal{U}}

\newcommand{\cX}{\mathcal{X}}
\newcommand{\cY}{\mathcal{Y}}
\newcommand{\cZ}{\mathcal{Z}}

\newcommand{\identitymap}{\mathrm{id}}

\def\makeuppercase#1{
\expandafter\newcommand\csname tl#1\endcsname{\widetilde{#1}}
}

\def\makelowercase#1{
\expandafter\newcommand\csname tl#1\endcsname{\widetilde{#1}}
}

\newcommand{\N}{\mathbb{N}}

\newcommand{\C}{\mathbb{C}}

\newcommand{\F}{\mathbb{F}}

\newcommand{\regA}{\mathbf{A}}
\newcommand{\regB}{\mathbf{B}}
\newcommand{\regC}{\mathbf{C}}
\newcommand{\regD}{\mathbf{D}}
\newcommand{\regE}{\mathbf{E}}

\newcommand{\regR}{\mathbf{R}}

\newcommand{\regX}{\mathbf{X}}

\newcommand{\secp}{\lambda}

\newcommand{\weight}{\mathrm{wt}}

\newcommand{\B}{\entity{B}}

\newcommand{\LPN}{\mathrm{LPN}}
\newcommand{\XOR}{\mathrm{XOR}}

\newcommand*{\sk}{\keys{sk}}
\newcommand*{\pk}{\keys{pk}}

\newcommand*{\keys}[1]{\mathsf{#1}}

\newcommand*{\algo}[1]{\ensuremath{\mathsf{#1}}}

\newcommand*{\entity}[1]{\mathcal{#1}}
\newenvironment{boxfig}[2]{\begin{figure}[#1]\fbox{\begin{minipage}{0.97\linewidth}
                        \vspace{0.2em}
                        \makebox[0.025\linewidth]{}
                        \begin{minipage}{0.95\linewidth}
            {{
                        #2 }}
                        \end{minipage}
                        \vspace{0.2em}
                        \end{minipage}}}{\end{figure}}

\newcommand{\bit}{\{0,1\}}

\newcommand{\trit}{\{0,1,2\}}

\newcommand{\KeyGen}{\algo{KeyGen}}

\newcommand{\Enc}{\algo{Enc}}

\newcommand{\Dec}{\algo{Dec}}

\newcommand{\PRC}{\algo{PRC}}
\newcommand{\PRFC}{\algo{PRFC}}

\newcommand{\negl}{{\mathsf{negl}}}

\newcommand{\poly}{{\mathrm{poly}}}

\DeclareMathOperator*{\Exp}{\mathbb{E}}

\newcommand{\Ber}{\mathrm{Ber}}

\newcommand{\XMAPSTO}[1]{\xmapsto{\makebox[3.5em]{\hfil$#1$\hfil}}}

\newcommand{\phys}{\mathrm{P}}
\newcommand{\logic}{\mathrm{L}}
\newcommand{\QECC}{\algo{QECC}}

\usetikzlibrary{decorations.markings}
\tikzset{
  cross/.style={
    postaction={decorate,decoration={markings,
    mark=at position 0.45 with {\draw[-,line width=1pt] (-10pt,-10pt) -- (10pt,10pt);\draw[-,line width=1pt] (-10pt,10pt) -- (10pt,-10pt);}}}
  }
}

\makeatletter
\DeclareRobustCommand
  \myvdots{\vbox{\baselineskip4\p@ \lineskiplimit\z@
    \hbox{.}\hbox{.}\hbox{.}}}
\makeatother

\newcommand{\ketbra}[2]{\lvert #1 \rangle\mkern-3mu \langle #2 \rvert}

\theoremstyle{definition}

\newtheorem{assumption}{Assumption}

\title{Quantum Pseudorandom Error-Correcting Codes}

\ifnum\anonymous=1
\ifnum\llncs=1
\author{\empty}\institute{\empty}
\else
\author{}
\fi
\else
\ifnum\llncs=1
\author{
}
\institute{
	Yukawa Institute for Theoretical Physics, Kyoto University
}
\else
\author[1]{Min-Hsiu Hsieh
\thanks{Email: \texttt{min-hsiu.hsieh@foxconn.com}}} 

\author[2,1]{ Shogo Yamada
\thanks{Email: \texttt{shogo.yamada@yukawa.kyoto-u.ac.jp}}}

\affil[1]{Hon Hai Research Institute, Taipei, Taiwan}
\affil[2]{{Yukawa Institute for Theoretical Physics, Kyoto University, Kyoto, Japan}
}

\fi %%%%% END OF LNCS branch
\fi

\date{\today}

\begin{document}
\begin{CJK}{UTF8}{ipxm}
%日本語コメント用

%macros for this paper
\newcommand{\GraphSample}{\algo{GraphSample}}
\newcommand{\Recover}{\algo{Recover}}
\newcommand{\BitDec}{\algo{BitDec}}
\newcommand{\PhaseDec}{\algo{PhaseDec}}

\maketitle
\begin{abstract}
Pseudorandom error-correcting codes (PRCs), introduced by Christ and Gunn [CRYPTO 2024], are classical error-correcting codes whose codewords are computationally indistinguishable from uniformly random strings.
We initiate the study of quantum analogues of PRCs.
We define quantum pseudorandom error-correcting codes (QPRCs).
Assuming that Learning Parity with Noise (LPN) is hard for $2^{O(\sqrt{n})}$-time quantum algorithms, we construct QPRCs whose encodings are indistinguishable from Haar-random isometries.
We call these pseudorandom isometric error-correcting codes (PRICs), and our construction is robust to all $o(n \frac{\log \log n}{\log n})$-local quantum noise, where $n$ is the number of physical qubits.
Under the same assumption, we also construct QPRCs whose encodings are indistinguishable from the completely depolarizing channel and that are robust to all $\alpha n$-local quantum noise for some constant $\alpha>0$.
The latter QPRCs give a direct quantum analogue of classical PRCs.

To construct PRICs, we develop two technical ingredients of independent interest.
First, we introduce a new classical cryptographic primitive which we call pseudorandom functional error-correcting codes (PRFCs) and construct them under the same LPN assumption.
PRFCs are pseudorandom functions equipped with a polynomial-time decoder that uses the secret key to recover the input even after a constant fraction of the output bits are flipped.
Second, we develop a new efficient decoding procedure within the codeword-stabilized (CWS) framework introduced in [Cross, Smith, Smolin, and Zeng, IEEE ISIT 2008], which is a generic framework for constructing quantum error-correcting codes by combining (possibly nonlinear) classical error-correcting codes and graphs.
For CWS codes based on any efficiently decodable classical code that corrects the induced bit errors, including nonlinear codes, our decoding procedure corrects all $o(n \frac{\log \log n}{\log n})$-local quantum noise, with high probability over our sampled graphs.
This gives a general way to construct keyed quantum error-correcting codes with efficient decoders by using nonlinear classical error-correcting codes as a black box, and resolves the open problem of finding a general method for efficiently decoding such CWS codes based on nonlinear classical codes, raised in [Li, Dumer, Grassl, and Pryadko, Physical Review A 2010].

\end{abstract}

\ifnum\submission=0
\clearpage
\newpage
\setcounter{tocdepth}{2}
\tableofcontents
\newpage
\fi
\section{Introduction}
\label{sec:introduction}

Pseudorandomness and error correction are two central themes in theoretical computer science, yet they impose competing requirements.
Pseudorandom objects are efficiently generated and computationally indistinguishable from uniformly random ones, and they play crucial roles in many fields, including cryptography and complexity theory.
Error-correcting codes encode classical messages into longer bit strings from which the original messages can be recovered after noise.\footnote{Unless stated otherwise, error-correcting codes considered in this work have polynomial-time encoding and decoding algorithms.}
Error correction relies on structure for decoding, whereas pseudorandomness requires hiding that structure from the adversary.

Nevertheless, Christ and Gunn~\cite{C:ChrGun24} achieved pseudorandomness and error correction simultaneously in the classical setting by constructing pseudorandom error-correcting codes (PRCs).
Informally, PRCs are error-correcting codes whose codewords are computationally indistinguishable from uniformly random strings.

This classical combination has applications in both cryptography and coding theory.
Cryptographic applications include watermarking for language models, covert communication between AI agents without a pre-shared key, and noise-robust steganography~\cite{C:ChrGun24,STOC:AACDG25,PNR26}.
In coding theory, Lu, Silbak, and Wichs~\cite{LSW26} used PRCs as a key ingredient to construct the first error-correcting codes achieving Shannon capacity against polynomial-time adversaries with access to the encoder.

In the quantum setting, error correction and pseudorandomness also have central roles independently.
Quantum error-correcting codes (QECCs) encode quantum states into a larger register to protect them against noise~\cite{Shor95,CalderbankShor96,Steane96}.
Quantum pseudorandom primitives include pseudorandom states (PRSs), pseudorandom unitaries (PRUs), and pseudorandom isometries (PRIs)~\cite{C:JiLiuSon18,EC:AGKL24,FOCS:MPSY24,STOC:MaHua25}, which are computationally indistinguishable from their Haar-random counterparts.
Here, Haar-random refers to uniform sampling of pure states, unitaries, or isometries.
These primitives have applications in cryptography~\cite{C:JiLiuSon18,C:MorYam22,EC:AGKL24}, learning theory~\cite{SHH24,Quantum:ZhaoEtAl24}, and theoretical physics~\cite{SHH24,SMLBH25}.

Combining pseudorandomness and error correction in the quantum setting introduces additional structural challenges.
To recover a state entangled with a reference system after noise, a QECC must correct phase errors as well as bit errors, requiring structure beyond that needed to decode classical messages.
At the same time, the encoding operation must be indistinguishable from Haar-random quantum operations, which are more complex than classical random functions, even on superposition inputs.

This leads to the following question:
\begin{center}
    \it Can we construct QECCs whose encodings are pseudorandom as quantum operations?
\end{center}

\ifnum\extendedabstract=0
\subsection{This Work}
\else
\paragraph{This work.}
\fi

In this work, we define quantum pseudorandom error-correcting codes (QPRCs) as a quantum variant of PRCs and construct them under the Learning Parity with Noise (LPN) assumption~\cite{Pie12}, which is a standard cryptographic assumption.
Informally, for an error $\cE$ and an ideal distribution $\cD$, we define a QPRC as a pair of quantum polynomial-time (QPT) algorithms: an encoder $\Enc$ that encodes $n_\logic$ logical qubits into $n_\phys$ physical qubits using a key $\sk$, and a decoder $\Dec$ that decodes $n_\phys$ physical qubits to $n_\logic$ logical qubits using $\sk$.\ifnum\extendedabstract=0\footnote{The formal definition also includes a key generation algorithm.
For details, see \cref{def:PK-QPRC,def:SK-QPRC}.} 
\else
\footnote{The formal definition also includes a key generation algorithm. 
For details, see the full version.}
\fi 
We require the following robustness and pseudorandomness properties.
\begin{itemize}
    \item Robustness: the channel $\Exp_{\sk}[\Dec(\sk,\cdot)\circ\cE\circ\Enc(\sk,\cdot)]$ is negligibly close to the identity channel in diamond distance.

    \item Pseudorandomness: no QPT adversary can distinguish oracle access to $\Enc(\sk,\cdot)$ from oracle access to a channel sampled from $\cD$.
\end{itemize}

We call QPRCs that are pseudorandom with respect to Haar-random isometries pseudorandom isometric error-correcting codes (PRICs).
As a main result, we construct PRICs robust to local quantum noise under the LPN assumption.
Informally, $t$-local quantum noise includes arbitrary noise acting on any set of at most $t$ physical qubits.

\ifnum\extendedabstract=0
\begin{theorem}[\cref{coro:PRIC_from_assumption}, informal]\label{thm:PRIC_intro}
    Assume that LPN is hard for $2^{O(\sqrt{n})}$-time quantum algorithms.
    Then, for every function $t=o(n_\phys\log\log n_\phys/\log n_\phys)$, there exist PRICs robust to any $t$-local quantum noise.
    In addition, these PRICs use $n_\phys=O(n_\logic)$ physical qubits for $n_\logic$ logical qubits.
\end{theorem}
\else
\begin{theorem}\label{thm:PRIC_intro}
    Assume that LPN is hard for $2^{O(\sqrt{n})}$-time quantum algorithms.
    Then, for every function $t=o(n_\phys\log\log n_\phys/\log n_\phys)$, there exist PRICs robust to any $t$-local quantum noise.
    In addition, these PRICs use $n_\phys=O(n_\logic)$ physical qubits for $n_\logic$ logical qubits.
\end{theorem}
\fi

\Cref{thm:PRIC_intro} shows that strong pseudorandomness (computational indistinguishability from Haar-random isometries), correction of almost-linearly many local errors, and a constant encoding rate can be achieved simultaneously.
Correcting more errors generally requires more redundancy, yet a constant rate limits the number of physical qubits to $O(n_\logic)$. Within this constraint, efficient decoding must exploit structure in the code. Strong pseudorandomness creates the opposing demand that the encoding hide this structure from adversaries, making the simultaneous guarantee far from immediate. 

Among standard quantum pseudorandomness notions, PRICs achieve the strongest pseudorandomness that still allows correction of local quantum noise.
A natural choice is to require the encoding operation to be computationally indistinguishable from Haar-random unitaries, as in PRUs.
However, such encodings have equally sized input and output spaces, leaving no redundancy for error correction.
Consequently, there is physically meaningful one-local quantum noise that they cannot correct.\ifnum\extendedabstract=0\footnote{This impossibility is formalized in \cref{thm:impossibility_w_small_redandancy}.}
\fi
 Correcting this noise requires more physical qubits than logical qubits, which is permitted by isometries.
In addition, the pseudorandomness of PRICs is the same as that of PRIs and stronger than that of PRSs, because it allows adversaries to adaptively query the encoding operation on arbitrary quantum inputs.

For the completely depolarizing channel, we also construct public-key QPRCs under the same LPN assumption.
The completely depolarizing channel discards its input and outputs the maximally mixed state, and is equivalent to the channel applying a fresh Haar-random isometry for each execution.
In the public-key setting, $\Enc$ takes a public key $\pk$ in place of $\sk$, while $\Dec$ still takes the secret key $\sk$, and pseudorandomness must hold even when the adversary is given the public key.

\ifnum\extendedabstract=0
\begin{theorem}[\cref{coro:QPRC}, informal]\label{thm:QPRC_intro}
    Let $\cD$ be the distribution that outputs the completely depolarizing channel with probability one.
    Assume that LPN is hard for $2^{O(\sqrt{n})}$-time quantum algorithms.
    Then, for some constant $\alpha>0$, there exist public-key QPRCs that are pseudorandom with respect to $\cD$ and robust to all $\alpha n_\phys$-local quantum noise.
    In addition, these QPRCs use $n_\phys=O(n_\logic)$ physical qubits for $n_\logic$ logical qubits.
\end{theorem}
\else
\begin{theorem}\label{thm:QPRC_intro}
    Let $\cD$ be the distribution that outputs the completely depolarizing channel with probability one.
    Assume that LPN is hard for $2^{O(\sqrt{n})}$-time quantum algorithms.
    Then, for some constant $\alpha>0$, there exist public-key QPRCs that are pseudorandom with respect to $\cD$ and robust to all $\alpha n_\phys$-local quantum noise.
    In addition, these QPRCs use $n_\phys=O(n_\logic)$ physical qubits for $n_\logic$ logical qubits.
\end{theorem}
\fi

These QPRCs, whose encodings are indistinguishable from the completely depolarizing channel, are direct quantum analogues of classical PRCs and support public-key encoding while correcting a constant fraction of local quantum noise at a constant encoding rate.
The ideal oracle for classical PRCs ignores its input and returns a fresh uniformly random string on each query.
Likewise, measuring the output of the completely depolarizing channel in the computational basis gives a fresh uniformly random string, independently of the input.
This also permits public-key encoding, although PRICs cannot have public-key encoders even without noise.\ifnum\extendedabstract=0\footnote{This impossibility is formalized in \cref{prop:no_public_key_haar_qprc}.}
\else
\footnote{We prove this impossibility in the full paper.}
\fi
Our QPRCs in \cref{thm:QPRC_intro} correct errors on a constant fraction of the physical qubits, whereas the guarantee in \cref{thm:PRIC_intro} covers almost-linearly many errors under the stronger pseudorandomness.

\ifnum\extendedabstract=0
\subsection{Technical Contribution}
\else
\paragraph{Technical contribution.}
\fi

\ifnum\extendedabstract=0
Our construction of PRICs relies on two new ingredients, each of independent interest.
The first is a new classical cryptographic primitive, which we call a pseudorandom functional error-correcting code (PRFC).
PRFCs are pseudorandom functions (PRFs) equipped with a polynomial-time decoder that uses a key $k$ to recover $x$ from $f_k(x)$ even after a constant fraction of the bits of $f_k(x)$ are flipped, and we construct them under the LPN assumption.
\fi

\ifnum\extendedabstract=1
Our construction of PRICs relies on two new ingredients, each of independent interest.
The first is a new classical cryptographic primitive, which we call a pseudorandom functional error-correcting code (PRFC), and the second is a general construction of keyed QECCs from any classical code, including any nonlinear code, as a black box.
PRFCs enable us to construct PRICs that correct local Pauli \(X\) errors, but the PRICs cannot correct other Pauli errors, including phase errors.
The method behind our second result then extends this construction to correct general local quantum noise while preserving its indistinguishability from Haar-random isometries.

We first explain PRFCs.
PRFCs are pseudorandom functions (PRFs) equipped with a polynomial-time decoder that uses a key $k$ to recover $x$ from $f_k(x)$ even after a constant fraction of the bits of $f_k(x)$ are flipped.
We construct them under the LPN assumption.
\fi

\ifnum\extendedabstract=0
\begin{theorem}[\cref{coro:LPN_imply_PRFC}, informal]\label{thm:PRFC_intro}
    If LPN is hard for $2^{O(\sqrt{n})}$-time quantum algorithms, then PRFCs exist.
    In addition, these PRFCs encode $n$ bits to $m=O(n)$ bits. 
\end{theorem}
\else
\begin{theorem}\label{thm:PRFC_intro}
    If LPN is hard for $2^{O(\sqrt{n})}$-time quantum algorithms, then PRFCs exist.
\end{theorem}
\fi

PRFCs satisfy the pseudorandomness of PRFs, rather than that of PRCs, while retaining error correction.
PRCs require that oracle access to their encoding operation be computationally indistinguishable from access to an operation that outputs a fresh uniformly random string on each query.
In contrast, PRFCs require that oracle access to $f_k$ be computationally indistinguishable from access to a single uniformly random function.
This form of pseudorandomness is standard in cryptography, and recent work~\cite{PNR26} has enabled new applications by combining error correction with a form of pseudorandomness different from that of PRCs.
We therefore hope that PRFCs will enable new applications.

Our second technical result gives a black-box transformation from classical codes, including nonlinear ones, into keyed QECCs with efficient decoding for errors on almost linearly many physical qubits.

\ifnum\extendedabstract=0
\begin{theorem}[\cref{thm:better-keyed-cws}, informal]\label{thm:Graph_for_CWS_intro}
    For every $t=o(n_{\phys}\log\log n_{\phys}/\log n_{\phys})$, there exists a pair of QPT algorithms $(\Enc^{(\cdot)},\Dec^{(\cdot)})$ such that $\Exp_{\sk}[\Dec^C(\sk,\cdot)\circ\cE\circ\Enc^C(\sk,\cdot)]$ is negligibly close to the identity map in the diamond distance for any $t$-local error $\cE$ if a classical code $C$ corrects a constant fraction of bit errors.
\end{theorem}
\else
\begin{theorem}\label{thm:Graph_for_CWS_intro}
    For every $t=o(n_{\phys}\log\log n_{\phys}/\log n_{\phys})$, there exists a pair of QPT algorithms $(\Enc^{(\cdot)},\Dec^{(\cdot)})$ such that $\Exp_{\sk}[\Dec^C(\sk,\cdot)\circ\cE\circ\Enc^C(\sk,\cdot)]$ is negligibly close to the identity map in the diamond distance for any $t$-local error $\cE$ if a classical code $C$ corrects a constant fraction of bit errors.
\end{theorem}
\fi

We realize this transformation within the codeword-stabilized (CWS) framework~\cite{CSSZ08}.
A standard approach to constructing QECCs from classical codes uses stabilizer codes, but it requires the underlying classical code to be linear.
CWS provides a way to construct QECCs from a classical code with a graph, and includes stabilizer codes as a special case.
Unlike the stabilizer-code case, CWS also permits nonlinear classical codes, so our transformation can use such codes.

\cref{thm:Graph_for_CWS_intro} resolves the efficient-decoding problem for CWS codes raised by~\cite{LDGP10} in the keyed setting.
The original CWS construction yields an efficient quantum encoder from an efficient classical encoder~\cite{CSSZ08}, but efficient quantum decoding was not known to follow from efficient classical decoding.
Even the best known general decoder for CWS codes~\cite{LDGP10} takes superpolynomial time when correcting errors on $\omega(\log n_\phys)$ qubits.
In our keyed construction, decoding remains efficient for errors on almost linearly many physical qubits.

\ifnum\extendedabstract=0
\subsection{Discussion and Future Work}
We compare classical PRCs and PRFCs with QPRCs for the completely depolarizing channel and PRICs, respectively.
PRCs and PRFCs encode bit strings to be robust to classical noise, whereas QPRCs and PRICs encode quantum states to be robust to quantum noise.
Another point of comparison is the operation from which each encoder must be computationally indistinguishable.
For PRCs, the ideal operation ignores its input and returns a fresh uniformly random bit string on each execution; for QPRCs with complete depolarization, it discards its input and outputs the maximally mixed state.
In contrast, PRFCs and PRICs are compared with a uniformly random function and a Haar-random isometry, respectively, each sampled once and held fixed across executions.
\Cref{tab:ideal-oracles} summarizes these correspondences.

\begin{table}[h]
\centering
\renewcommand{\arraystretch}{1.3}
\begin{tabular*}{0.75\linewidth}{@{\extracolsep{\fill}}lcc@{}}
\hline
Encoded input & Uniform output & Fixed random map \\
\hline
Classical bit strings & PRC~\cite{C:ChrGun24} & PRFC (\cref{thm:PRFC_intro}) \\
Quantum states & QPRC (\cref{thm:QPRC_intro}) & PRIC (\cref{thm:PRIC_intro}) \\
\hline
\end{tabular*}
\caption{What the four types of codes encode and the operations from which their encoders must be computationally indistinguishable. QPRCs in the table are pseudorandom with respect to the completely depolarizing channel.}
\label{tab:ideal-oracles}
\end{table}

We next compare the output lengths and achieved robustness of these four constructions, focusing on bit-flip noise for the classical codes and local quantum noise for the quantum codes.
For the classical codes, $n$ and $m$ denote the input and output bit lengths.
For the quantum codes, $n_\logic$ and $n_\phys$ denote the logical and physical qubit counts.
Recoverability requires $m\ge n$ and $n_\phys\ge n_\logic$, so the linear output lengths are optimal up to constant factors.
Noise that flips a constant fraction of the classical output bits has weight $\Theta(m)$, while the corresponding quantum noise has weight $\Theta(n_\phys)$.
Since noise weight cannot exceed the output length, the linear robustness achieved by PRCs, PRFCs, and QPRCs is optimal in asymptotic order.
\Cref{tab:code-parameters} summarizes these output-length and robustness guarantees.
It remains open whether there exist PRICs robust to $\Theta(n_\phys)$-local quantum noise, even when $n_\phys$ is allowed to grow superlinearly with $n_\logic$.

\begin{table}[h]
\centering
\renewcommand{\arraystretch}{1.3}
\begin{tabular*}{0.84\linewidth}{@{\extracolsep{\fill}}lcc@{}}
\hline
Construction & Output length & Correctable noise \\
\hline
PRC~\cite{C:ChrGun24} & $m=\Theta(n)$ & Noise flipping $\Theta(m)$ bits \\
PRFC (\cref{thm:PRFC_intro}) & $m=\Theta(n)$ & Noise flipping $\Theta(m)$ bits \\
QPRC (\cref{thm:QPRC_intro}) & $n_\phys=\Theta(n_\logic)$ & $\Theta(n_\phys)$-local quantum noise \\
PRIC (\cref{thm:PRIC_intro}) & $n_\phys=\Theta(n_\logic)$ & $o(n_\phys\log\log n_\phys/\log n_\phys)$-local quantum noise \\
\hline
\end{tabular*}
\caption{Achievable output lengths and noise-robustness guarantees. The classical entries concern noise flipping a constant fraction of output bits, and the quantum entries concern local quantum noise. QPRCs in the table are pseudorandom with respect to the completely depolarizing channel.}
\label{tab:code-parameters}
\end{table}

One may hope to extend classical watermarking and steganography applications to the quantum setting, but this is nontrivial due to the difference between classical and quantum randomness.
The watermarking approach~\cite{C:ChrGun24} replaces part of the internal randomness used by a probabilistic algorithm in a language model with PRC codewords.
The randomness intrinsic to quantum computation arises from measurement, whose outcomes cannot in general be freely chosen in the same way as classical random bits.
In addition, the robust steganography construction in~\cite{C:ChrGun24} uses rejection sampling to embed PRC codewords into samples from a cover distribution.
Here, rejection sampling draws samples from the cover distribution and keeps those that satisfy a test determined by the bit to be embedded.
For quantum states, implementing an analogous test involves measurement, and conditioning on acceptance needs postselection in general.
An analysis specific to the quantum setting is therefore needed to establish both efficient implementation and indistinguishability from the cover states.
We leave open even the question of whether corresponding quantum applications are possible.

Quantum analogues of the coding applications also require further work on robustness guarantees beyond those established in this paper.
A central feature of the capacity-achieving codes of Lu, Silbak, and Wichs~\cite{LSW26} is that their robustness holds even after the adversary has obtained encodings of chosen messages under the same key.
In the classical secret-key adaptive-robustness game of~\cite{STOC:AACDG25}, the challenger records each queried message and its returned codeword before the adversary submits a possibly corrupted challenge codeword.
Adaptive robustness requires that no efficient adversary can, except with negligible probability, submit a challenge codeword that is within the allowed Hamming distance of a recorded codeword but does not decode to the corresponding message.
Even formulating a quantum version of this game is nontrivial because no-cloning prevents even a computationally unbounded challenger from keeping copies of arbitrary unknown inputs and outputs without disturbing the adversary's registers.
We leave the formulation and construction of adaptively robust quantum codes for future work.

Beyond extending classical applications, we ask whether QPRCs and PRICs can enable applications specific to quantum settings.
One possible direction is to model quantum information encoding in chaotic systems, including black holes.
Hayden and Preskill~\cite{HaydenPreskill07} modeled black-hole dynamics by Haar-random unitaries and studied the error-correcting properties of the resulting encodings.
Two related lines of work are particularly relevant here.
One studies the error-correcting properties of encodings constructed from random unitaries and develops decoding methods for them~\cite{arXiv:YosKit17,PRX:YosYao19,PRR:NBFG25}.
The other studies the relationship between quantum pseudorandomness and chaotic physical systems~\cite{JHEP:KimTanPre20,EngelhardtFolkestadLevineVerheijdenYang25,SMLBH25}.
Since PRICs combine quantum pseudorandomness with efficient recovery, we hope that they can serve as models of chaotic systems that capture both properties, and we leave the investigation of this possibility for future work.

\subsection{Organization}

The rest of the paper develops the definitions and constructions introduced above.
In \cref{sec:overview}, we give a technical overview of the definitions and constructions, followed by a comparison with related work.
In \cref{sec:preliminaries}, we collect the notation, noise models, quantum error-correcting codes, CWS codes, and cryptographic primitives used later.
In \cref{sec:QPRC_def}, we define QPRCs including PRICs.
In \cref{sec:QPRC_const}, we construct QPRCs for a completely depolarizing channel.
In \cref{sec:PRFC}, we define PRFCs and construct them.
In \cref{sec:PRIC_const}, we construct PRICs.
Then, \cref{sec:better_graph} gives the graph sampler, the efficient phase-recovery algorithm, and the keyed CWS construction from classical codes used to explain the independent coding-theoretic result above.

\fi

\section{Technical Overview}
\label{sec:overview}

In this section, we give a high-level overview of our constructions.
We first construct public-key QPRCs for the completely depolarizing ideal in \cref{subsec:overview-depolarizing}.
We then introduce PRFCs and use them to construct PRICs that correct only local Pauli $X$ errors in \cref{subsec:overview-bit-errors}.
Finally, we use CWS codes to lift this construction to PRICs that correct general local quantum noise in \cref{subsec:overview-general-errors}.

Before these constructions, we recall the definition of post-quantum PRCs.
A PRC is a tuple of classical polynomial-time algorithms $(\KeyGen,\Enc,\Dec)$ that encodes $n$-bit messages as $m$-bit strings.\footnote{The formal definition also includes soundness, which we omit from this overview because the constructions below rely on robustness and pseudorandomness. For the formal definition, see \cref{def:PKPRC,def:SKPRC}.}
In the public-key setting, $\KeyGen(1^\secp)$ outputs $(\sk,\pk)$, $\Enc(\pk,x)$ outputs an $m$-bit codeword $c$, and $\Dec(\sk,c)$ returns an $n$-bit message or $\bot$.
In the secret-key setting, both $\Enc$ and $\Dec$ use $\sk$.
Robustness requires recovery of each fixed message after the specified classical noise.
Throughout this overview, we consider classical noise that flips a constant fraction of the $m$ bits.
Pseudorandomness requires that, for any QPT adversary $\cA$ making only classical queries,
\begin{align}
\Bigl|\Pr_{(\sk,\pk)\gets\KeyGen(1^\secp)}[1\gets\cA^{\Enc(\pk,(\cdot))}(1^\secp,\pk)]
-\Pr_{(\sk,\pk)\gets\KeyGen(1^\secp)}[1\gets\cA^{\cU}(1^\secp,\pk)]\Bigr|
\le\negl(\secp),
\end{align}
where $\Enc(\pk,(\cdot))$ uses its internal randomness for each query, and $\cU$ is an oracle that ignores its input and returns a fresh uniform $m$-bit string on every query.
For secret-key PRCs, robustness uses $\Enc(\sk,x)$ in place of $\Enc(\pk,x)$, and the same pseudorandomness is required against any QPT adversaries without $\sk$ and making only classical queries to $\Enc(\sk,(\cdot))$ or $\cU$.

All our constructions use post-quantum PRCs as a black box, which exist under the LPN assumption against $2^{O(\sqrt{n})}$-time quantum adversaries.\footnote{For post-quantum PRCs, there is another assumption that combines the hardness of planted XOR and LPN against QPT adversaries. For details, see \cref{assumption:for_PRC,sec:post_quantum_CG24}.}
\cite{C:ChrGun24} gives public-key and secret-key PRCs under the hardness of LPN against subexponential-time classical adversaries.
Under the corresponding quantum hardness of LPN, the same construction and proof yield post-quantum public-key and secret-key PRCs because we consider QPT adversaries making only classical oracle queries.
We give the details in \cref{sec:post_quantum_CG24} for completeness.
In what follows, we omit the qualifier ``post-quantum'' when referring to these PRCs.

\subsection{QPRCs for the Completely Depolarizing Channel}
\label{subsec:overview-depolarizing}
We first construct public-key QPRCs whose encodings are pseudorandom with respect to the completely depolarizing channel.
Here, the completely depolarizing channel outputs the $n_\phys$-qubit maximally mixed state for any $n_\logic$-qubit input state.

The construction combines QECCs and PRCs by using a quantum one-time pad.
This is similar to the standard approach of encrypting a quantum state~\cite{C:BroJef15}.
Fix an error bound $t$.
Let $(\QECC.\Enc,\QECC.\Dec)$ be a QECC that corrects $t$-local quantum noise, and let $(\PRC.\KeyGen,\PRC.\Enc,\PRC.\Dec)$ be a public-key PRC that corrects up to $t$ bit errors; the QPRC uses the PRC key pair.
The QPRC encoder samples fresh Pauli keys $x,z$, computes $c_{x\|z}\gets\PRC.\Enc(\pk_\PRC,x\|z)$, and applies the following map:
\begin{align}
\ket{\psi}_\regA\longmapsto (X^xZ^z)_\regB\QECC.\Enc_{\regA\to\regB}\ket{\psi}_\regA\otimes\ket{c_{x\|z}}_\regC.
\label{eq:QPRC-encoding-overview}
\end{align}
Here $x\|z$ denotes concatenation.
The security of PRCs, together with the argument of~\cite{C:BroJef15}, shows that this channel is indistinguishable from the completely depolarizing channel.

The QPRC decoder measures $\regC$, recovers $x\|z$ with the PRC decoder, removes the Pauli pad from $\regB$, and applies the QECC decoder.
Robustness to local quantum noise follows from protecting both the quantum state and the Pauli key against local errors.
It suffices to establish robustness against local Pauli errors, since this implies robustness against general local quantum noise.
Let $P_{\regB\regC}=Q_\regB\otimes Q'_\regC$ be a Pauli error of weight at most $t$ acting on the encoded registers $\regB\regC$.
After measuring $\regC$ in the computational basis, the phase component of $Q'$ has no effect on the outcome, while its bit-flip component changes at most $t$ bits.
Since the PRC corrects up to $t$ bit errors, its decoder recovers $x\|z$.
Removing the one-time pad $X^xZ^z$ leaves a $t$-local error on $\regB$, which the QECC corrects.
Thus the key-averaged recovery channel is negligibly close to the identity.

\subsection{Warm-Up: Construction of PRICs Robust to Local Pauli $X$ Errors via PRFCs}
\label{subsec:overview-bit-errors}
Our next goal is to construct PRICs from PRCs.
The pseudorandomness of PRCs, however, is not suitable on its own for PRICs.
For each fixed secret key, a PRC encoder uses fresh internal randomness on each execution, whereas a PRIC encoder is a fixed isometry.

To address this mismatch, we introduce PRFCs.
Informally, PRFCs are PRFs that are robust to noise.\footnote{Since PRFCs are classical primitives, one could also consider classically secure PRFCs.
However, for our purposes, we need post-quantum security, and therefore we consider only post-quantum secure PRFCs.
Throughout this paper, we omit the term post-quantum.}
A PRFC consists of polynomial-time classical algorithms $(\KeyGen,\Enc,\Dec)$, where $\KeyGen$ outputs a secret key $\sk$, $\Enc(\sk,\cdot)$ is deterministic and maps $x$ to an $m$-bit codeword, and $\Dec(\sk,\cdot)$ decodes noisy codewords.
As for PRCs, robustness requires recovery of each fixed input after the specified noise.
Pseudorandomness requires that every QPT adversary, even with coherent oracle queries, cannot distinguish oracle access to $\Enc(\sk,\cdot)$ from oracle access to a uniformly random function $f:\bit^{n(\secp)}\to\bit^{m(\secp)}$.

\paragraph{PRFCs from secret-key PRCs.}
We first construct PRFCs from secret-key PRCs.
Let $(\PRC.\KeyGen,\allowbreak\PRC.\Enc,\allowbreak\PRC.\Dec)$ be a secret-key PRC.
Let $\ell_{\PRC}(\secp)$ be the number of internal random bits used by $\PRC.\Enc$.
Let $\{\pi_k\}_{k\in\bit^\secp}$ be a PRP over $\bit^{n}$.
Let $\{f_{k'}\}_{k'\in\bit^\secp}$ be a PRF from $\bit^{n}$ to $\bit^{\ell_{\PRC}}$.
The PRFC secret key consists of a PRC secret key $\sk_\PRC$, a PRP key $k$, a PRF key $k'$, and a uniformly random string $r\in\bit^{\ell_\PRC}$ of the same length as the internal randomness of $\PRC.\Enc$.
On input $x$, the PRFC encoder is defined by
\begin{align}
    \Enc(\sk,x)\coloneqq
    \PRC.\Enc(\sk_\PRC,\pi_k(x);f_{k'}(x)\oplus r).
\end{align}
Here, the notation $\PRC.\Enc(\sk_\PRC,y;s)$ is the output of $\PRC.\Enc(\sk_\PRC,y)$ when using $s$ as its internal randomness.
Robustness follows from that of the underlying PRC.

Let us explain why this construction satisfies pseudorandomness.
First, replace the PRP $\pi_k$ and the PRF $f_{k'}$ by a uniformly random permutation $\pi$ and a uniformly random function $f$.
Then the encoder behaves like
\begin{align}
    x
    \longmapsto
    \PRC.\Enc(\sk_\PRC,\pi(x);f(x)\oplus r).
\end{align}
For each $x$, $\pi(x)$ and $f(x)\oplus r$ are independent uniformly random bit strings.
Thus, pseudorandomness of the underlying PRC shows that this oracle is indistinguishable from a uniformly random function.

\paragraph{How to use PRFCs for constructing PRICs.}
Next we construct PRICs robust only to local Pauli $X$ errors with PRFCs.
For our construction, we use a PRFC, a PRF, and a PRP; the latter two can be constructed from the PRFC because its encoder is a PRF and PRFs imply PRPs~\cite{Quantum:Zhandry25}.
Let $n_\logic$ be the number of input qubits, and $\ell,r,$ and $m$ be polynomials.
The number of output qubits is $n_\phys\coloneqq n_\logic+\ell+r+m$.
We construct a PRIC $(\KeyGen^X,\Enc^X,\Dec^X)$ as follows.
Let $\PRFC$ be a PRFC whose encoder, for each secret key $\sk'$, maps $\bit^{n_\logic+\ell}$ to $\bit^m$.
Define $g_{\sk'}(\cdot)\coloneqq \PRFC.\Enc(\sk',\cdot)$ for each PRFC secret key $\sk'$.
Let $\{f_k\}_{k\in\bit^\secp}$ be a PRF from $\bit^{n_\logic+\ell}$ to $\bit$.
Let $\{\pi_{k'}\}_{k'\in\bit^\secp}$ be a PRP over $\bit^{n_\logic+\ell+r}$.

The algorithm $\KeyGen^X$ samples a secret key as follows: on input $1^\secp$, it samples $\sk'\gets \PRFC.\KeyGen(1^\secp)$, $k\gets\bit^\secp$, and $k'\gets\bit^\secp$, and outputs $\sk\coloneqq(\sk',k,k')$.
For each $\sk$, $\Enc^X(\sk,(\cdot))$ implements the isometry 
\begin{align}
    \Enc^X_\sk
    :
    \ket{x}_\regA
    \mapsto
    \frac{1}{\sqrt{2^\ell}}
    \sum_{y\in\bit^\ell}
    (-1)^{f_k(x\|y)}
    \ket{\pi_{k'}(x\|y\|0^r)}_\regB
    \ket{g_{\sk'}(x\|y)}_\regC .
\end{align}
Here, $\regB$ stores the permuted padded string, while $\regC$ stores the PRFC encoding of $x\|y$.

Before constructing the decoder, we explain why this satisfies pseudorandomness by the hybrid argument.
In the first hybrid, we replace pseudorandom objects with ideal random objects.
The resulting isometry is
\begin{align}
    V_1:\ket{x}
    \mapsto
    \frac{1}{\sqrt{2^\ell}}
    \sum_{y\in\bit^\ell}
    (-1)^{f(x\|y)}
    \ket{\pi(x\|y\|0^r)}
    \ket{g(x\|y)},
    \label{OveriewEq:H1}
\end{align}
where $f$, $g$, and $\pi$ are the corresponding ideal random objects.
In the next hybrid, we replace the concatenated function $z\mapsto \pi(z\|0^r)\| g(z)$ with $z\mapsto \pi'(z\|0^{r+m})$, where $\pi'$ is a uniformly random permutation over $n_\phys=n_\logic+\ell+r+m$ bits.
The resulting isometry is
\begin{align}
    V_2:\ket{x}\mapsto
    \frac{1}{\sqrt{2^\ell}}
    \sum_{y\in\bit^\ell}
    (-1)^{f(x\|y)}\ket{\pi'(x\|y\|0^{r+m})}_{\regB\regC}.
    \label{OveriewEq:H2}
\end{align}
The lemma below with $n=n_\logic+\ell$ shows that the output maps in the two hybrids are statistically indistinguishable when $r\ge n+\omega(\log\secp)$.
Because the proof of the lemma is straightforward, we omit it here.

\begin{lemma}[\cref{lem:pi||f_vs_pi'}, informal]
    Define $\cD$ and $\cD'$ to be the following distributions over injective functions $\bit^n\to\bit^{n+r+m}$.
    \begin{itemize}
        \item $\cD:$
        Choose a uniformly random function $g:\bit^n\to\bit^m$ and a uniformly random permutation $\pi:\bit^{n+r}\to\bit^{n+r}$.
        Define $h:\bit^n\to\bit^{n+r+m}$ by $h(z)\coloneqq \pi(z\|0^r)\|g(z)$.
        Output $h$.

        \item $\cD':$
        Choose a uniformly random permutation $\pi':\bit^{n+r+m}\to\bit^{n+r+m}$.
        Define $h:\bit^n\to\bit^{n+r+m}$ by $h(z)\coloneqq \pi'(z\|0^{r+m})$.
        Output $h$.
    \end{itemize}
    If $r\ge n+\omega(\log\secp)$, then the statistical distance between $\cD$ and $\cD'$ is at most $\negl(\secp)$.
\end{lemma}
Finally, applying the pseudorandomness of PF isometries in~\cite{FOCS:MPSY24}, when $\ell\ge\omega(\log\secp)$, the isometry in \cref{OveriewEq:H2} is indistinguishable from a Haar-random isometry by any QPT adversary.
Therefore, by putting everything together, $\Enc^X_\sk$ is computationally indistinguishable from a Haar-random isometry.

Next we construct $\Dec^X$ and show robustness to local Pauli $X$ errors.
Let $\ket{\psi}=\sum_x\alpha_x\ket{x}$ be an input state, and suppose that a Pauli $X$ error $X^v$ acts on $\Enc_\sk^X\ket{\psi}$.
Write $v=v_0\|v_1$ according to the registers $\regB$ and $\regC$.
To describe the decoding step simply, first suppose that the PRFC decoder recovers $x\|y$ for all $x$ and $y$.\footnote{The actual proof also accounts for components of the superposition on which PRFC decoding fails.}
The decoder applies the isometry $\ket{c}_\regC\mapsto\ket{c}_\regC\ket{\PRFC.\Dec(\sk',c)}_\regD$, yielding
\begin{align}
    X^v\Enc_\sk^X\ket{\psi}
    &\propto\sum_{x,y}\alpha_x
    (-1)^{f_k(x\|y)}
    \ket{\pi_{k'}(x\|y\|0^r)\oplus v_0}_\regB
    \ket{g_{\sk'}(x\|y)\oplus v_1}_\regC
    \notag\\
    &\mapsto\sum_{x,y}\alpha_x
    (-1)^{f_k(x\|y)}
    \ket{\pi_{k'}(x\|y\|0^r)\oplus v_0}_\regB
    \ket{g_{\sk'}(x\|y)\oplus v_1}_\regC
    \ket{x\|y}_\regD.\label{OverviewEq:PRFCdecode}
\end{align}
Write $\regD=\regA\regE$, where $\regA$ contains the first $n_\logic$ qubits and $\regE$ contains the remaining $\ell$ qubits.
Next, controlled on $\regD$, $\Dec^X_\sk$ uncomputes $\pi_{k'}(x\|y\|0^r)$ and $g_{\sk'}(x\|y)$ from $\regB\regC$ and cancels the phase $(-1)^{f_k(x\|y)}$:
\begin{align}
    \big(\text{\cref{OverviewEq:PRFCdecode}}\big)
    &\mapsto\sum_{x,y}\alpha_x
    \ket{v_0}_\regB
    \ket{v_1}_\regC
    \ket{x\|y}_\regD
    =\ket{v}_{\regB\regC}\otimes
    \left(\sum_x\alpha_x\ket{x}_{\regA}\right)
    \otimes
    \left(\sum_{y\in\bit^\ell}\ket{y}_{\regE}\right).
\end{align}
The decoder outputs $\regA$, whose state is $\ket{\psi}$ in this ideal case.

\subsection{Lifting Robustness to General Local Errors}
\label{subsec:overview-general-errors}

The construction in the previous subsection is robust to local Pauli $X$ errors, but not yet to general Pauli errors.
For a Pauli error $Z^vX^u$ on the registers $\regB\regC$, the $X$-part can be corrected by the PRFC decoder, whereas the $Z$-part introduces the phase $(-1)^{v\cdot\gamma_\sk(x\|y)}$, where $\gamma_\sk(x\|y)=\pi_{k'}(x\|y\|0^r)\|g_{\sk'}(x\|y)$.
After applying $\Dec^X$ without knowing $v$, the state becomes, up to a global phase,
\begin{align}
    (Z^vX^u\cdot\Enc_\sk^X\ket{\psi})_{\regB\regC}
    \mapsto
    \ket{u}_{\regB\regC}
    \otimes
    \left(
        \frac{1}{\sqrt{2^\ell}}
        \sum_{x,y}
        \alpha_x
        (-1)^{v\cdot\gamma_\sk(x\|y)}
        \ket{x}_\regA
        \ket{y}_\regE
    \right),
\end{align}
so the original logical state is not recovered in general.

\paragraph{Codeword-stabilized framework.}
To make the phase recoverable, we add an encoding layer based on the CWS framework~\cite{CSSZ08}.
The CWS framework first maps computational-basis states to codewords of a classical code and then applies a graph-basis unitary.
Let $G$ be a simple graph on $n_\phys$ vertices, and let $A\in\bit^{n_\phys\times n_\phys}$ be its adjacency matrix, where $A_{i,j}=1$ if vertices $i$ and $j$ are adjacent and $A_{i,j}=0$ otherwise.
Define $U_G\coloneqq\left(\prod_{\{i,j\}\in E}\CZ_{ij}\right)H^{\otimes n_\phys}$.
Using a classical code with codeword $c_x$ and $U_G$, the encoding maps
\begin{align}
    \sum_x\alpha_x\ket{x}
    \mapsto
    \sum_x\alpha_x\ket{c_x}
    \mapsto
    U_G\bigg(\sum_x\alpha_x\ket{c_x}\bigg).
\end{align}
For a Pauli error $Z^vX^u$, the CWS identity gives, up to a global phase,
\begin{align}
    U_G^\dagger Z^vX^u U_G
    \propto X^{v\oplus Au}Z^u.
\end{align}
Thus, the Pauli error induces a bit error $e=v\oplus Au$, which depends on both $u$ and $v$.

However, no known graph family comes with an efficient algorithm that recovers $u$ from the induced error $e=v\oplus Au$ in the error range we need.
We therefore construct a graph sampler together with a recovery algorithm that recovers $u$ from $e$ with overwhelming probability over the sampled graph for each fixed low-weight error.
\begin{theorem}[\cref{thm:GraphSample_and_Recover}, informal]\label{OverviewThm:GraphSample_and_Recover}
    Fix a constant $p>0$ and a function $t:\N\to\N$ satisfying $t(n_\phys)=o(n_\phys\log\log n_\phys/\log n_\phys)$.
    There are a PPT graph-sampling algorithm $\GraphSample$ and a deterministic polynomial-time algorithm $\Recover$ such that, for every sufficiently large even $n_\phys$ and every fixed $u,v\in\bit^{n_\phys}$ with $|\Support(u)\cup\Support(v)|\le t(n_\phys)$, every graph $G$ output by $\GraphSample(1^{n_\phys/2})$ has $n_\phys$ vertices and satisfies $\weight(v\oplus Au)\le p n_\phys$, and
    \begin{align*}
        \Pr_{G\gets\GraphSample(1^{n_\phys/2})}[\Recover(G,v\oplus Au)=u]\ge 1-\negl(n_\phys),
    \end{align*}
    where $A$ is the adjacency matrix of $G$.
\end{theorem}

\paragraph{Construction of PRICs robust to general local errors.}
Before explaining the proof of \cref{OverviewThm:GraphSample_and_Recover}, we use its algorithms to lift the PRIC $(\KeyGen^X,\Enc^X,\Dec^X)$ to a PRIC $(\KeyGen,\Enc,\Dec)$ robust to general local noise.
On input $1^\secp$, $\KeyGen$ samples $\sk_{\mathrm{in}}\gets\KeyGen^X(1^\secp)$, an independent uniformly random $n_\phys$-qubit Pauli $P$, and $G\gets\GraphSample(1^{n_\phys/2})$, and outputs $\sk\coloneqq(\sk_{\mathrm{in}},G,P)$.
The encoder implements $\Enc_\sk\coloneqq PU_G\cdot\Enc^X_{\sk_{\mathrm{in}}}$.
Because $G$ and $P$ are independent of $\sk_{\mathrm{in}}$, pseudorandomness follows from that of $\Enc^X$ and the left invariance of the Haar measure.

To decode, $\Dec_\sk$ first applies $P^\dagger$, undoing the random Pauli used by the encoder.
Averaging over the independent random $P$ twirls any key-independent $w$-local noise into a mixture of fixed $w$-local Pauli errors.
For $w=o(n_\phys\log\log n_\phys/\log n_\phys)$, it therefore suffices to analyze a fixed Pauli error $Z^vX^u$ of weight at most $w$.
The decoder next applies $U_G^\dagger$.
For this fixed error, the conjugation produces, up to a global phase, the bit error $X^e$ and phase error $Z^u$, where $e=v\oplus Au$.
Following the decoding procedure beginning with \cref{OverviewEq:PRFCdecode}, we obtain
\begin{align}
    \ket{e}_{\regB\regC}\otimes\left(\frac{1}{\sqrt{2^\ell}}\sum_{x,y}\alpha_x(-1)^{u\cdot\gamma_{\sk_{\mathrm{in}}}(x\|y)}\ket{x}_{\regA}\ket{y}_{\regE}\right).
\end{align}
The decoder measures $\regB\regC$ to obtain $e$, runs $\Recover(G,e)$ to recover $u$, and, when recovery succeeds, uses $u$ and $\sk_{\mathrm{in}}$ to cancel the phase.
The resulting state is
\begin{align}
    \left(\sum_x\alpha_x\ket{x}_{\regA}\right)\otimes\left(\frac{1}{\sqrt{2^\ell}}\sum_y\ket{y}_{\regE}\right).
\end{align}
For every such fixed error, the parameter choice puts $e$ within the PRFC correction radius for every sampled graph, while \cref{OverviewThm:GraphSample_and_Recover} recovers $u$ with overwhelming probability over $G$.
The decoder therefore outputs the original state from $\regA$, establishing robustness against the original $w$-local noise.

\paragraph{Efficient phase recovery.}
It remains to explain why \cref{OverviewThm:GraphSample_and_Recover} holds by constructing $\GraphSample$ and $\Recover$.
To recover $u$ from $e=v\oplus Au$, $\GraphSample$ samples a graph $G$ with the following properties:
\begin{itemize}
    \item Unique-neighbor condition: Errors on a small set of vertices remain visible at many neighboring vertices.
    \item Local-overlap condition: The unknown $v$ corrupts only a few neighbors of any vertex.
\end{itemize}

We now construct $\Recover$ to obtain $u$ from $e$.
Starting from $\widehat u=0$, it forms the residual $r=e\oplus A\widehat u$ and flips a vertex when sufficiently many of its neighbors have a $1$ in $r$.
Let $z=u\oplus\widehat u$, so $r=v\oplus Az$.
While the estimate is incorrect, the unique-neighbor condition makes some incorrectly estimated vertex visible in $Az$, and local overlap ensures that it remains visible in $r$.
When $\widehat u=u$, we have $r=v$, and local overlap ensures that no candidate remains, so $\Recover$ stops.

Finally, $\GraphSample$ combines independent random perfect matchings between two equally sized vertex sets by edge parity.
Every resulting graph satisfies the degree bound, and for each fixed $v$, it satisfies the unique-neighbor and local-overlap conditions with overwhelming probability.
Therefore, $\GraphSample$ and $\Recover$ satisfy the guarantees in \cref{OverviewThm:GraphSample_and_Recover}.

\subsection{Related Work}

\paragraph{Classical pseudorandom codes.}
Classical pseudorandom codes were introduced by Christ and Gunn~\cite{C:ChrGun24}.
Subsequent works considered stronger robustness notions, including ideal pseudorandom codes~\cite{STOC:AACDG25} and improved constructions from permuted puzzles~\cite{EPRINT:CGGMW25}.
Our QPRCs and PRICs are quantum analogues of the basic robustness notion of Christ and Gunn; quantum analogues of these stronger robustness notions are beyond the scope of this work.

\paragraph{Quantum pseudorandomness.}
Quantum pseudorandomness has been studied for states, unitaries, and isometries~\cite{C:JiLiuSon18,EC:AGKL24,FOCS:MPSY24,STOC:MaHua25}.
QPRCs for the completely depolarizing channel use a different comparison distribution from these primitives.
In contrast, if one ignores the decoder, a PRIC is a PRI. 
Thus PRICs imply PRIs, and therefore imply PRSs.
PRICs and PRUs impose different requirements: PRUs are compared with Haar-random unitaries, whereas PRICs are compared with Haar-random isometries and must also correct noise.

\paragraph{Haar-random isometries and random quantum codes.}
Random quantum codes based on Haar-random subspaces or encodings have been used in random-coding and decoupling arguments for quantum error correction and entanglement transmission~\cite{HSW08,ADHW09}.
More recently, Ma, Tan, and Wright showed that Haar-random quantum codes approximately correct Pauli error sets up to the quantum Hamming bound~\cite{MTW25}.
In their notation, a Haar-random $K$-dimensional subspace of $\mathbb{C}^N$ approximately corrects a set of $m$ Pauli errors when $mK\ll N$.
Translating this to an encoding of $n_\logic$ logical qubits into $n_\phys$ physical qubits gives $K=2^{n_\logic}$ and $N=2^{n_\phys}$.
Since all Pauli errors on at most $\alpha n_\phys$ qubits form a set of size at most $2^{(H(\alpha)+\alpha\log_2 3)n_\phys+o(n_\phys)}$, Haar-random codes can approximately correct a linear number of local errors whenever $n_\logic/n_\phys+H(\alpha)+\alpha\log_2 3<1$.
These results are information-theoretic: the encoder is sampled from the Haar distribution, and the corresponding encoding and decoding procedures are not generally efficient.
In contrast, PRICs require computational indistinguishability from Haar-random isometries and robustness with efficient encoders and decoders.

\paragraph{Efficient decoding of CWS codes.}
Our PRIC construction uses the graph-state component of the CWS framework~\cite{CSSZ08}.
Li, Dumer, and Pryadko gave a clustered error-correction procedure that tests all errors on a prescribed set of $t$ qubits with one measurement~\cite{LDP09}.
When the error support is unknown, this procedure still ranges over possible supports.
Our keyed construction instead uses the observed $e=v\oplus Au$ and graph-based phase recovery to infer $u$ without enumerating supports.
The resulting efficient-decoding guarantee for CWS codes built from classical codes, including nonlinear ones, is stated in \cref{thm:Graph_for_CWS_intro}.

\section{Preliminaries}
\label{sec:preliminaries}

\subsection{Notations}

\paragraph{Basic notation.}
%check done 2026/6/18
We use standard notation from quantum computation and cryptography.
For bit strings $x$ and $y$, $(x,y)$ denotes their concatenation.
For an integer $i$, $x_i$ denotes the $i$th bit of $x$.
$\Support(x)$ denotes the set $\{i: x_i=1\}$.
$\weight(x)$ denotes the Hamming weight of $x$, i.e., $\weight(x)\coloneqq|\{i: x_i=1\}|$.
For a finite set $S$, $x\gets S$ means that $x$ is sampled uniformly from $S$.
We use $\secp$ to denote the security parameter.
For an integer $D$, $[D]$ denotes the set $\{1,2,\ldots,D\}$.
%For an integer $0\le t\le D$, we define $[D]^t_\dist\coloneqq\{(x_1,\ldots,x_t)\in[D]^t: x_i\neq x_j \text{ for all } i\neq j\}$.
We write $\negl$ for a negligible function and $\poly$ for a polynomial.
PPT and QPT stand for (classical) probabilistic polynomial time and quantum polynomial time, respectively.
For a PPT algorithm $\cA$ that takes an $n$-bit string as input and uses an $\ell$-bit string as internal randomness, $\cA(x;r)$ denotes the output of $\cA$ on input $x\in\bit^n$ with internal randomness $r\in\bit^\ell$.

\paragraph{Quantum states and norms.}
A register $\regA$ is a named finite-dimensional Hilbert space $\cH \cong \C^d$, where $d$ is called the dimension of $\regA$.
For two registers $\regA$ and $\regB$, $\regA\regB$ denotes their tensor product.
We use $I_\regA$ to denote the identity operator on $\regA$.
We often omit the subscript $\regA$ when it is clear from the context.
For a vector $\ket{\psi}$, we define its Euclidean norm by $\|\ket{\psi}\|_2\coloneqq\sqrt{\braket{\psi|\psi}}$.
For any matrix $A$ and $p\in(0,\infty)$, we define the Schatten $p$-norm by $\|A\|_p\coloneqq(\Tr[(A^\dag A)^{p/2}])^{1/p}$.
For any matrix $A$, the operator norm is defined by $\|A\|_\infty\coloneqq\max_{\ket{\psi}}\|A\ket{\psi}\|_2$, where the maximization is taken over all pure states $\ket{\psi}$.
For two states $\rho$ and $\sigma$, we define their fidelity by $F(\rho,\sigma)\coloneqq\|\sqrt{\rho}\sqrt{\sigma}\|_1$.

\paragraph{Quantum operations.}
For each register $\regA$, let $\identitymap_\regA$ denote the linear map from operators on $\regA$ to operators on $\regA$ defined by $\identitymap_\regA(X)\coloneqq X$ for every operator $X$ on $\regA$.
A linear map $\cE_{\regA\to\regB}$ from operators on $\regA$ to operators on $\regB$ is said to be completely positive if, for every register $\regR$, the map $\identitymap_\regR \otimes \cE_{\regA\to\regB}$ is positive, i.e., it maps positive operators to positive operators.
Such a map is said to be trace preserving if $\Tr[\cE_{\regA\to\regB}(\rho)] = \Tr[\rho]$ for every operator $\rho$ on $\regA$.
A quantum channel is a linear map that is both completely positive and trace preserving (CPTP).
When $\regA=\regB$, we abbreviate $\cE_{\regA\to\regA}$ as $\cE_{\regA}$.
For CPTP maps $\cE_{\regA\to\regB}$ and $\cF_{\regB\to\regC}$, we define their composition $(\cF \circ \cE)_{\regA\to\regC}$ as the CPTP map $\rho\mapsto \cF_{\regB\to\regC}(\cE_{\regA\to\regB}(\rho))$.
For an isometry $V:\regA\to\regB$, $V_{\regA\to\regB}(\cdot)$ denotes the CPTP map defined by $\rho\mapsto V\rho V^\dag$.
For two CPTP maps $\cE,\cF:\regA\to\regB$, their diamond distance is defined by $\|\cE_{\regA\to\regB}-\cF_{\regA\to\regB}\|_\diamond\coloneqq\frac{1}{2}\max_{\ket{\psi}_{\regR\regA}}\|(\identitymap_\regR\otimes\cE_{\regA\to\regB})(\ketbra{\psi}{\psi})-(\identitymap_\regR\otimes\cF_{\regA\to\regB})(\ketbra{\psi}{\psi})\|_1$, where $\regR$ is a register whose dimension is the same as that of $\regA$, and the maximization is taken over all pure states on $\regR\regA$.

For $x,z\in\bit^n$, $X^x$ and $Z^z$ denote the $n$-qubit Pauli operators $\bigotimes_{i\in [n]}X^{x_i}$ and $\bigotimes_{i\in [n]}Z^{z_i}$, respectively.
We define the $n$-qubit Pauli group by $\cP_n\coloneqq\{v\bigotimes_{i\in [n]}X^{x_i}Z^{z_i}\}_{x,z\in\bit^n,v\in\{\pm1,\pm i\}}$.
For an integer $n$, $\mu_n$ denotes the Haar measure over $n$-qubit unitaries (i.e., the unique left- and right-invariant measure).
For integers $n$ and $m$ satisfying $n<m$, $\mu_{n\to m}$ denotes the Haar measure over isometries from $n$ qubits to $m$ qubits (i.e., the unique left- and right-invariant measure).

\subsection{Lemmas}

\begin{lemma}\label{lem:gentle_measurement}
    For any pure state $\ket{\psi}$ and any projection $\Pi$, if 
    \begin{align}
        \bigg\|
         \ket{\psi}-\frac{\Pi\ket{\psi}}{\| \Pi\ket{\psi} \|_2}
        \bigg\|_2
        \le \epsilon,
    \end{align}
    then $\| \Pi\ket{\psi} \|^2_2\ge1-\epsilon^2$.
\end{lemma}

\begin{proof}
    Since $\|\ket{\psi}-\frac{\Pi\ket{\psi}}{\|\Pi\ket{\psi}\|_2}\|_2^2=2-2\|\Pi\ket{\psi}\|_2$, we have $\|\Pi\ket{\psi}\|_2 \ge 1-\frac{\epsilon^2}{2}$.
    This implies $\|\Pi\ket{\psi}\|_2^2\ge(1-\epsilon^2/2)^2\ge 1-\epsilon^2$, which concludes the proof.
\end{proof}

\begin{lemma}[Markov]\label{lem:Markov}
    Let $X$ be a positive random variable.
    Then, for any $a>0$,
    \begin{align}
        \Pr[X\ge a]\le \frac{\Exp[X]}{a}.
    \end{align}
\end{lemma}

\subsection{$t$-wise Independent Functions}

We recall the definition of $t$-wise independent function families.

\begin{definition}[$t$-wise independent function family]\label{def:t-wise_independent_function}
    Let $N$ and $M$ be positive integers.
    Let $\cF=\{f: [N]\to[M]\}$ be a family of functions.
    $\cF$ is a $t$-wise independent function family if, for all distinct $x_1,\ldots,x_t\in[N]$ and all $y_1,\ldots,y_t\in[M]$,
    \begin{align}
        \Pr_{f\gets\cF}[f(x_1)=y_1\wedge\cdots\wedge f(x_t)=y_t]=M^{-t}.
    \end{align}
\end{definition}

We use the following two lemmas.

\begin{lemma}[\cite{Gol08}]\label{lem:alg_hash}
    There exists a classical algorithm running in time $\poly(\log N,\log M,t)$ that, on input $t$, $N$, and $M$, samples a function $f:[N]\to[M]$ from a $t$-wise independent function family.
    In addition, each such function $f$ is computable in time $\poly(\log N,\log M,t)$.
\end{lemma}

\begin{lemma}[\cite{zhandry2015secure}]\label{lem:sim_with_hash}
    Let $\cG=\{g:[N]\to[M]\}$ be a $2t$-wise independent function family.
    Then, for any quantum algorithm $\cA^{(\cdot)}$ that makes $t$ queries,
    \begin{align}
        \Pr_{g\gets\cG}[1\gets\cA^{g}]=\Pr_{f}[1\gets\cA^{f}],
    \end{align}
    where $f:[N]\to[M]$ is a uniformly random function. 
\end{lemma}

\subsection{Noise Models and Quantum Error-Correcting Codes}
\label{subsec:noise_and_QECC}

We introduce the noise models considered throughout this paper.
A classical noise channel is a (possibly inefficient) classical probabilistic algorithm $\mathcal{E}$.
When clear from context, we omit the qualifier ``classical.''
We will often use the following notion of bounded noise.

\begin{definition}[Bounded noise]\label{def:p-bounded_noise}
    For $p\in[0,1]$ and a function $m:\N\to\N$, a family $\{\cE_\secp\}_{\secp\in\N}$ of channels
    $\cE_\secp:\bit^{m(\secp)}\to\bit^{m(\secp)}$
    is called a $p$-bounded noise channel if
    \begin{align}
        \Pr_{x\gets\bit^{m(\secp)}}\big[\weight(\cE_\secp(x)\oplus x)>p\cdot m(\secp)\big]
        \le
        \negl(\secp),
    \end{align}
    where the probability is taken over the choice of $x$ and the internal randomness of $\cE_\secp$.
\end{definition}

For quantum states, we consider local quantum noise.
To formalize it, we first introduce some notation.
For any $m$-qubit Pauli operator $P\in\cP_m$, let $\weight(P)$ denote the number of qubits on which $P$ acts nontrivially.
For Pauli expansions, we use the phase-free representatives $\{X^xZ^z:x,z\in\bit^m\}$, one from each equivalence class under global phases.
Given an operator $E$ acting on $m$ qubits, write its unique Pauli expansion as $E=\sum_{x,z\in\bit^m}\alpha_{x,z}X^xZ^z$.
We then define
\begin{align}
    \weight(E)\coloneqq\max_{x,z\in\bit^m:\,\alpha_{x,z}\neq0}\weight(X^xZ^z).
\end{align}

\begin{definition}[Local noise channel]\label{def:local_noise}
    For functions $t,m:\N\to\N$, a family $\{\cE_\secp\}_{\secp\in\N}$ of CPTP maps $\cE_\secp$ from $m(\secp)$ qubits to $m(\secp)$ qubits is called $t$-local noise if there exists a family $\{\widetilde{\cE}_\secp\}_{\secp\in\N}$ of CPTP maps $\widetilde{\cE}_\secp$ from $m(\secp)$ qubits to $m(\secp)$ qubits such that each $\widetilde{\cE}_\secp$ admits a Kraus representation $\widetilde{\cE}_\secp(\cdot)=\sum_i E_i(\cdot)E_i^\dag$ with $\weight(E_i)\le t(\secp)$ for all $i$, and
    \begin{align}
        \|\cE_\secp-\widetilde{\cE}_\secp\|_\diamond\le\negl(\secp).
    \end{align}
\end{definition}

We use the following lemma, whose proof is straightforward and deferred to \cref{sec:omitted_proof}.

\begin{restatable}{lemma}{PauliTwirling}\label{lem:Pauli_twirling_for_local_noise}
    Let $t,m$ be integers.
    Given a CPTP map $\cE$ from $m$ qubits to $m$ qubits, define the CPTP map $\cE^{\mathrm{tw}}$ from $m$ qubits to $m$ qubits by $\cE^{\mathrm{tw}}\coloneqq\Exp_{P\gets\cP_m}P^\dag\circ\cE\circ P$.
    If $\cE$ has a Kraus representation $\cE(\cdot)=\sum_i E_i(\cdot)E_i^\dag$ with $\weight(E_i)\le t$ for all $i$, then there exists a probability distribution $\{p_Q\}_Q$ such that
    \begin{align}
        \cE^{\mathrm{tw}}(\cdot)=\sum_{Q\in\cP_m:\,\weight(Q)\le t} p_Q\,Q(\cdot)Q^\dag.
    \end{align}
\end{restatable}

We define quantum error-correcting codes (QECCs) as pairs of efficient encoding and decoding algorithms.

\begin{definition}[QECC]\label{def:QECC}
    Let $n,m:\N\to\N$ be polynomials.
    Let $\cE=\{\cE_\secp\}_{\secp\in\N}$ be a family of CPTP maps from $m(\secp)$ qubits to $m(\secp)$ qubits.
    A quantum error-correcting code (QECC) for $\cE$ is a pair of QPT algorithms $(\Enc,\Dec)$, where $\Enc_\secp\coloneqq\Enc(1^\secp,\cdot)$ maps $n(\secp)$ qubits to $m(\secp)$ qubits, and $\Dec_\secp\coloneqq\Dec(1^\secp,\cdot)$ maps $m(\secp)$ qubits to $n(\secp)$ qubits, such that
    \begin{align}
        \|\Dec_\secp\circ\cE_\secp\circ\Enc_\secp-\identitymap_{n(\secp)}\|_\diamond\le\negl(\secp),
    \end{align}
    where $\identitymap_{n(\secp)}$ denotes the identity map on $n(\secp)$ qubits.
    In addition, we require that $\Enc_\secp$ is a QPT-implementable isometry.
\end{definition}

For local noise, there are many constructions of QECCs. For example, see~\cite{Book:NieChu10}.

\begin{theorem}\label{thm:QECC_exist}
    There exist constants $R,\alpha>0$ such that the following holds.
    For any polynomial $m$ and any polynomial $n$ satisfying
    $n(\secp)\le Rm(\secp)$ and $m(\secp)=O(n(\secp))$, there is a QECC
    that encodes $n(\secp)$ qubits into $m(\secp)$ qubits and corrects every
    local noise channel whose Kraus operators have weight at most
    $\lfloor\alpha m(\secp)\rfloor$.
\end{theorem}

\subsection{Graph-Theoretic Preliminaries}
\label{subsec:graph_preliminaries}

In this subsection, we introduce some graph-theoretic notation.

\begin{definition}[Simple graph]\label{def:simple_graph}
    Let $n \in \mathbb{N}$.
    A \emph{simple graph} on $[n]$ is a pair $G=([n],E)$, where
    \begin{align}
        E \subseteq 
        \bigl\{
            \{i,j\} \subseteq [n] : i \neq j
        \bigr\}.
    \end{align}
\end{definition}

\begin{definition}[Adjacency matrix]\label{def:adjacency_matrix}
    Let $G=([n],E)$ be a simple graph.
    Its adjacency matrix over $\mathbb{F}_2$ is the matrix $A=(A_{i,j}) \in \mathbb{F}_2^{n \times n}$ defined by
    \begin{align}
        A_{i,j}
        \coloneqq
        \begin{cases}
            1, & \text{if } \{i,j\} \in E, \\
            0, & \text{otherwise.}
        \end{cases}
    \end{align}
\end{definition}

\begin{definition}[Degree]\label{def:degree}
    Let $G=([n],E)$ be a simple graph.
    The degree of $i \in [n]$ is defined by
    \begin{align}
        \deg_G(i) \coloneqq |\{\, j \in [n] : \{i,j\} \in E \,\}|.
    \end{align}
    The maximum degree of $G$ is defined by $\Delta(G)\coloneqq\max_{i\in[n]}\deg_G(i).$
\end{definition}

\subsection{Codeword-Stabilized Codes}
\label{subsec:CWS}

We briefly recall the standard binary CWS formalism \cite{CSSZ08}, which will be used in our construction of PRICs in \cref{sec:PRIC_const}.
Our presentation follows the standard-form description of CWS codes via graph states and classical binary codes.

\begin{definition}[CWS encoding unitary]\label{def:CWS_encoding_unitary}
    Let $G=([n],E)$ be a simple graph.
    We define the unitary
    \begin{align}
        U_G
        \coloneqq
        \left(\prod_{\{i,j\}\in E} \CZ_{ij}\right) H^{\otimes n},
    \end{align}
    where $H$ is the Hadamard gate and $\CZ_{ij}$ is the controlled-$Z$ gate on qubits $i$ and $j$.
\end{definition}

In our PRIC construction, we first encode the classical string $z$ into a codeword $c_z$ using the underlying classical encoder, such as the PRFCs, and then apply the encoding unitary $U_G$ to obtain the corresponding CWS codeword $U_G\ket{c_z}$.
In the decoding procedure, we apply $U_G^\dagger$ before running the underlying classical decoder. 
Therefore, we need an explicit description of how $U_G^\dagger$ acts in the presence of Pauli errors.
The following lemma gives an explicit expression for $U_G^\dagger P U_G$ when $P$ is a Pauli operator.
Since the proof is straightforward, we defer it to \cref{sec:omitted_proof}.

\begin{restatable}{lemma}{ErrorPropagation}\label{lem:error_propagation}
    Let $G=([n],E)$ be a simple graph with adjacency matrix $A$ over $\mathbb{F}_2$.
    Let $P=Z^vX^u$ with $u,v\in\mathbb{F}_2^n$.
    Then
    \begin{align}
        U_G^\dagger P U_G = (-1)^{q_G(u)}X^{v\oplus Au} Z^u,
        \qquad q_G(u)\coloneqq\sum_{\{i,j\}\in E}u_i u_j.
    \end{align}
\end{restatable}

\subsection{Cryptographic Assumptions}
In this subsection, we introduce several cryptographic assumptions that will be used to construct post-quantum PRCs.
First, we recall the Learning Parity with Noise (LPN) assumption \cite{Pie12}.
To define the LPN assumption, we first introduce some notation.
For $p\in[0,1/2)$, $\Ber(p)$ denotes the Bernoulli distribution on $\bit$ with bias $p$, i.e., $\Pr_{e\gets\Ber(p)}[e=1]=p$.
$\Ber(n,p)$ denotes the distribution over $\bit^n$ where each bit is independently sampled from $\Ber(p)$.
The LPN assumption (against QPT adversaries) is defined as follows.

\begin{assumption}[\textbf{LPN assumption \cite{Pie12}}]\label{assumption:LPN}
    For $\eta\in(0,1/2)$ and a function $g:\N\to\N$, the $\LPN_{g,\eta}$ assumption is the following statement.
    For any QPT algorithm $\cA$,
    \begin{align}
        \left|
         \Pr_{\substack{A\gets\F_2^{n \times g(n)},\\s\gets\F_2^{g(n)},e\gets\Ber(n,\eta)}}[1\gets\cA(A,As\oplus e)]
         -\Pr_{A\gets\F_2^{n \times g(n)},u\gets\F_2^{n}}[1\gets\cA(A,u)]
        \right|
        \le\negl(n).
    \end{align}
\end{assumption}

Next, we define the planted XOR assumption (against QPT adversaries).

\begin{assumption}[\textbf{Planted XOR assumption \cite{ASSVN23}}]\label{assumption:planted_XOR}
    For functions $m,t:\N\to\N$, the $\XOR_{m,t}$ assumption is the following statement.
    For any QPT algorithm $\cA$,
    \begin{align}
        \bigg|
         \Pr_{G\gets\F_2^{n\times m(n)}}[1\gets\cA(G)]
         -\Pr_{G\gets\cD(n,m(n),t(n))}[1\gets\cA(G)]
        \bigg|
        \le\negl(n),
    \end{align}
    where $\cD(n,m(n),t(n))$ is the following distribution over $\F_2^{n\times m(n)}:$
    \begin{enumerate}
        \item 
        Sample $s\gets\{s\in\F_2^n: \weight(s)=t(n)\}$, i.e., $s$ is uniformly sampled from $\F^n_2$ conditioned on having Hamming weight $t(n)$.
        \item Sample a uniformly random matrix $G\in\F_2^{n\times m(n)}$ subject to $s^\top G=0$.
        \item Output $G$.
    \end{enumerate}
\end{assumption}

To construct post-quantum PRCs, we use the following assumption, which is the post-quantum analog of the assumption used to construct classically secure PRCs in \cite{C:ChrGun24}.

\begin{assumption}[\textbf{Post-quantum version of \cite{C:ChrGun24}}]\label{assumption:for_PRC}
    Either of the following two statements is true.
    \begin{itemize}
        \item There exists a constant $\eta\in(0,1/2)$ such that, for any function $g(n)=\Omega(\log^2n)$, the $\LPN_{g,\eta}$ assumption holds.
        \item 
        There exist constants $\eta\in(0,1/2)$ and $\epsilon\in(0,1/2)$ such that, for any function $t(n)=\Theta(\log n)$, both the $\LPN_{n^\epsilon,\eta}$ assumption and the $\XOR_{2n^\epsilon,t}$ assumption hold.
    \end{itemize}
\end{assumption}

\begin{remark}
    The first item assumes $\LPN_{g,\eta}$ for every $g(n)=\Omega(\log^2 n)$, where $n$ is the number of samples and $g(n)$ is the secret length.
    When $g(n)=\Theta(\log^2 n)$, polynomial time in $n$ is $2^{O(\sqrt{g(n)})}$ time in the secret length; in the expression $2^{O(\sqrt n)}$ used in \cref{sec:introduction}, $n$ denotes the secret length.
\end{remark}

\if0
\begin{definition}[Sparse Vectors and Matrices]
    We define that 
    \begin{align}
        \cS_{t,n}\coloneqq\{s\in\F_2^n: \weight(s)=t\}
    \end{align}
    is the set of $t$-sparse vectors.    We also define
    \begin{align}
        \cS_{t,r,n}\coloneqq\{P\in\F_2^{r \times n}: \weight(P_{i,;})=t~ \text{ for all }i\in[r]\}
    \end{align}
    as the set of matrices whose rows are all $t$-sparse vectors, where, for $P\in \F_2^{r \times n}$, $P_{i,;}\in\F_2^n$ denotes the $i$th row of $P$.
\end{definition}
\fi

\subsection{Cryptographic Primitives}

In this subsection, we introduce cryptographic primitives.
We first define pseudorandom functions (PRFs) and pseudorandom permutations (PRPs).

\begin{definition}[Pseudorandom functions (PRFs)]\label{def:PRF}
    Let $n,m:\N\to\N$ be polynomials.
    A deterministic polynomial-time algorithm $F:\bit^\secp\times\bit^{n(\secp)}\to\bit^{m(\secp)}$
    is a pseudorandom function (PRF) if, for any QPT adversary $\cA$,
    \begin{align}
        \bigg|
         \Pr_{k\gets\bit^\secp}[1\gets\cA^{f_k}(1^\secp)]
         -\Pr_{f}[1\gets\cA^{f}(1^\secp)]
        \bigg|
        \le\negl(\secp),
    \end{align}
    where $f_k(\cdot)\coloneqq F(k,\cdot)$.
    Here, $f$ is a uniformly random function $\bit^{n(\secp)}\to\bit^{m(\secp)}$.
\end{definition}

\begin{definition}[Pseudorandom permutations (PRPs)]\label{def:PRP}
    Let $n:\N\to\N$ be a polynomial.
    A pair of deterministic polynomial-time algorithms $\Pi:\bit^\secp\times\bit^{n(\secp)}\to\bit^{n(\secp)}$ and $\Pi^{-1}:\bit^\secp\times\bit^{n(\secp)}\to\bit^{n(\secp)}$ is a pseudorandom permutation (PRP) if $\pi_k(\cdot)\coloneqq \Pi(k,\cdot)$ and $\pi^{-1}_k(\cdot)\coloneqq \Pi^{-1}(k,\cdot)$ are permutations over $\bit^{n(\secp)}$ such that $\pi_k\circ \pi^{-1}_k$ and $\pi^{-1}_k\circ\pi_k$ are the identity, and, for any QPT adversary $\cA$,
    \begin{align}
        \bigg|
         \Pr_{k\gets\bit^\secp}[1\gets\cA^{\pi_k,\pi_k^{-1}}(1^\secp)]
         -\Pr_{\pi}[1\gets\cA^{\pi,\pi^{-1}}(1^\secp)]
        \bigg|
        \le\negl(\secp).
    \end{align}
    Here, $\pi$ is a uniformly random permutation $\bit^{n(\secp)}\to\bit^{n(\secp)}$.
\end{definition}

\if0
\begin{definition}[Pseudorandom Genrators]
    Let $n,m:\N\to\N$ be functions such that $n(\secp)<m(\secp)$ for sufficiently large $\secp\in\N$.
    A polynomial-time deterministic algorithm $G$ is a pseudorandom generator (PRG) if it, on input $\secp$ and $r\in\bit^{n(\secp)}$, outputs a $G(1^\secp,r)\in\bit^{m(\secp)}$, and, for any QPT adversary $\cA$ and sufficiently large $\secp\in\N$,
    \begin{align}
        \bigg|
         \Pr_{r\gets\bit^{n(\secp)}}[1\gets\cA(G(1^\secp,r))]
         -\Pr_{u\gets\bit^{m(\secp)}}[1\gets\cA(u)]
        \bigg|
        \le\negl(\secp).
    \end{align}
\end{definition}
\fi

Next, we recall the definition of PRCs from \cite{C:ChrGun24}.
Throughout this paper, unless stated otherwise, security is required to hold against quantum adversaries; we therefore omit the qualifier ``post-quantum.''
Accordingly, we keep the definition of PRCs unchanged, except that the security requirement is formulated with respect to QPT distinguishers.
There are two variants of PRCs: public-key PRCs and secret-key PRCs.
We first define public-key PRCs.

\begin{definition}[Public-key PRCs \cite{C:ChrGun24}]\label{def:PKPRC}
    Let $n,m:\N\to\N$ be polynomials.
    Let $\cE\coloneq \{\cE_\secp\}_{\secp\in\mathbb{N}}$ be a noise channel, where $\cE_\secp:\bit^{m(\secp)}\to\bit^{m(\secp)}$.
    An $(n,m)$ public-key pseudorandom error-correcting code (PRC) robust to $\cE$ is a tuple of algorithms $(\KeyGen,\Enc,\Dec)$ satisfying the following syntax and conditions.
    \begin{itemize}
        \item Syntax:
         \begin{itemize}
             \item $\KeyGen:$ 
             It is a PPT algorithm that takes the security parameter $\secp$ as input and outputs a pair consisting of a secret key $\sk$ and a public key $\pk$.
             \item $\Enc:$ 
             It is a PPT algorithm that takes a public key $\pk$ and a plaintext $x\in\bit^{n(\secp)}$ as input and outputs a ciphertext $c\in\bit^{m(\secp)}$.
             \item $\Dec$: 
             It is a deterministic polynomial-time algorithm that takes a secret key $\sk$ and a bit string $c\in\bit^{m(\secp)}$ as input and outputs $x\in\bit^{n(\secp)}$ or $\bot$.
         \end{itemize}
        \item Robustness:
        For any $x\in\bit^{n(\secp)}$,
        \begin{align}
            \Pr
            \left[
             x\gets\Dec(\sk,\cE_\secp( c)):
             \begin{array}{r}
               (\sk,\pk)\gets\KeyGen(1^\secp)\\
                c\gets\Enc(\pk,x)
             \end{array}
            \right]
            \ge1-\negl(\secp).
        \end{align}

        \item Soundness:
        For any fixed $c\in\bit^{m(\secp)}$,
        \begin{align}
            \Pr[\bot\gets\Dec(\sk,c):(\sk,\pk)\gets\KeyGen(1^\secp)]\ge1-\negl(\secp).
        \end{align}

        \item Pseudorandomness:
        For any QPT adversary $\cA$ that is allowed to make only classical queries to an oracle,
        \begin{align}
            \bigg|
             \Pr_{(\sk,\pk)\gets\KeyGen(1^\secp)}[1\gets\cA^{\Enc_\pk}(1^\secp,\pk)]
             -\Pr_{(\sk,\pk)\gets\KeyGen(1^\secp)}[1\gets\cA^{\cU}(1^\secp,\pk)]
            \bigg|
            \le\negl(\secp),
        \end{align}
        where $\Enc_\pk$ and $\cU$ are defined as follows.
        \begin{itemize}
            \item For each $\pk$, $\Enc_\pk$ is the following oracle: on input $x\in\bit^{n(\secp)}$, it runs $y\gets\Enc(\pk,x)$ and outputs $y$.
            \item $\cU$ is the following oracle: for each query, on any input, it chooses $y\gets\bit^{m(\secp)}$ and outputs $y$.
        \end{itemize}
    \end{itemize}
    When $n$ and $m$ are clear from the context, we simply refer to $(\KeyGen, \Enc, \Dec)$ as a public-key PRC robust to $\cE$.
    Similarly, when $\cE$ is clear from the context, we simply refer to it as an $(n,m)$ public-key PRC.
    When $n$, $m$, and $\cE$ are all clear from the context, we simply refer to it as a public-key PRC.
\end{definition}

\begin{remark}[Deterministic requirement for the decoding algorithm]
    For convenience, we require $\Dec$ to be deterministic in the above definition.
    This entails no loss of generality.
    Indeed, if $\Dec$ were PPT, then $\KeyGen$ could sample its internal randomness as part of the secret key $\sk$ and regard $\Dec$ as a deterministic algorithm with that randomness fixed.
\end{remark}

\begin{remark}[Stronger robustness notions]\label{remark:robustness_of_PRC}
    Stronger robustness notions for PRCs have been considered in the literature, such as adaptive robustness and strongly-adaptive robustness~\cite{STOC:AACDG25,EPRINT:CGGMW25}.
    In this paper, however, our main goal is to construct the quantum variants, i.e., QPRCs and PRICs.
    For these constructions, it suffices to start from a PRC satisfying the above robustness notion.
    Therefore, we restrict our attention to this basic notion of robustness throughout the paper.
\end{remark}

\begin{remark}[On soundness]
    We include soundness in the above definition following \cite{C:ChrGun24}.
    This condition is mainly a technical condition that rules out trivial PRCs when $n=0$.
    Indeed, when the message space contains only one message, robustness alone does not prevent the decoder from accepting every input as a valid codeword.
    For PRCs with nonzero message length, such a triviality is not the issue we need to address.
    In our setting, the properties that will be essential for the subsequent constructions are pseudorandomness and robustness.
\end{remark}

\begin{remark}[On classical queries]
    In the above security definition, the adversary is allowed to make only classical queries.
    One could strengthen the definition by allowing quantum queries, as in the case of encryption schemes \cite{C:BonZha13}.
    Although the above security notion is weaker than such a quantum-query notion, it is sufficient for our purposes, namely for constructing QPRCs and PRICs.
    We leave the study of such a stronger post-quantum security notion to future work.
\end{remark}

Secret-key PRCs are defined analogously.

\begin{definition}[Secret-key PRCs \cite{C:ChrGun24}]\label{def:SKPRC}
    Let $n,m:\N\to\N$ be polynomials.
    Let $\cE\coloneq \{\cE_\secp\}_{\secp\in\mathbb{N}}$ be a noise channel, where $\cE_\secp:\bit^{m(\secp)}\to\bit^{m(\secp)}$.
    An $(n,m)$ secret-key pseudorandom error-correcting code (PRC) robust to $\cE$ is a tuple of algorithms $(\KeyGen,\Enc,\Dec)$ satisfying the following syntax and conditions.
    \begin{itemize}
        \item Syntax:
         \begin{itemize}
             \item $\KeyGen:$ 
             It is a PPT algorithm that takes the security parameter $\secp$ as input and outputs a secret key $\sk$.
             \item $\Enc:$ 
             It is a PPT algorithm that takes a secret key $\sk$ and a plaintext $x\in\bit^{n(\secp)}$ as input and outputs a ciphertext $c\in\bit^{m(\secp)}$.
             \item $\Dec$: 
             It is a deterministic polynomial-time algorithm that takes a secret key $\sk$ and a bit string $c\in\bit^{m(\secp)}$ as input and outputs $x\in\bit^{n(\secp)}$ or $\bot$.
         \end{itemize}
        \item Robustness:
        For any $x\in\bit^{n(\secp)}$,
        \begin{align}
            \Pr
            \left[
             x\gets\Dec(\sk,\cE_\secp( c)):
             \begin{array}{r}
               \sk\gets\KeyGen(1^\secp)\\
                c\gets\Enc(\sk,x)
             \end{array}
            \right]
            \ge1-\negl(\secp).
        \end{align}

        \item Soundness:
        For any fixed $c\in\bit^{m(\secp)}$,
        \begin{align}
            \Pr[\bot\gets\Dec(\sk,c):\sk\gets\KeyGen(1^\secp)]\ge1-\negl(\secp).
        \end{align}

        \item Pseudorandomness:
        For any QPT adversary $\cA$ that is allowed to make only classical queries to an oracle,
        \begin{align}
            \bigg|
             \Pr_{\sk\gets\KeyGen(1^\secp)}[1\gets\cA^{\Enc_\sk}(1^\secp)]
             -\Pr[1\gets\cA^{\cU}(1^\secp)]
            \bigg|
            \le\negl(\secp),
        \end{align}
        where $\Enc_\sk$ and $\cU$ are defined as follows.
        \begin{itemize}
            \item For each $\sk$, $\Enc_\sk$ is the following oracle: on input $x\in\bit^{n(\secp)}$, it runs $y\gets\Enc(\sk,x)$ and outputs $y$.
            \item $\cU$ is the following oracle: for each query, on any input, it chooses $y\gets\bit^{m(\secp)}$ and outputs $y$.
        \end{itemize}
    \end{itemize}
    When $n$ and $m$ are clear from the context, we simply refer to $(\KeyGen, \Enc, \Dec)$ as a secret-key PRC robust to $\cE$.
    Similarly, when $\cE$ is clear from the context, we simply refer to it as an $(n,m)$ secret-key PRC.
    When $n$, $m$, and $\cE$ are all clear from the context, we simply refer to it as a secret-key PRC.
\end{definition}

It is clear that secret-key PRCs imply IND-CPA-secure secret-key encryption schemes (SKEs).
Since IND-CPA-secure SKEs imply PRFs and PRPs, we obtain the following.

\begin{corollary}\label{coro:PRC_imply_PRF}
    If secret-key PRCs exist, then PRFs and PRPs exist.
\end{corollary}

In \cite{C:ChrGun24}, public-key and secret-key PRCs robust to $p$-bounded noise for any $p\in[0,1/2)$, as defined in \cref{def:p-bounded_noise}, are constructed under \cref{assumption:for_PRC} against PPT adversaries.
Although their result is stated for PPT adversaries, the same proof, with \cref{assumption:for_PRC} formulated against quantum adversaries, yields public-key and secret-key PRCs secure against QPT adversaries.
Thus, we obtain the following.

\begin{restatable}{theorem}{PostQuantumPRC}\label{thm:post-quantum_PRC}
    Let $p\in[0,1/2)$.
    There exists a constant $\alpha$ satisfying the following.
    For polynomials $n,m:\N\to\N$ satisfying $m(\secp)\ge \alpha\cdot n(\secp)$ for all sufficiently large $\secp$, if \cref{assumption:for_PRC} is true, then there exist $(n,m)$ public-key PRCs and $(n,m)$ secret-key PRCs robust to $p$-bounded noise.
\end{restatable}

We give the proof in \cref{sec:post_quantum_CG24} for completeness.

\section{Definitional Work}
\label{sec:QPRC_def}

In this section, we define quantum pseudorandom error-correcting codes (QPRCs).
Let $n_\logic,n_\phys:\N\to\N$ be polynomials, where $n_\logic(\secp)\le n_\phys(\secp)$ for all $\secp$.
We refer to $n_\logic$ and $n_\phys$ as the numbers of logical and physical qubits, respectively.
In the following, for any family of CPTP maps $\cE\coloneqq\{\cE_\secp\}_{\secp\in \N}$ mapping $n_\phys(\secp)$ qubits to $n_\phys(\secp)$ qubits, we call such a family $\cE$ a noise channel.
In addition, for any family of distributions $\cD\coloneqq\{\cD_\secp\}_{\secp\in \N}$ over CPTP maps from $n_\logic(\secp)$ qubits to $n_\phys(\secp)$ qubits, we call such a family $\cD$ an ideal distribution.

\subsection{Quantum Pseudorandom Error-Correcting Codes}
\label{subsec:QPRC_def}
First, we define public-key QPRCs.

\begin{definition}[Public-key quantum pseudorandom error-correcting codes (PKQPRCs)]\label{def:PK-QPRC}
    Let $\cE\coloneqq\{\cE_\secp\}_{\secp\in \N}$ be a noise channel.
    Let $\cD\coloneqq\{\cD_\secp\}_{\secp\in \N}$ be an ideal distribution.
    A public-key $\cD$-QPRC robust to $\cE$ is a tuple of algorithms $(\KeyGen,\Enc,\Dec)$ satisfying the following syntax and conditions.
    \begin{itemize}
        \item Syntax:
         \begin{itemize}
             \item $\KeyGen:$ 
             It is a QPT algorithm that takes the security parameter $\secp$ as input and outputs a pair consisting of a classical secret key $\sk$ and a classical public key $\pk$.
             \item $\Enc:$ 
             It is a QPT algorithm that, given a public key $\pk$, implements a CPTP map $\Enc_\pk$ from $n_\logic$ logical qubits to $n_\phys$ physical qubits.
             \item $\Dec$: 
             It is a QPT algorithm that, given a secret key $\sk$, implements a CPTP map $\Dec_\sk$ from $n_\phys$ physical qubits to $n_\logic$ logical qubits.
         \end{itemize}
        \item Robustness:
        \begin{align}
            \bigg\|
             \Exp_{(\pk,\sk)\gets\KeyGen(1^\secp)}(\Dec_\sk\circ\cE_\secp\circ\Enc_\pk)
             -\identitymap
            \bigg\|_\diamond
            \le\negl(\secp).
        \end{align}

        \item Pseudorandomness:
        For any QPT adversary $\cA^{(\cdot)}$,
        \begin{align}
            \bigg|
             \Pr_{(\pk,\sk)\gets\KeyGen(1^\secp)}[1\gets\cA^{\Enc_\pk}(1^\secp,\pk)]
             -\Pr_{\substack{(\pk,\sk)\gets\KeyGen(1^\secp)\\
               \cI\gets\cD_\secp}}[1\gets\cA^{\cI}(1^\secp,\pk)]
            \bigg|
            \le\negl(\secp).
        \end{align}
    \end{itemize}
    When $\cE$ is clear from the context, we refer to it as a public-key $\cD$-QPRC.
    When $\cD$ is clear from the context, we refer to it as a public-key QPRC robust to $\cE$.
    When both $\cD$ and $\cE$ are clear from the context, we refer to it as a public-key QPRC.
\end{definition}

In the definition of pseudorandomness, the ideal channel is sampled independently of the key and reused for all oracle queries. When the ideal distribution outputs the completely depolarizing channel with probability one, this is the quantum analogue of the classical PRC ideal oracle, which returns a fresh uniform string on each query.

In the construction in \cref{sec:QPRC_const}, we specialize the ideal distribution to the point-mass distribution on the completely depolarizing channel, which is the quantum analogue of the classical PRC ideal oracle.

\begin{remark}\label{remark:robustness_of_QPRC}
    As mentioned in \cref{remark:robustness_of_PRC}, several stronger notions of robustness have been considered for classical PRCs \cite{STOC:AACDG25,EPRINT:CGGMW25}.
    One can similarly formulate their quantum counterparts.
    However, studying such stronger notions is beyond the scope of this paper, and we leave their study to future work.
\end{remark}

Next, we define secret-key QPRCs.

\begin{definition}[Secret-key quantum pseudorandom error-correcting codes]\label{def:SK-QPRC}
    Let $\cE\coloneqq\{\cE_\secp\}_{\secp\in \N}$ be a noise channel.
    Let $\cD\coloneqq\{\cD_\secp\}_{\secp\in \N}$ be an ideal distribution.
    A secret-key $\cD$-QPRC robust to $\cE$ is a tuple of algorithms $(\KeyGen,\Enc,\Dec)$ satisfying the following syntax and conditions.
    \begin{itemize}
        \item Syntax:
         \begin{itemize}
             \item $\KeyGen:$ 
             It is a QPT algorithm that takes the security parameter $\secp$ as input and outputs a classical secret key $\sk$.
             \item $\Enc:$ 
             It is a uniform QPT algorithm that, given a secret key $\sk$, implements a keyed, efficiently computable, stateless CPTP map $\Enc_\sk$ from $n_\logic$ logical qubits to $n_\phys$ physical qubits.
             \item $\Dec$: 
             It is a QPT algorithm that, given a secret key $\sk$, implements a CPTP map $\Dec_\sk$ from $n_\phys$ physical qubits to $n_\logic$ logical qubits.
         \end{itemize}
        \item Robustness:
        \begin{align}
            \bigg\|
             \Exp_{\sk\gets\KeyGen(1^\secp)}(\Dec_\sk\circ\cE_\secp\circ\Enc_\sk)
             -\identitymap
            \bigg\|_\diamond
            \le\negl(\secp).
        \end{align}
        
        \item Pseudorandomness:
        For any QPT adversary $\cA^{(\cdot)}$,
        \begin{align}
            \bigg|
             \Pr_{\sk\gets\KeyGen(1^\secp)}[1\gets\cA^{\Enc_\sk}(1^\secp)]
              -\Pr_{\substack{\sk\gets\KeyGen(1^\secp)\\
               \cI\gets\cD_\secp}}[1\gets\cA^{\cI}(1^\secp)]
            \bigg|
             \le\negl(\secp),
        \end{align}
    \end{itemize}
    When $\cE$ is clear from the context, we simply refer to it as a secret-key $\cD$-QPRC.
    When $\cD$ is clear from the context, we refer to it as a secret-key QPRC robust to $\cE$.
    When both $\cD$ and $\cE$ are clear from the context, we refer to it as a secret-key QPRC.
\end{definition}

\subsection{Pseudorandom Isometric Error-Correcting Codes}
\label{subsec:PRIC_def}
In this subsection, we define PRICs as QPRCs by choosing the Haar measures of isometries as an ideal distribution, and discuss why we define PRICs this way.
We define and often use $\mu_{n_\logic\to n_\phys}\coloneqq\{\mu_{n_\logic(\secp)\to n_\phys(\secp)}\}_\secp$, where $\mu_{n_\logic(\secp)\to n_\phys(\secp)}$ is the Haar measure over isometries from $n_\logic(\secp)$ qubits to $n_\phys(\secp)$ qubits.

\begin{definition}[PRICs]\label{def:PRIC}
    Let $\cE$ be a noise channel.
    We call a tuple of algorithms $(\KeyGen,\Enc,\Dec)$ a pseudorandom isometric code (PRIC) robust to $\cE$ if it is a secret-key $\mu_{n_\logic\to n_\phys}$-QPRC robust to $\cE$.
\end{definition}

\paragraph{Why are PRICs secret-key primitives?}
In the above definition, we consider only the secret-key setting.
This is necessary because, regardless of noise, no public-key $\mu_{n_\logic\to n_\phys}$-QPRCs exist.

\begin{proposition}[Impossibility of public-key $\mu_{n_\logic\to n_\phys}$-QPRCs]
\label{prop:no_public_key_haar_qprc}
No public-key $\mu_{n_\logic\to n_\phys}$-QPRC exists, regardless of the
noise channel.
\end{proposition}

\begin{proof}[Proof of \cref{prop:no_public_key_haar_qprc}]
Let $(\KeyGen,\Enc,\Dec)$ be a tuple of algorithms satisfying the syntax of public-key QPRCs.
We construct a QPT adversary that breaks the pseudorandomness of this tuple with respect to $\mu_{n_\logic\to n_\phys}$.
Let $\secp$ be the security parameter. We omit the dependence on $\secp$ in the notation for simplicity.
Consider the following QPT algorithm $\cA$.
\begin{enumerate}
    \item Receive a public key $\pk$ and oracle access to $\cO$.

    \item Choose $b\gets\bit$.

    \begin{itemize}
        \item If $b=0$, prepare two copies of $\sigma\coloneqq\mathcal O(\ketbra{0^{n_\logic}}{0^{n_\logic}})$ by querying $\cO$.
        Perform the SWAP test on the two copies and output $1$ if the test
        rejects, and $0$ otherwise.

        \item If $b=1$, prepare a single copy of $\rho_\pk\coloneqq\Enc_\pk(\ketbra{0^{n_\logic}}{0^{n_\logic}})$ using $\pk$, and prepare $\sigma$ by querying $\cO$.
        Do the SWAP test for them, and output $1$ if the test is successful, otherwise, output $0$.
    \end{itemize}
\end{enumerate}
For each $\pk$, we have
\begin{align}
    \Pr[1\gets\cA^\cO(\pk)]=\frac{2-\Tr[\sigma^2]+\Tr[\rho_\pk\sigma]}{4},
\end{align}
where we omit the input $1^\secp$ for simplicity.
If $\cO=\Enc_\pk$, then $\sigma=\rho_\pk$, and this probability is $1/2$.
If $\cO=V(\cdot)V^\dag$ is a Haar-random isometry, then the probability is
\begin{align}
    \Pr_{V}[1\gets\cA^{V}(\pk)]
    &=\frac{1}{2}-\frac{1}{4}\Exp_V
    \bigg[
     \Tr[(V\ketbra{0^{n_\logic}}{0^{n_\logic}}V^\dagger)^2]
    \bigg]
    +\frac{1}{4}\Exp_V
    \bigg[
     \Tr[\rho_\pk\cdot V\ketbra{0^{n_\logic}}{0^{n_\logic}}V^\dagger]
    \bigg]
    \notag\\
    &=\frac{1}{4}+\frac{1}{4\cdot2^{n_\phys}},
\end{align}
where the expectation is taken over $\mu_{n_\logic\to n_\phys}$, and we have used $\Exp_V[V\ketbra{0^{n_\logic}}{0^{n_\logic}}V^\dagger]=I^{\otimes n_\phys}/2^{n_\phys}$.
Thus, we obtain
\begin{align}
    \Pr_{(\pk,\sk)\gets\KeyGen(1^\secp)}[1\gets\cA^{\Enc_\pk}(\pk)]-\Pr_{(\pk,\sk)\gets\KeyGen(1^\secp),V}[1\gets\cA^{V}(\pk)]
    =\frac{1}{4}-\frac{1}{4\cdot2^{n_\phys}},
\end{align}
which is non-negligible.
Therefore $\cA$ breaks pseudorandomness, which completes the proof.
\end{proof}

\begin{remark}
\label{remark:no_public_key_isometry_one_design}
We can extend the above impossibility to any ideal distribution $\nu=\{\nu_\secp\}_\secp$ forming an exact isometry $1$-design.
This is because we used only the first moment of Haar-random isometries in the proof. 
\end{remark}

\paragraph{Why do we not consider Haar-random unitaries as an ideal distribution?}
One may consider Haar-random unitaries as an ideal distribution for QPRCs rather than Haar-random isometries.
This choice requires the logical and physical spaces to have the same
dimension, leaving no redundancy to protect against irreversible noise.
Thus, we expect that such QPRCs do not exist.

Related limitations are well established~\cite{GrasslHuberWinter22,BarnumFuchsJozsaSchumacher96}, and these results support the intuition that insufficient redundancy
precludes error correction, but they do not directly give the impossibility statement needed here.
Our definition allows arbitrary CPTP encoders and decoders with a
shared classical secret key, and measures recovery error by the
diamond distance of the key-averaged channel from the identity.

We therefore give an impossibility result for our setting as follows.

\begin{theorem}[Impossibility of QPRCs with small redundancy]\label{thm:impossibility_w_small_redandancy}
    Let $\cD$ be an arbitrary ideal distribution.
    Let $t:\N\to\N$ be an arbitrary function satisfying $n_\phys(\secp)-n_\logic(\secp)<t(\secp)\le n_\phys(\secp)$ for all sufficiently large $\secp$.
    Then, there exists a $t$-local error map $\cE$ such that no secret-key $\cD$-QPRCs robust to $\cE$ exist.
    In addition, if there exists a positive integer $c$ such that $n_\phys(\secp)-n_\logic(\secp)< c$ for all sufficiently large $\secp$, there exists a $c$-local error map $\cE$ such that no secret-key $\cD$-QPRCs robust to $\cE$ exist.
\end{theorem}

Taking $c=1$ in \cref{thm:impossibility_w_small_redandancy} shows that QPRCs compared with Haar-random unitaries cannot be robust to every one-local noise channel.

\begin{corollary}
    Suppose that $n_\logic=n_\phys$, and let $\mu_{n_\phys}\coloneqq\{\mu_{n_\phys(\secp)}\}_\secp$, where $\mu_{n_\phys(\secp)}$ is the Haar measure over unitaries acting on $n_\phys(\secp)$ qubits.
    Then, there exists a one-local error map $\cE$ such that no secret-key $\mu_{n_\phys}$-QPRCs robust to $\cE$ exist.
\end{corollary}

To prove \cref{thm:impossibility_w_small_redandancy}, we need the following lemma.

\begin{restatable}{lemma}{ImpossibilityOfCPTPEncode}
\label{lem:cptp_encoding_no_redundancy}
    Let $(\KeyGen,\Enc,\Dec)$ be a tuple of algorithms satisfying the syntax of secret-key QPRCs.
    Let $\secp$ be the security parameter.
    For an integer $t$ satisfying $0\le t\le n_\phys(\secp)$, define a CPTP map $\Delta_{[t]}$ from $n_\phys(\secp)$ qubits to $n_\phys(\secp)$ qubits as the completely depolarizing map on the first $t$ qubits, i.e.,
    \begin{align}
        \Delta_{[t]}(\cdot)\coloneqq
        \frac{I^{\otimes t}}{2^t}\otimes\Tr_{[t]}[(\cdot)],
    \end{align}
    where $\Tr_{[t]}$ denotes the partial trace over the first $t$ qubits.
    Then,
    \begin{align}
        \left\|
          \Exp_\sk\bigl[\Dec_\sk\circ\Delta_{[t]}\circ\Enc_\sk\bigr]
    -\identitymap
        \right\|_\diamond
        \ge\max\{0,1-4^{n_\phys-n_\logic-t}\}.
    \end{align}
\end{restatable}

We defer the proof to \cref{sec:omitted_proof}. 
Now we give the proof of \cref{thm:impossibility_w_small_redandancy}.

\begin{proof}[Proof of \cref{thm:impossibility_w_small_redandancy}]
    In this proof, we suppress the dependence on $\cD$ for simplicity.
    We first prove the first claim.
    To this end, we construct a $t$-local error map $\cE$ for which no tuple of algorithms satisfying the syntax of a secret-key QPRC is robust.
    Define $\cE\coloneqq\{\cE_\secp\}_{\secp}$, where, for all sufficiently large $\secp$, $\cE_\secp\coloneqq\Delta_{[t(\secp)]}$ is the completely depolarizing map on the first $t(\secp)$ qubits, and, for smaller $\secp$, $\cE_\secp$ is the identity map.
    It is clear that $\cE$ is $t$-local.
    Let $(\KeyGen,\Enc,\Dec)$ be a tuple of algorithms satisfying the syntax of secret-key QPRCs.
    Fix a sufficiently large security parameter $\secp$.
    We omit its dependence for notational simplicity.
    From $t>n_\phys-n_\logic$ and \cref{lem:cptp_encoding_no_redundancy}, we have
    \begin{align}
        \left\|
          \Exp_\sk\bigl[\Dec_\sk\circ\Delta_{[t]}\circ\Enc_\sk\bigr]
          -\identitymap
        \right\|_\diamond
        \ge\max\{0,1-4^{n_\phys-n_\logic-t}\}
        \ge \frac{3}{4},
    \end{align}
    Thus, $(\KeyGen,\Enc,\Dec)$ is not robust against $\cE$.

    The second claim follows by the same argument, which concludes the proof.
\end{proof}

\section{Construction of QPRCs for Completely Depolarizing Channel}
\label{sec:QPRC_const}

In this section, we construct QPRCs encoding $n_\logic$ logical qubits into $n_\phys$ physical qubits, robust to $t$-local noise with $t=O(n_\phys)$ and $n_\phys=O(n_\logic)$.
As the ideal distribution $\cD$, we consider the distribution that outputs the completely depolarizing channel with probability one, and we omit $\cD$ in this section.
Since public-key QPRCs imply secret-key QPRCs, it suffices to construct public-key QPRCs.

\begin{construction}[Public-key QPRC]\label{const:PKQPRC}
    Let $p\in(0,1/2)$.
    Let $n_\logic,t,m,\ell:\N\to\N$ be polynomials.
    Let $(\QECC.\Enc,\QECC.\Dec)$ be a QECC for $t$-local noise that encodes $n_\logic$ qubits into $\ell$ qubits.
    Let $(\PRC.\KeyGen,\PRC.\Enc,\PRC.\Dec)$ be a $(2\ell,m)$ public-key PRC robust to any $p$-bounded noise.
    Define $n_\phys\coloneqq \ell+m$.
    We construct a public-key QPRC $(\KeyGen,\Enc,\Dec)$ encoding $n_\logic$ logical qubits into $n_\phys$ physical qubits as follows.
    \begin{itemize}
        \item 
        $\KeyGen:$ 
        On input $1^\secp$, run $(\sk_\PRC,\pk_\PRC)\gets\PRC.\KeyGen(1^\secp)$.
        Output the pair $(\sk,\pk)$, where $\sk\coloneqq \sk_\PRC$ and $\pk\coloneqq\pk_\PRC$.
        
        \item 
        $\Enc:$
        Parse $\pk= \pk_\PRC$.
        Let $\regA,\regB$, and $\regC$ be $n_\logic$-qubit, $\ell$-qubit, and $m$-qubit registers, respectively.
        Given $\pk$, implement the following operation.
        Choose $x,z\gets\bit^\ell$ and internal randomness $r$ for $\PRC.\Enc$, and apply the isometry 
        \begin{align}
            \ket{\psi}_\regA\mapsto (X^xZ^z)_{\regB}\cdot
            \QECC.\Enc_{\regA\to\regB}\cdot
            \ket{\psi}_\regA
            \otimes\ket{\PRC.\Enc(\pk_\PRC,x\|z;r)}_\regC
        \end{align}
        from $\regA$ to $\regB\regC$.

        \item 
        $\Dec:$
        Parse $\sk=\sk_\PRC$.
        On input an $n_\phys$-qubit state $\rho$ on $\regB\regC$, apply the following operations:
        \begin{enumerate}
            \item Measure $\regC$ in the computational basis to get $c\in\bit^m$.
            Run $\PRC.\Dec(\sk_\PRC,c)$. If the result is a string $x\|z\in\bit^{2\ell}$, use it; if the result is $\bot$, set $x\|z\coloneqq 0^{2\ell}$. Apply the unitary $(X^xZ^z)^\dag$ to register $\regB$.
            \item 
            Apply $\QECC.\Dec_{\regB\to\regA}$ and output register $\regA$.
        \end{enumerate}
    \end{itemize}
\end{construction}

We prove the following theorem.

\begin{theorem}\label{thm:PKQPRC}
    Let $p\in(0,1/2)$.
    If $t(\secp)\le p\cdot m(\secp)$ for sufficiently large $\secp$, then \cref{const:PKQPRC} is a public-key QPRC robust to any $t$-local noise.
\end{theorem}

From \cref{thm:QECC_exist,thm:post-quantum_PRC,thm:PKQPRC}, we have the following.

\begin{corollary}\label{coro:QPRC}
    There exist constants $c_0$ and $c_1$ satisfying the following.
    For polynomials $n_\logic,n_\phys,t:\N\to\N$ satisfying $n_\phys\ge c_0\cdot n_\logic(\secp)$ and $t(\secp)\le c_1\cdot n_\phys(\secp)$ for all sufficiently large $\secp$, if \cref{assumption:for_PRC} is true, then there exist public-key and secret-key QPRCs encoding $n_\logic$ logical qubits into $n_\phys$ physical qubits, robust to any $t$-local noise.
\end{corollary}

To prove \cref{thm:PKQPRC}, we establish the robustness and pseudorandomness of \cref{const:PKQPRC} in \cref{subsec:QPRC_robustness,subsec:QPRC_pseudorandomness}, respectively.

\subsection{Robustness}
\label{subsec:QPRC_robustness}
\newcommand{\PauliRecover}{\algo{PauliRecover}}

\begin{lemma}\label{lem:robustness_PKQPRC}
    \cref{const:PKQPRC} is robust to any $t$-local noise if $t(\secp)\le p\cdot m(\secp)$ for sufficiently large $\secp$.
\end{lemma}
\begin{proof}[Proof of \cref{lem:robustness_PKQPRC}]
    By the definition of $t$-local noise, choose an exact-local channel
    $\widetilde{\cE}_\secp$ with
    $\|\cE_\secp-\widetilde{\cE}_\secp\|_\diamond\le\negl(\secp)$.
    We first carry out the argument for $\widetilde{\cE}_\secp$ and write it as
    $\cE_\secp$ below. Contractivity under the encoding and decoding CPTP maps
    adds at most $\negl(\secp)$ when we return to the original channel.
    In the following, we denote by $\Exp_{(\sk,\pk)}$ the expectation over $(\sk,\pk)\gets\KeyGen(1^\secp)$.
    Let $\regA$ and $\regA'$ be $n_\logic$-qubit registers.
    Our goal is to show
    \begin{align}
        \frac{1}{2}
        \bigg\|
         \bigg(
          \identitymap_{\regA'}\otimes
          \Exp_{(\sk,\pk)}
          (\Dec_\sk\circ\cE_\secp\circ\Enc_\pk)_\regA
         \bigg)
         (\ketbra{\psi}{\psi}_{\regA'\regA})
         -\ketbra{\psi}{\psi}_{\regA'\regA}
        \bigg\|_1
        \le\negl(\secp)
    \end{align}
    for any $2n_\logic$-qubit pure state $\ket{\psi}$ and any $t$-local noise $\cE_\secp$.
    In the following, we omit the identity map $\identitymap$ on the register $\regA'$ and the subscript $\secp$ of $\cE_\secp$ for notational simplicity.
    For the hybrid argument, we define the following states.

    \begin{itemize}
        \item $\rho^{(0)}_{\regA'\regA}:$
        This is defined as
        \begin{align}
            \bigg(
             \Exp_{(\sk,\pk)}
             (\Dec_\sk\circ\cE\circ\Enc_\pk)_\regA
            \bigg)
            (\ketbra{\psi}{\psi}_{\regA'\regA})
        \end{align}

        \item $\rho^{(1)}_{\regA'\regA}:$
        This is defined as
        \begin{align}
            &\bigg(
            \Exp_{(\sk,\pk)}
            \QECC.\Dec_{\regB\to\regA}\circ
            \PauliRecover_{\sk,\regB\regC\to\regB}\circ 
            \cE_{\regB\regC}\circ
            \Enc_{\pk,\regA\to\regB\regC}
            \bigg)
            (\ketbra{\psi}{\psi}_{\regA'\regA}),
        \end{align}
        where $\PauliRecover_\sk$ is the CPTP map that operates in the first step of $\Dec$, i.e., it is defined as follows:
        \begin{align}
            \PauliRecover_{\sk,\regB\regC\to\regB}(\cdot)\coloneqq \sum_{c\in\bit^m} 
            \bigg(
            (X^{x_{\sk,c}}Z^{z_{\sk,c}})^\dag_\regB\otimes \bra{c}_\regC
            \bigg)
            (\cdot)_{\regB\regC}
            \bigg(
            (X^{x_{\sk,c}}Z^{z_{\sk,c}})_\regB\otimes \ket{c}_\regC
            \bigg).
        \end{align}
        Here $x_{\sk,c}\Vert z_{\sk,c}$ is $\PRC.\Dec(\sk,c)$ when decoding returns a $2\ell$-bit string, and is $0^{2\ell}$ when decoding returns $\bot$.

        \item $\rho^{(2)}_{\regA'\regA}:$
        This is defined as
        \begin{align}
            \bigg(
            \Exp_{x,z\gets\bit^{\ell}}
            \QECC.\Dec_{\regB\to\regA}\circ
            (X^xZ^z)^\dag_{\regB}\circ
            \cE_{x,z,\regB}\circ
            (X^xZ^z)_{\regB}\circ
            \QECC.\Enc_{\regA\to\regB}
            \bigg)
            (\ketbra{\psi}{\psi}_{\regA'\regA}),
        \end{align}
        where, for each $x,z\in\bit^{\ell}$,
        \begin{align}
            \cE_{x,z,\regB}(\cdot)\coloneqq
            \Tr_{\regC}
            \Bigg[
            \cE_{\regB\regC}
            \bigg(
            (\cdot)_{\regB}\otimes
            \Exp_{\substack{(\sk,\pk),r}}
            \ketbra{\PRC.\Enc(\pk,x\|z;r)}{\PRC.\Enc(\pk,x\|z;r)}_\regC
            \bigg)
            \Bigg].
        \end{align}

        \item $\rho^{(3)}_{\regA'\regA}:$
        This is defined as $\ketbra{\psi}{\psi}_{\regA'\regA}$.
    \end{itemize}

    It is sufficient to prove that $\rho^{(0)}$ is statistically close to $\rho^{(3)}$.
    From the definitions of $\Enc$ and $\Dec$ in \cref{const:PKQPRC}, we have $\rho^{(0)}=\rho^{(1)}$.
    From the following claim, we have $\frac{1}{2}\|\rho^{(1)}-\rho^{(2)}\|_1\le\negl(\secp)$.
    
    \begin{claim}\label{claim:sigma_vs_sigma'}
        For any fixed $x,z\in\bit^{\ell}$, we have
        $\frac{1}{2}\|\sigma_{x,z}-\sigma'_{x,z}\|_1\le\negl(\secp)$,
        uniformly in $x,z$, where
        \begin{align}
            \sigma_{x,z,\regA'\regB}\coloneqq
            &
            \bigg(
            \Exp_{\substack{(\sk,\pk),r}}
            \PauliRecover_{\sk,\regB\regC\to\regB}\circ 
            \cE_{\regB\regC}\circ
            \Enc_{\pk,(x,z,r)\regA\to\regB\regC}
            \bigg)
            (\ketbra{\psi}{\psi}_{\regA'\regA})
        \end{align}
        and 
        \begin{align}
            \sigma'_{x,z,\regA'\regB}\coloneqq
            \bigg(
            (X^xZ^z)_\regB^\dag\circ
            \cE_{x,z,\regB}\circ
            (X^xZ^z)_\regB\circ
            \QECC.\Enc_{\regA\to\regB}
            \bigg)
            (\ketbra{\psi}{\psi}_{\regA'\regA}).
        \end{align}
    \end{claim}

    We prove the claim below.
    For each $x,z$, $\cE_{x,z}$ is $t$-local noise by definition.
    Thus, $(X^xZ^z)^\dag\circ\cE_{x,z}\circ(X^xZ^z)$ is also $t$-local noise.
        Therefore, since $(\QECC.\Enc,\QECC.\Dec)$ is a QECC for $t$-local noise,
        $\frac{1}{2}\|\rho^{(2)}-\rho^{(3)}\|_1\le\negl(\secp)$.
    Combining these bounds, we obtain $\frac{1}{2}\|\rho^{(0)}-\rho^{(3)}\|_1\le\negl(\secp)$, which concludes the proof.
\end{proof}

To complete the proof, we prove \cref{claim:sigma_vs_sigma'}.

\begin{proof}[Proof of \cref{claim:sigma_vs_sigma'}]
        For each public key $\pk$, define
        \begin{align}
            \Omega^\pk_{\regA'\regB\regC}
            &\coloneqq
            \bigg(
            \Exp_{r}
            \cE\circ
             \Enc_{\pk,(x,z,r)\regA\to\regB\regC}
            \bigg)
            (\ketbra{\psi}{\psi}_{\regA'\regA}).
        \end{align}
        Let
        \begin{align}
            G_{\sk,x,z}
            &\coloneqq
            \{c\in\bit^m:\PRC.\Dec(\sk,c)=x\|z\}.
        \end{align}
        Also define
        \begin{align}
            \sigma^{(\sk,\pk)}_{\regA'\regB}
            &\coloneqq
            \PauliRecover_{\sk,\regB\regC\to\regB}
            \big(
            \Omega^\pk_{\regA'\regB\regC}
            \big), \\
            \sigma'^{(\sk,\pk)}_{\regA'\regB}
            &\coloneqq
            (X^xZ^z)_\regB^\dag\,
            \Tr_{\regC}
            \Big[
            \Omega^\pk_{\regA'\regB\regC}
            \Big]\,
            (X^xZ^z)_\regB.
        \end{align}
        Note that $\sigma=\Exp_{(\sk,\pk)}[\sigma^{(\sk,\pk)}]$ and $\sigma'=\Exp_{(\sk,\pk)}[\sigma'^{(\sk,\pk)}]$.
        By the definition of $\PauliRecover_{\sk,\regB\regC\to\regB}$,
        \begin{align}
            \sigma^{(\sk,\pk)}_{\regA'\regB}
            &=
            \sum_{c\in G_{\sk,x,z}}
            (X^xZ^z)_\regB^\dag\,
            \Tr_{\regC}
            \Big[
            (I_{\regA'\regB}\otimes \ketbra{c}{c}_{\regC})
            \Omega^\pk_{\regA'\regB\regC}
            \Big]\,
            (X^xZ^z)_\regB
            \notag\\
            &\quad+
            \sum_{c\notin G_{\sk,x,z}}
            (X^{x_{\sk,c}}Z^{z_{\sk,c}})_\regB^\dag\,
            \Tr_{\regC}
            \Big[
            (I_{\regA'\regB}\otimes \ketbra{c}{c}_{\regC})
            \Omega^\pk_{\regA'\regB\regC}
            \Big]\,
            (X^{x_{\sk,c}}Z^{z_{\sk,c}})_\regB,
        \end{align}
        whereas
        \begin{align}
            \sigma'^{(\sk,\pk)}_{\regA'\regB}
            &=
            \sum_{c\in G_{\sk,x,z}}
            (X^xZ^z)_\regB^\dag\,
            \Tr_{\regC}
            \Big[
            (I_{\regA'\regB}\otimes \ketbra{c}{c}_{\regC})
            \Omega^\pk_{\regA'\regB\regC}
            \Big]\,
            (X^xZ^z)_\regB
            \notag\\
            &\quad+
            \sum_{c\notin G_{\sk,x,z}}
            (X^xZ^z)_\regB^\dag\,
            \Tr_{\regC}
            \Big[
            (I_{\regA'\regB}\otimes \ketbra{c}{c}_{\regC})
            \Omega^\pk_{\regA'\regB\regC}
            \Big]\,
            (X^xZ^z)_\regB.
        \end{align}
        Therefore,
        \begin{align}
        \frac{1}{2}
        \big\|
        \sigma^{(\sk,\pk)}_{\regA'\regB}
        -
        \sigma'^{(\sk,\pk)}_{\regA'\regB}
        \big\|_1
        &\le
        \frac{1}{2}
        \Bigg\|
        \sum_{c\notin G_{\sk,x,z}}
        (X^{x_{\sk,c}}Z^{z_{\sk,c}})_\regB^\dag\,
        \Tr_{\regC}             \Big[             (I_{\regA'\regB}\otimes \ketbra{c}{c}_{\regC})\Omega^\pk_{\regA'\regB\regC}             \Big]\,
        (X^{x_{\sk,c}}Z^{z_{\sk,c}})_\regB
        \Bigg\|_1
        \notag\\
        &\quad+
        \frac{1}{2}
        \Bigg\|
        \sum_{c\notin G_{\sk,x,z}}
        (X^xZ^z)_\regB^\dag\,
        \Tr_{\regC}             \Big[             (I_{\regA'\regB}\otimes \ketbra{c}{c}_{\regC})\Omega^\pk_{\regA'\regB\regC}             \Big]\,
        (X^xZ^z)_\regB
        \Bigg\|_1.
    \end{align}
    Both operators on the right-hand side are positive semidefinite, so their trace norm equals their trace. 
    Thus,
        \begin{align}
            \frac{1}{2}
            \big\|
            \sigma^{(\sk,\pk)}_{\regA'\regB}
            -
            \sigma'^{(\sk,\pk)}_{\regA'\regB}
            \big\|_1
            &\le
            \sum_{c\notin G_{\sk,x,z}}
            \Tr
            \Big[
            (I_{\regA'\regB}\otimes \ketbra{c}{c}_{\regC})
            \Omega^\pk_{\regA'\regB\regC}
            \Big] \\
            &=
            1-
            \sum_{c\in G_{\sk,x,z}}
            \Tr
            \Big[
            (I_{\regA'\regB}\otimes \ketbra{c}{c}_{\regC})
            \Omega^\pk_{\regA'\regB\regC}
            \Big].
        \end{align}
        Since $\sigma=\Exp_{(\sk,\pk)}[\sigma^{(\sk,\pk)}]$ and $\sigma'=\Exp_{(\sk,\pk)}[\sigma'^{(\sk,\pk)}]$,
        convexity of the trace norm implies
        \begin{align}
            \frac{1}{2}
            \big\|
            \sigma_{\regA'\regB}
            -
            \sigma'_{\regA'\regB}
            \big\|_1
            &\le
            \Exp_{(\sk,\pk)}
            \Bigg[
            1-
            \sum_{c\in G_{\sk,x,z}}
            \Tr
            \Big[
            (I_{\regA'\regB}\otimes \ketbra{c}{c}_{\regC})
            \Omega^\pk_{\regA'\regB\regC}
            \Big]
            \Bigg].
        \end{align}
        Define
        \begin{align}
            \eta_{\regB}
            &\coloneqq
            \Tr_{\regA'}
            \Big[
            \Big(
            (X^xZ^z)_\regB\circ
            \QECC.\Enc_{\regA\to\regB}
            \Big)
            (\ketbra{\psi}{\psi}_{\regA'\regA})
            \Big].
        \end{align}
        Using $\cE_{\regB\regC}$ and the fixed state $\eta_{\regB}$, define the following classical channel $\mathcal{N}_\eta$ on $\bit^m$:
        on input $c\in\bit^m$, it prepares $\eta_{\regB}\otimes\ketbra{c}{c}_{\regC}$, applies $\cE_{\regB\regC}$, and measures $\regC$ in the computational basis to output the outcome.
        Then
        \begin{align}
            \sum_{c\in G_{\sk,x,z}}
            \Tr
            \Big[
            (I_{\regA'\regB}\otimes \ketbra{c}{c}_{\regC})
            \Omega^\pk_{\regA'\regB\regC}
            \Big]
            =
            \Pr
            \Big[
            \PRC.\Dec\big(\sk,\mathcal{N}_\eta(\PRC.\Enc(\pk,x\|z))\big)=x\|z
            \Big],
        \end{align}
        where the probability is over the internal randomness of $\PRC.\Enc$ and the induced channel $\cN_\eta$.
        Because $\cE_{\regB\regC}$ is $t$-local and $t\le pm$, the induced classical noise changes at most $pm$ bits and is therefore $p$-bounded.
        Therefore, by the robustness of the underlying PRC,
        \begin{align}
            \frac{1}{2}
            \big\|
            \sigma_{\regA'\regB}
            -
            \sigma'_{\regA'\regB}
            \big\|_1
            \le
            1-\Pr_{(\sk,\pk)\gets\KeyGen(1^\secp)}
            \Big[
            \PRC.\Dec\big(\sk,\mathcal{N}_\eta(\PRC.\Enc(\pk,x\|z))\big)=x\|z
            \Big]
            \le
            \negl(\secp),
        \end{align}
        which concludes the proof.
\end{proof}

\subsection{Pseudorandomness}
\label{subsec:QPRC_pseudorandomness}

\begin{lemma}\label{lem:pseudorandom_PKQPRC}
    \cref{const:PKQPRC} satisfies pseudorandomness.
\end{lemma}
\begin{proof}[Proof of \cref{lem:pseudorandom_PKQPRC}]
    Let $\cA$ be any QPT adversary. 
    We define the following sequence of oracles for the hybrid argument.

    \begin{itemize}
        \item ${\rm{H}}_0:$
        This is the real oracle: for each query, it samples $x,z\gets\bit^\ell$ and $r$ uniformly at random and applies
        \begin{align}
             \ket{\psi}_\regA&\longmapsto (X^xZ^z)_\regB W_{\regA\to\regB} \ket{\psi}_A \otimes \ket{\PRC.\Enc(\pk,x\|z;r)}_\regC,
        \end{align}
        where $W_{\regA\to\regB}\coloneqq\QECC.\Enc_{\regA\to\regB}$.

        \item ${\rm{H}}_1:$
        For each query, the oracle samples $x,z\gets\bit^\ell$ and $u \gets \bit^{m}$ uniformly at random and applies
        \begin{align}
             \ket{\psi}_\regA&\longmapsto (X^xZ^z)_\regB W_{\regA\to\regB} \ket{\psi}_\regA \otimes \ket{u}_\regC .
        \end{align}

        \item ${\rm{H}}_2:$
        This is the ideal oracle that applies the completely depolarizing channel $\mathcal{D}$ defined by $\mathcal{D}(X)=\operatorname{Tr}(X)I_{\regB\regC}/2^{n_\phys}$ on each query.
    \end{itemize}

    Our goal is to show that ${\rm{H}}_0$ is indistinguishable from ${\rm{H}}_2$.
    Note that ${\rm{H}}_0$ depends on $\pk$, but ${\rm{H}}_1$ and ${\rm{H}}_2$ are independent of $\pk$.
    From the security of $\PRC$, we have
    \begin{align}
        \bigg|
        \Pr_{(\sk,\pk)\gets\KeyGen(1^\secp)}[1\gets\cA^{{\rm{H}}_0}(\pk)]
        -\Pr_{(\sk,\pk)\gets\KeyGen(1^\secp)}[1\gets\cA^{{\rm{H}}_1}(\pk)]
        \bigg|
        \le\negl(\secp).
    \end{align}
    Now we argue that ${\rm{H}}_1$ is perfectly indistinguishable from ${\rm{H}}_2$.
    Let $\ket{\phi}_{\regX\regA}$ be any state, and let $\regX$ be an arbitrary finite-dimensional register.
    A straightforward calculation gives
    \begin{align}
        \Exp_{x,z\gets\bit^\ell} (X^xZ^z\circ W)_{\regA\to\regB}(\ketbra{\phi}{\phi}_{\regX\regA})
         =\Tr_\regA[\ketbra{\phi}{\phi}_{\regX\regA}]\otimes \frac{I_\regB}{2^{\ell}}
    \end{align}
    and, for the completely depolarizing channel,
    \begin{align}
        (\identitymap_\regX\otimes\mathcal{D}_{\regA\to\regB\regC})(\ketbra{\phi}{\phi}_{\regX\regA})
        =\Tr_\regA[\ketbra{\phi}{\phi}_{\regX\regA}]\otimes \frac{I_{\regB\regC}}{2^{n_\phys}}.
    \end{align}
    Thus, since $\Exp_{u}\ketbra{u}{u}=\frac{I}{2^m}$, we have
    \begin{align}
        \Pr_{(\sk,\pk)\gets\KeyGen(1^\secp)}[1\gets\cA^{{\rm{H}}_1}(\pk)]
        =\Pr_{(\sk,\pk)\gets\KeyGen(1^\secp)}[1\gets\cA^{{\rm{H}}_2}(\pk)].
    \end{align}
    Because ${\rm H}_1$ samples fresh independent $x,z,u$ on every query and
    ${\rm H}_2$ is a memoryless CPTP channel, the calculation above is an
    equality of the induced CPTP channels. It therefore remains valid for
    adaptive multi-query adversaries, including reference-entangled queries.
    Therefore, we obtain
    \begin{align}
        \bigg|
        \Pr_{(\sk,\pk)\gets\KeyGen(1^\secp)}[1\gets\cA^{{\rm{H}}_0}(\pk)]
        -\Pr_{(\sk,\pk)\gets\KeyGen(1^\secp)}[1\gets\cA^{{\rm{H}}_2}(\pk)]
        \bigg|
        \le\negl(\secp),
    \end{align}
    which concludes the proof.
\end{proof}

\if0
\subsection{Soundness}
\label{subsec:QPRC_soundness}

\begin{lemma}\label{lem:soundness_PKQPRC}
    \cref{const:PKQPRC} satisfies soundness.
\end{lemma}

\begin{proof}[Proof of \cref{lem:soundness_PKQPRC}]
    Our goal is to show
    \begin{align}
        \Pr_{\sk\gets\KeyGen(1^\secp)}[\bot\gets\Dec(\sk,\rho)]\ge1-\negl(\secp)
    \end{align}
    for any $(n_\phys+m)$-qubit state $\rho$.
    Due to the concavity, it is sufficient to consider the case when $\rho$ is a pure state $\ket{\psi}$.
    Let $\ket{\psi}=\sum_{c\in\bit^m}\ket{\alpha_c}_\regB\ket{c}_\regC$, where $\regB$, and $\regC$ are an $n_\phys$-qubit and an $m$-qubit registers, respectively. 
    Parse $\sk=\sk_\PRC$.
    Then, we have
    \begin{align}
        \Pr_{\sk\gets\KeyGen(1^\secp)}[\bot\gets\Dec(\sk,\ket{\psi})]
        =&\Exp_{\sk\gets\KeyGen(1^\secp)}
        \bigg[
         \bigg\|
          \ketbra{\bot}{\bot}_\regC\cdot
         \ket{\psi}_{\regB\regC}
        \bigg\|_2^2
        \bigg]
        \notag\\
        =&
        \sum_{c}\|\ket{\alpha_c}\|^2_2\cdot
        \Pr_{\sk_\PRFC\gets\PRC.\KeyGen(1^\secp)}[\PRFC.\Dec(\sk_\PRC,c)=\bot]
        \notag\\
        \ge&
        \sum_{c}\|\ket{\alpha_c}\|^2_2\cdot
        (1-\negl(\secp))
        \notag\\
        =&1-\negl(\secp),
    \end{align}
    where we have used soundness of PRC in the fourth line, and $\sum_{c}\|\ket{\alpha_c}\|^2_2=\|\ket{\psi}\|^2_2=1$ in the last line.
    This concludes the proof.
\end{proof}
\fi

\section{Pseudorandom Functional Codes}
\label{sec:PRFC}

In this section, we introduce a PRF-like variant of PRCs, which we call \emph{pseudorandom functional codes} (PRFCs).
Informally, a PRFC is a PRF equipped with an error-correcting property, and it is precisely the primitive needed for our construction of PRICs.
We formalize PRFCs in \cref{subsec:PRFC:def} and construct them from secret-key PRCs in \cref{subsec:PRFC:const}.
To the best of our knowledge, this primitive and its construction are new even in the classical setting.

\subsection{Definition}
\label{subsec:PRFC:def}

In this subsection, we define PRFCs.

\begin{definition}[Pseudorandom functional codes (PRFCs)]\label{def:PRFC}
    Let $n, m : \N \to \N$ be polynomials.
    Let $\cE\coloneq \{\cE_\secp\}_{\secp\in\mathbb{N}}$ be a noise channel, where $\cE_\secp:\bit^{m(\secp)}\to\bit^{m(\secp)}$.
    An $(n,m)$ pseudorandom functional code (PRFC) robust to $\cE$ is a tuple of algorithms $(\KeyGen, \Enc, \Dec)$ satisfying the following syntax and properties.

    \begin{itemize}
        \item Syntax:
         \begin{itemize}
             \item $\KeyGen:$ 
             It is a PPT algorithm that takes the security parameter $1^\secp$ as input and outputs a secret key $\sk$.
             \item $\Enc:$ 
             It is a deterministic polynomial-time algorithm that takes a secret key $\sk$ and a plaintext $x \in \bit^{n(\secp)}$ as input and outputs a ciphertext $c \in \bit^{m(\secp)}$.
             \item $\Dec$: 
             It is a deterministic polynomial-time algorithm that takes a secret key $\sk$ and a bit string $c \in \bit^{m(\secp)}$ as input and outputs $x \in \bit^{n(\secp)}$.
         \end{itemize}
        \item Robustness:
        For any $x\in\bit^{n(\secp)}$,
        \begin{align}
            \Pr_{\sk\gets\KeyGen(1^\secp)}
            \left[
             \Dec(\sk,c')=x:c'\gets\cE_\secp(\Enc(\sk,x))
            \right]
            \ge1-\negl(\secp).
        \end{align}

        \item Pseudorandomness:
        For any QPT adversary $\cA$,
        \begin{align}
            \bigg|
             \Pr_{\sk\gets\KeyGen(1^\secp)}[1\gets\cA^{\Enc(\sk,\cdot)}(1^\secp)]
             -\Pr_{f}[1\gets\cA^f(1^\secp)]
            \bigg|
            \le\negl(\secp),
        \end{align}
        where $f:\bit^n\to\bit^m$ is a uniformly random function.
        Here, for any function $g:\bit^n\to\bit^m$, $\cA^g$ can query the unitary that maps
        \begin{align}
            \ket{x}\ket{y}\mapsto
            \ket{x}\ket{y\oplus g(x)}
        \end{align}
        for any $x\in\bit^n$ and $y\in\bit^m$.
    \end{itemize}
\end{definition}

\begin{remark}
    Unlike the definition of PRCs in \cref{def:SKPRC}, \cref{def:PRFC} does not include a soundness requirement. 
    For PRCs, soundness is needed because when $n=0$, the notion becomes trivial without it. In this paper, however, we construct PRFCs only in the regime $n \ge \omega(\log \secp)$, which is also the only regime needed in our construction of PRICs. For this reason, we omit soundness from the definition for simplicity.
\end{remark}

\begin{remark}
    In the above definition, we require that $\Dec$ is a deterministic algorithm.
    However, this entails no loss of generality, because the internal randomness of $\Dec$ can be included in the secret key $\sk$.
\end{remark}

Regardless of the noise $\cE$, $F(\sk,x)\coloneqq \Enc(\sk,x)$ is a PRF if $(\KeyGen, \Enc, \Dec)$ is a PRFC robust to $\cE$.
In addition, PRPs are implied by PRFs \cite{Quantum:Zhandry25}.
Thus, we immediately obtain the following.

\begin{corollary}\label{thm:PRFC_imply_PRF/PRP}
    For any polynomials $n$ and $m$, if there exists an $(n,m)$-PRFC robust to some noise $\cE$, then PRFs and PRPs exist.
\end{corollary}

\subsection{Construction}
\label{subsec:PRFC:const}

We construct PRFCs from secret-key PRCs, PRFs, and PRPs as follows.

\begin{construction}[PRFCs]\label{const:PRFC}
    Let $n,m,\ell:\N\to\N$ be polynomials.
    Let $(\PRC.\KeyGen, \PRC.\Enc, \PRC.\Dec)$ be an $(n,m)$ secret-key PRC robust to a noise channel $\cE$ such that $\PRC.\Enc$ uses an $\ell(\secp)$-bit string as internal randomness.
    Let $\Pi:\bit^\secp\times \bit^{n(\secp)}\to\bit^{n(\secp)}$ be a PRP, and define $\pi_k(\cdot)\coloneqq \Pi(k,\cdot)$ for each $k\in\bit^\secp$.
    Let $F:\bit^\secp\times \bit^{n(\secp)}\to\bit^{\ell(\secp)}$ be a PRF, and define $f_k(\cdot)\coloneqq F(k,\cdot)$ for each $k\in\bit^\secp$.
    We construct a PRFC $(\KeyGen, \Enc, \Dec)$ as follows.
    \begin{itemize}
        \item $\KeyGen:$ 
        On input $1^\secp$, run $\sk'\gets\PRC.\KeyGen(1^\secp)$, and choose $k,k'\gets\bit^\secp$.
        Choose $r\gets\bit^{\ell(\secp)}$.
        Output $\sk\coloneqq(\sk',k,k',r)$.
        
        \item $\Enc:$
        Parse $\sk=(\sk',k,k',r)$.
        On input $x\in\bit^{n(\secp)}$, output
        \begin{align}
            \Enc(\sk,x)\coloneqq
            \PRC.\Enc(\sk',\pi_{k}(x);f_{k'}(x)\oplus r).
        \end{align}

        \item $\Dec:$
        Parse $\sk=(\sk',k,k',r)$.
        On input $c\in\bit^{m(\secp)}$, run $\PRC.\Dec(\sk',c)$ to get $v\in\bit^{n(\secp)}\cup\{\bot\}$.
        If $v\in\bit^{n(\secp)}$, output $\pi_{k}^{-1}(v)$.
        Otherwise, output $0^{n(\secp)}$.
    \end{itemize}
\end{construction}

We show that the above construction satisfies \cref{def:PRFC}.

\begin{theorem}\label{thm:PRFC}
    If $n(\secp)=\omega(\log\secp)$ and $n(\secp)\le m(\secp)$, then $(\KeyGen,\Enc,\Dec)$ in \cref{const:PRFC} is an $(n,m)$ PRFC robust to $\cE$.
\end{theorem}

Before the proof, we mention some consequences.
Since secret-key PRCs imply PRFs and PRPs by \cref{coro:PRC_imply_PRF}, we obtain the following corollary.

\begin{corollary}
    Let $n,m:\N\to\N$ be functions satisfying $n(\secp)=\omega(\log\secp)$ and $n(\secp)\le m(\secp)$.
    If $(n,m)$ secret-key PRCs robust to noise $\cE$ exist, then $(n,m)$ PRFCs robust to $\cE$ exist.
\end{corollary}

In addition, from \cref{thm:post-quantum_PRC}, we obtain an instantiation of PRFCs robust to bounded noise.

\begin{corollary}\label{coro:LPN_imply_PRFC}
    For every $p\in[0,1/2)$, there exists a constant $\alpha>0$ such that the following holds.
    For polynomials $n,m:\N\to\N$ satisfying $n(\secp)=\omega(\log\secp)$ and $m(\secp)\ge \alpha\cdot n(\secp)$ for all sufficiently large $\secp$, if \cref{assumption:for_PRC} is true, then there exist $(n,m)$ PRFCs robust to $p$-bounded noise.
\end{corollary}

Now we prove \cref{thm:PRFC}.

\begin{proof}[Proof of \cref{thm:PRFC}]
    In the following, we omit the dependence on $\secp$ for notational simplicity.
    It is easy to see that $(\KeyGen,\Enc,\Dec)$ satisfies robustness.
    Fix $k',k\in\bit^\secp$.
    For any $x\in\bit^n$,
    \begin{align}
    &\Pr_{\substack{\sk'\gets\PRC.\KeyGen(1^\secp)\\r\gets\bit^\ell}}
    \left[
        \Dec((\sk',k,k',r),c')=x:
        c'\gets\cE_\secp(\Enc((\sk',k,k',r),x))
    \right]
    \notag\\
    \ge{}&
    \Pr_{\substack{\sk'\gets\PRC.\KeyGen(1^\secp)\\r'\gets\bit^\ell}}
    \left[
        \PRC.\Dec(\sk',c')=\pi_k(x):
        c'\gets\cE_\secp(\PRC.\Enc(\sk',\pi_k(x);r'))
    \right]
    \notag\\
    \ge{}&1-\negl(\secp).
    \end{align}
    Here $r'=f_{k'}(x)\oplus r$ is uniform when $r$ is uniform, and the last inequality follows from the robustness of the underlying PRC.
    
    It remains to show pseudorandomness.
    For the sake of contradiction, assume that there exists a polynomial $p$ and a QPT adversary $\cA$ such that, for infinitely many $\secp$,
    \begin{align}
        \bigg|
         \Pr_{\sk\gets\KeyGen(1^\secp)}[1\gets\cA^{\Enc(\sk,\cdot)}]
         -\Pr_{f}[1\gets\cA^f]
        \bigg|
        \ge\frac{1}{p},
        \label{eq:PRFC_assumption}
    \end{align}
    where $f$ is a uniformly random function from $n$ bits to $m$ bits.
    It is sufficient to prove that $\PRC$ does not satisfy the security, assuming the security of a PRF $\{f_k\}_k$ and a PRP $\{\pi_k\}_k$.
    To this end, we use the above $\cA$ to construct a QPT adversary $\cB$ that breaks the security of $(\PRC.\KeyGen,\PRC.\Enc,\PRC.\Dec)$.
    Suppose that $\cA$ makes at most $q$ queries, where $q$ is polynomial in $\secp$.
    Let $t$ be a polynomial that we determine later.
    $\cB$ is the following algorithm:
    \begin{enumerate}
        \item 
        Choose a function $g:\bit^n\to\bit^m$ as follows.
        For each $i\in[t]$, choose $z_i\gets\bit^n$. 
        Query an oracle on $z_i$ and obtain $y_i\in\bit^m$.
        Using the algorithm from \cref{lem:alg_hash}, sample a function $h:\bit^n\to[t]$ from a $4q$-wise independent function family.
        Then define $g$ by $g(x)\coloneqq y_{h(x)}$.\label{step:smaple_g}
        \item 
        Simulate $\cA^g$ using $h$ and $\{y_i\}_{i\in[t]}$ to get $b\in\bit$.
        \item 
        Output $b$.
    \end{enumerate}
    It is clear that $\cB$ is a QPT algorithm.
    To prove that $\cB$ breaks the security of the secret-key PRC, consider the following hybrids of distributions over $\bit$.
    \begin{itemize}
        \item 
        ${\rm H}_0$: this is the distribution of $b\gets\cB^{\PRC.\Enc(\sk',\cdot)}$; that is, choose $\sk'\gets\PRC.\KeyGen(1^\secp)$ and sample a function $g:\bit^n\to\bit^m$ as in step \ref{step:smaple_g}.
        Then, run $b\gets\cA^g$, and output $b$.
        \item 
        ${\rm H}_1$: this is the same as ${\rm H}_0$ except for sampling $g:\bit^n\to\bit^m$ as follows.
        Choose $k,k'\gets\bit^\secp$ and $r\gets\bit^\ell$.
        Then, define $g(x)\coloneqq \PRC.\Enc(\sk',\pi_{k}(x);f_{k'}(x)\oplus r)$.
        \item 
        ${\rm H}_2$: this is the same as ${\rm H}_1$ except for sampling $g:\bit^n\to\bit^m$ as a uniformly random function $g:\bit^n\to\bit^m$.
        \item 
        ${\rm H}_3$: this is the distribution of $b\gets\cB^\cO$, where $\cO$ is the oracle that outputs a uniformly random $m$-bit string.
    \end{itemize}
    Note that $\Pr[1\gets{\rm H}_0]=\Pr_{\sk'\gets\PRC.\KeyGen(1^\secp)}[1\gets\cB^{\PRC.\Enc(\sk',\cdot)}]$ and $\Pr[1\gets{\rm H}_3]=\Pr_{\cO}[1\gets\cB^\cO]$.
    We prove the following two claims.
    \begin{claim}\label{claim:PRFC:H0_to_H1}
        There exists a constant $c>0$ such that
        \begin{align}
            |\Pr[1\gets{\rm H}_0]-\Pr[1\gets{\rm H}_1]|\le \frac{27q^3}{t}+\frac{cq^3}{2^n}+\negl(\secp).
        \end{align}
    \end{claim}
    \begin{claim}\label{claim:PRFC:H2_to_H3}
        We have
        \begin{align}
            |\Pr[1\gets{\rm H}_2]-\Pr[1\gets{\rm H}_3]|\le \frac{27q^3}{t}.
        \end{align}
    \end{claim}
    It is clear that
    \begin{align}
        \Pr[1\gets{\rm H}_1]=\Pr_{\sk\gets\KeyGen(1^\secp)}[1\gets\cA^{\Enc(\sk,\cdot)}]
    \end{align}
    and
    \begin{align}
        \Pr[1\gets{\rm H}_2]=\Pr_{f}[1\gets\cA^f],
    \end{align}
    where $f$ is a uniformly random function from $n$ bits to $m$ bits.
    Thus, from \cref{eq:PRFC_assumption}, 
    \begin{align}
        |\Pr[1\gets{\rm H}_1]-\Pr[1\gets{\rm H}_2]|\ge\frac{1}{p}.\label{eq:PRFC:H1_to_H2}
    \end{align}
    Now define $t\coloneqq 4\cdot27 q^3p$.
    Then, from \cref{claim:PRFC:H0_to_H1,claim:PRFC:H2_to_H3,eq:PRFC:H1_to_H2}, we have
    \begin{align}
        |\Pr[1\gets{\rm H}_0]-\Pr[1\gets{\rm H}_3]|
        &\ge \frac{1}{p}-\frac{2\cdot27q^3}{t}-\frac{cq^3}{2^n}-\negl(\secp)\notag\\
        &=\frac{1}{2p}-\negl(\secp)\notag\\
        &\ge\frac{1}{4p},\label{eq:PRFC:contradiction}
    \end{align}
    where we have used $t= 4\cdot27 q^3p$ and $n=\omega(\log\secp)$.
    Since $\Pr[1\gets{\rm H}_0]=\Pr_{\sk'\gets\PRC.\KeyGen(1^\secp)}[1\gets\cB^{\PRC.\Enc(\sk',\cdot)}]$ and $\Pr[1\gets{\rm H}_3]=\Pr_{\cO}[1\gets\cB^\cO]$, \cref{eq:PRFC:contradiction} implies that $\cB$ breaks the security of a secret-key PRC.
    This concludes the proof.
\end{proof}

To prove \cref{claim:PRFC:H0_to_H1,claim:PRFC:H2_to_H3}, we need the following lemmas.

\begin{lemma}[Corollary 7.5 in \cite{FOCS:Zhandry12}]\label{lem:small_range_oracle}
    Let $\cX$ and $\cY$ be sets.
    Let $\cD$ be a distribution over $\cY$.
    Define $\cD^\cX$ to be the distribution of the following function $f:\cX\to\cY$: for each $x\in\cX$, choose $y_x\gets \cD$ and define $f(x)=y_x$.
    For an integer $t$, define ${\rm{SR}}_t^\cD(\cX)$ as the following distribution on functions $g:\cX\to\cY:$
    \begin{enumerate}
        \item For each $i\in[t]$, choose $y_i\gets \cD$.
        For each $x\in\cX$, choose $i\gets[t]$ and define $g(x)\coloneqq y_i$.
    \end{enumerate}
    Then, for any quantum algorithm $\cA$ that makes at most $q$ queries, 
    \begin{align}
        \bigg|
         \Pr_{g\gets{\rm{SR}}_t^\cD(\cX)}[1\gets\cA^{g}]
         -\Pr_{f\gets \cD^\cX}[1\gets\cA^f]
        \bigg|
        \le\frac{27q^3}{t}.
    \end{align}
\end{lemma}

\begin{lemma}[\cite{Zhandry15,QIC:Yuen13}]\label{lem:function_vs_permutation}
    There exists a constant $c>0$ such that, for any integer $n$ and any quantum algorithm that makes $q$ queries,
    \begin{align}
        \bigg|
         \Pr_{f}[1\gets\cA^f]-\Pr_{\pi}[1\gets\cA^\pi]
        \bigg|
        \le \frac{cq^3}{2^n},
    \end{align}
    where $f$ is a uniformly random function from $n$ bits to $n$ bits, and $\pi$ is a uniformly random permutation over $n$ bits.
\end{lemma}

Now we are ready to show \cref{claim:PRFC:H0_to_H1,claim:PRFC:H2_to_H3}.
First, we give the proof of \cref{claim:PRFC:H0_to_H1}.

\begin{proof}[Proof of \cref{claim:PRFC:H0_to_H1}]
    Let us consider the following hybrids of distributions over $\bit$.
    In the following, the {\color{red}red} text is the new step in the current hybrid.
    \begin{itemize}
        \item 
        ${\rm H'}_0$: Choose $\sk'\gets\PRC.\KeyGen(1^\secp)$ and a function $g:\bit^n\to\bit^m$ as explained below, run $b\gets\cA^g$, and output $b$.
        The function $g$ is chosen as follows.
        For each $i\in[t]$, choose $z_i\gets\bit^n$. 
        Query $\PRC.\Enc(\sk',\cdot)$ on $z_i$ and obtain $y_i\in\bit^m$.
        Using the algorithm from \cref{lem:alg_hash}, sample a function $h:\bit^n\to[t]$ from a $4q$-wise independent function family.
        Then, define $g$ by $g(x)\coloneqq y_{h(x)}$.
        Finally, run $b\gets\cA^g$, and output $b$.
        
        \item 
        ${\rm H'}_1$: this is the same as ${\rm H'}_0$ except for sampling $g:\bit^n\to\bit^m$ as follows.
        {\color{red}Choose a random function $h:\bit^n\to[t]$,} and define $g(x)\coloneqq y_{h(x)}$.
        
        \item 
        ${\rm H'}_2$: this is the same as ${\rm H'}_1$ except for sampling $g:\bit^n\to\bit^m$ as follows.
        {\color{red}For each $x\in\bit^n$, choose $z_x\gets\bit^n$, and run $y_x\gets\PRC.\Enc(\sk',z_x)$.
        Then, define $g(x)\coloneqq y_x$.}
        
        \item 
        ${\rm H'}_3$: this is the same as ${\rm H'}_2$ except for sampling $g:\bit^n\to\bit^m$ as follows.
        {\color{red}Choose a random function $u:\bit^n\to\bit^n$.}
        For each $x\in\bit^n$, run $y_x\gets\PRC.\Enc(\sk',z_x)$, {\color{red}where $z_x\coloneqq u(x)$.} 
        Then, define $g(x)\coloneqq y_x$.
        
        \item 
        ${\rm H'}_4$: this is the same as ${\rm H'}_3$ except for sampling $g:\bit^n\to\bit^m$ as follows.
        {\color{red}Choose a random permutation $\pi:\bit^n\to\bit^n$.}
        For each $x\in\bit^n$, run $y_x\gets\PRC.\Enc(\sk',z_x)$, {\color{red}where $z_x\coloneqq \pi(x)$.} 
        Then, define $g(x)\coloneqq y_x$.

        \item 
        ${\rm H'}_5$: this is the same as ${\rm H'}_4$ except for sampling $g:\bit^n\to\bit^m$ as follows.
        Choose a random permutation $\pi:\bit^n\to\bit^n$ {\color{red}and a bit string $r\gets\bit^\ell$.}
        For each $x\in\bit^n$, {\color{red}choose $r_x\gets\bit^\ell$ and set $y_x\coloneqq\PRC.\Enc(\sk',z_x;r_x\oplus r)$,} where $z_x\coloneqq \pi(x)$.
        Then, define $g(x)\coloneqq y_x$.
        
        \item 
        ${\rm H'}_6$: this is the same as ${\rm H'}_5$ except for sampling $g:\bit^n\to\bit^m$ as follows.
        Choose a random permutation $\pi:\bit^n\to\bit^n$, {\color{red} a random function $f:\bit^n\to\bit^\ell$,} and a bit string $r\gets\bit^\ell$.
        Then, define $g(x)\coloneqq \PRC.\Enc(\sk',z_x;r_x\oplus r)$, where $z_x\coloneqq\pi(x)$ and {\color{red}$r_x\coloneqq f(x)$.}
        
        \item 
        ${\rm H'}_7$: this is the same as ${\rm H'}_6$ except for sampling $g:\bit^n\to\bit^m$ as follows.
        {\color{red}Choose $k,k'\gets\bit^\secp$} and a bit string $r\gets\bit^\ell$.
        Then, for each $x\in\bit^n$, define $g(x)\coloneqq \PRC.\Enc(\sk',z_x;r_x\oplus r)$, {\color{red}where $z_x\coloneqq\pi_{k}(x)$ and $r_x\coloneqq f_{k'}(x)$.}
    \end{itemize}
    Note that $\rm{H'}_{0}={\rm H}_0$, and $\rm{H'}_{7}={\rm H}_1$.
    We show that ${\rm H'}_{j-1}$ is statistically close to ${\rm H'}_j$ for each $j\in[7]$ as follows.
    \begin{itemize}
        \item 
        ${\rm H'}_0$ and ${\rm H'}_1$: 
         The difference between them lies in whether $h:\bit^n\to[t]$ is sampled from a uniformly random or a $4q$-wise independent function family.
         To simulate one query $\ket{x}\ket{w}\mapsto\ket{x}\ket{w\oplus y_{h(x)}}$, an uncomputation step is required, so $\cA^g$ is regarded as an algorithm making $2q$ queries to $h$. Thus, we have
        \begin{align}
            \Pr[1\gets{\rm H'}_0]=\Pr[1\gets{\rm H'}_1]\label{eq:PRFC:H'0_vs_H'1}
        \end{align}
        from \cref{lem:sim_with_hash} with $t=2q$.

        \item 
        ${\rm H'}_1$ and ${\rm H'}_2$:
        For each $\sk'$, define $\cD_{\sk'}$ as the following distribution over $\bit^m$.
         \begin{enumerate}
             \item Choose a $z\gets\bit^n$, and run $\PRC.\Enc(\sk',z)$ to get $y\in\bit^m$.
             \item Output $y$.
         \end{enumerate}
        Let us fix $\sk'$.
        Since $h:\bit^n\to[t]$ is a uniformly random function, in ${\rm H'}_1$, a function $g:\bit^n\to\bit^m$ is sampled from ${\rm{SR}}_t^{\cD_{\sk'}}(\bit^n)$, where ${\rm{SR}}_t^{\cD_{\sk'}}(\bit^n)$ is the distribution defined in \cref{lem:small_range_oracle}.
        On the other hand, in ${\rm H'}_2$, $g:\bit^n\to\bit^m$ is sampled from $\cD_{\sk'}^{\bit^n}$, where $\cD_{\sk'}^{\bit^n}$ is the distribution defined in \cref{lem:small_range_oracle}.
        Since $\cA$ is a $q$-query quantum algorithm, we obtain
        \begin{align}
            |\Pr[1\gets{\rm H'}_1]-\Pr[1\gets{\rm H'}_2]|\le\frac{27q^3}{t}\label{eq:PRFC:H'1_vs_H'2}
        \end{align}
        by applying \cref{lem:small_range_oracle}.

        \item 
        ${\rm H'}_2$ and ${\rm H'}_3$:
        The difference between them lies in whether, for each $x\in\bit^n$, $z_x\in\bit^n$ is chosen uniformly at random or determined as $u(x)$ using a uniformly random function $u:\bit^n\to\bit^n$.
        Since they give the same distribution, we have
        \begin{align}
            \Pr[1\gets{\rm H'}_2]=\Pr[1\gets{\rm H'}_3].\label{eq:PRFC:H'2_vs_H'3}
        \end{align}

        \item 
        ${\rm H'}_3$ and ${\rm H'}_4$:
        The difference between them lies in whether, for each $x\in\bit^n$, $z_x\in\bit^n$ is determined using a uniformly random function $u:\bit^n\to\bit^n$ or a uniformly random permutation $\pi:\bit^n\to\bit^n$.
        Since $\cA$ is a $q$-query quantum algorithm, for some constant $c>0$,
        \begin{align}
            |\Pr[1\gets{\rm H'}_3]-\Pr[1\gets{\rm H'}_4]|\le\frac{cq^3}{2^n}\label{eq:PRFC:H'3_vs_H'4}
        \end{align}
        by applying \cref{lem:function_vs_permutation}.

        \item 
        ${\rm H'}_4$ and ${\rm H'}_5$:
        From their definitions, ${\rm H'}_4$ is equal to ${\rm H'}_5$.
        Thus, we have
        \begin{align}
            \Pr[1\gets{\rm H'}_4]=\Pr[1\gets{\rm H'}_5].\label{eq:PRFC:H'4_vs_H'5}
        \end{align}

        \item 
        ${\rm H'}_5$ and ${\rm H'}_6$:
        The difference between them lies in whether, for each $x\in\bit^n$, an internal random string $r_x\in\bit^\ell$ is chosen uniformly at random or determined as $f(x)$ using a uniformly random function $f:\bit^n\to\bit^\ell$.
        Since they give the same distribution, we have
        \begin{align}
            \Pr[1\gets{\rm H'}_5]=\Pr[1\gets{\rm H'}_6].\label{eq:PRFC:H'5_vs_H'6}
        \end{align}

        \item
        ${\rm H'}_6$ and ${\rm H'}_7$:
        The difference between them lies in whether, to define $g:\bit^n\to\bit^m$, the permutation and the function are sampled uniformly at random or from a PRP and a PRF.
        Thus, we have
        \begin{align}
            |\Pr[1\gets{\rm H'}_6]-\Pr[1\gets{\rm H'}_7]|\le\negl(\secp)\label{eq:PRFC:H'6_vs_H'7}
        \end{align}
        from the security of PRPs and PRFs.
    \end{itemize}
    From \cref{eq:PRFC:H'0_vs_H'1,eq:PRFC:H'1_vs_H'2,eq:PRFC:H'2_vs_H'3,eq:PRFC:H'3_vs_H'4,eq:PRFC:H'4_vs_H'5,eq:PRFC:H'5_vs_H'6,eq:PRFC:H'6_vs_H'7}, we have
    \begin{align}
        |\Pr[1\gets{\rm H'}_0]-\Pr[1\gets{\rm H'}_7]|\le\frac{27q^3}{t}+\frac{cq^3}{2^n}+\negl(\secp).
    \end{align}
    Since $\rm{H'}_{0}={\rm H}_0$ and $\rm{H'}_{7}={\rm H}_1$, this concludes the proof.
\end{proof}

Next, we give the proof of \cref{claim:PRFC:H2_to_H3}.

\begin{proof}[Proof of \cref{claim:PRFC:H2_to_H3}]
    Let us consider the following hybrids of distributions over $\bit$.
    In the following, the {\color{red}red} text is the new step in the current hybrid.
    \begin{itemize}
        \item ${\rm H''}_0$: 
        Choose a random function $g:\bit^n\to\bit^m$, and run $\cA^g$ to get $b\in\bit$.
        Output $b$.
        
        \item ${\rm H''}_1$: 
        This is the same as ${\rm H''}_0$ except for sampling $g:\bit^n\to\bit^m$ as follows.
        {\color{red}
        For each $i\in[t]$, choose $y_i\gets\bit^m$.
        For each $x\in\bit^n$, choose $i_x\gets[t]$ and define $g(x)\coloneqq y_{i_x}$.
        }
        
        \item ${\rm H''}_2$: 
        This is the same as ${\rm H''}_1$ except for sampling $g:\bit^n\to\bit^m$ as follows.
        For each $i\in[t]$, choose a $y_i\gets\bit^m$.
        {\color{red}Choose a random function $h:\bit^n\to[t]$.
        Then, define $g(x)\coloneqq y_{h(x)}$.}
        
        \item ${\rm H''}_3$: 
        This is the same as ${\rm H''}_2$ except for sampling $g:\bit^n\to\bit^m$ as follows.
        For each $i\in[t]$, choose a $y_i\gets\bit^m$.
        {\color{red}Choose a function $h:\bit^n\to[t]$ from a $4q$-wise independent function family.}
        Then, define $g(x)\coloneqq y_{h(x)}$.

        \item ${\rm H''}_4$: 
        This is the same as ${\rm H''}_3$ except for sampling $g:\bit^n\to\bit^m$ as follows.
        For each $i\in[t]$, {\color{red}choose $z_i\gets\bit^n$ and query an oracle $\cO$ on $z_i$ to obtain $y_i\in\bit^m$, where $\cO$ outputs a uniformly random $m$-bit string regardless of the input.}
        Choose a function $h:\bit^n\to[t]$ from a $4q$-wise independent function family.
        Then, define $g(x)\coloneqq y_{h(x)}$.
    \end{itemize}
    Note that $\rm{H''}_{0}={\rm H}_2$, and $\rm{H''}_{4}={\rm H}_3$.
    We show that ${\rm H''}_{j-1}$ is statistically close to ${\rm H''}_j$ for each $j\in[4]$ as follows.
    \begin{itemize}
        \item ${\rm H''}_0$ and ${\rm H''}_1$:
        Let $\cD$ be the uniform distribution over $\bit^m$.
        Since, in ${\rm H''}_0$, $g:\bit^n\to\bit^m$ is a uniformly random function, it can be seen as sampled from $\cD^{\bit^n}$, where $\cD^{\bit^n}$ is defined in \cref{lem:small_range_oracle}.
        On the other hand, in ${\rm H''}_1$, $g:\bit^n\to\bit^m$ is sampled from ${\rm{SR}}_t^{\cD}(\bit^n)$, where ${\rm{SR}}_t^{\cD}(\bit^n)$ is defined in \cref{lem:small_range_oracle}.
        Since $\cA$ is a $q$-query quantum algorithm, we obtain 
        \begin{align}
            |\Pr[1\gets{\rm H''}_0]-\Pr[1\gets{\rm H''}_1]|\le\frac{27q^3}{t}\label{eq:PRFC:H''0_vs_H''1}
        \end{align}
        by applying \cref{lem:small_range_oracle}.

        \item ${\rm H''}_1$ and ${\rm H''}_2$:
        The difference between them lies in whether, for each $x\in\bit^n$, $i_x\in[t]$ is chosen uniformly at random or determined as $h(x)$ using a uniformly random function $h:\bit^n\to[t]$.
        Since they give the same distribution, we have
        \begin{align}
            \Pr[1\gets{\rm H''}_1]=\Pr[1\gets{\rm H''}_2].\label{eq:PRFC:H''1_vs_H''2}
        \end{align}

        \item ${\rm H''}_2$ and ${\rm H''}_3$:
        The difference between them lies in whether $h:\bit^n\to[t]$ is sampled from a uniformly random or a $4q$-wise independent function family.
        The same $4q$ requirement applies here because this hybrid uses the same coherent hash simulation as above, in which $\cA$ is regarded as a $2q$-query algorithm.
        Since $\cA$ is regarded as a $2q$-query quantum algorithm for this simulation, we have
        \begin{align}
            \Pr[1\gets{\rm H''}_2]=\Pr[1\gets{\rm H''}_3]\label{eq:PRFC:H''2_vs_H''3}
        \end{align}
        from \cref{lem:sim_with_hash}.

        \item ${\rm H''}_3$ and ${\rm H''}_4$:
        The difference between them lies in whether, for each $i\in[t]$, $y_i\in\bit^m$ is sampled uniformly at random or obtained by querying the oracle $\cO$ on a uniformly random $z_i$.
        They give the same distribution since $\cO$ outputs a uniformly random $m$-bit string, regardless of the query.
        Thus, we have
        \begin{align}
            \Pr[1\gets{\rm H''}_3]=\Pr[1\gets{\rm H''}_4].\label{eq:PRFC:H''3_vs_H''4}
        \end{align}
    \end{itemize}

    From \cref{eq:PRFC:H''0_vs_H''1,eq:PRFC:H''1_vs_H''2,eq:PRFC:H''2_vs_H''3,eq:PRFC:H''3_vs_H''4}, we have
    \begin{align}
        |\Pr[1\gets{\rm H''}_0]-\Pr[1\gets{\rm H''}_4]|\le\frac{27q^3}{t}.
    \end{align}
    Since $\rm{H''}_{0}={\rm H}_2$ and $\rm{H''}_{4}={\rm H}_3$, this concludes the proof.
\end{proof}

\section{Construction of PRICs}
\label{sec:PRIC_const}

In this section, we construct PRICs robust to $w$-local quantum noise for any function $w:\N\to\N$ satisfying $w=o(n_\phys\log\log n_\phys/\log n_\phys)$, using the graph-sampling and recovery algorithms below.
For the construction, we need the following graph-sampling and recovery theorem.

\begin{restatable}[Graph sampling and phase recovery]{theorem}{SampleRecover}
\label{thm:better-graph}
\label{thm:GraphSample_and_Recover}
Fix a constant $p>0$.
Let $t:\mathbb N\to\mathbb N$ be a function satisfying
\begin{align}
 t(n)=o\left(\frac{n\log\log n}{\log n}\right).
 \label{eq:better-radius-growth}
\end{align}
There exists a pair of algorithms $(\GraphSample,\Recover)$ with the following properties.
$\GraphSample$ is a PPT algorithm that takes $1^n$ as input and outputs a simple graph $G$ on vertex set $[2n]$ with adjacency matrix $A$.
$\Recover$ is a deterministic polynomial-time algorithm that takes $G$ and $e\in\bit^{2n}$ as input and outputs a $2n$-bit string.
For all sufficiently large $n$ and all $u,v\in\bit^{2n}$ satisfying $|\Support(u)\cup\Support(v)|\le t(n)$, the following hold.
\begin{itemize}
\item[\textnormal{(i) Induced-weight bound.}] Every output graph satisfies
\begin{align}
 \weight(v\oplus Au)\le pn.
 \label{eq:better-pointwise-weight}
\end{align}
\item[\textnormal{(ii) Recovery correctness.}]
\begin{align}
 \Pr_{G\gets\GraphSample(1^n)}
 [\Recover(G,v\oplus Au)=u]
  \ge1-\negl(n),
 \label{eq:better-pointwise-recovery}
\end{align}
\end{itemize}
\end{restatable}

The proof is given in \cref{sec:better_graph}.
Using the algorithms $\GraphSample$ and $\Recover$ from \cref{thm:GraphSample_and_Recover}, we construct PRICs from PRFCs, PRFs, and PRPs.
    In the following, fix constants $0<c<1/2$ and $0<p<1$, let $n_\logic,\ell,r,m:\N\to\N$ be polynomials, and define a polynomial $n_\phys\coloneqq n_\logic+\ell+r+m$.
We use $\secp$ as the security parameter and omit the $\secp$-dependence of $n_\logic,\ell,r,m,$ and $n_\phys$ for notational simplicity.
We may assume without loss of generality that $n_\phys$ is even for all sufficiently large $\secp$, since this can be ensured by an appropriate choice of $\ell$ and $r$.

Fix constants $0<p<1$ and $C\ge1$ such that $n_\phys\le Cm$ for all sufficiently large $\secp$.
Let $w:\N\to\N$ be a function satisfying $w=o(n_\phys\log\log n_\phys/\log n_\phys)$, and set $p_{\mathrm g}=p/C$.
Choose a function $t:\N\to\N$ such that $t(n)=o(n\log\log n/\log n)$ and $t(n_\phys(\secp)/2)\ge w(\secp)$ for all sufficiently large $\secp$.

\begin{construction}[PRIC]\label{const:PRIC}
    Let $(\PRFC.\KeyGen,\allowbreak \PRFC.\Enc,\PRFC.\Dec)$ be an
    $(n_\logic+\ell,m)$ PRFC robust to any $p$-bounded channel.
    Let $F:\bit^\secp\times\bit^{n_\logic+\ell}\to\bit$ be a PRF.
    Let $\Pi:\bit^\secp\times\bit^{n_\logic+\ell+r}\to\bit^{n_\logic+\ell+r}$ be a PRP.
    Let $\GraphSample$ and $\Recover$ be the algorithms from \cref{thm:better-graph} instantiated with the constant $p_{\mathrm g}$ and the radius function $t$ fixed above. We construct a PRIC encoding $n_\logic$ logical qubits into $n_\phys$ physical qubits and robust to $w$-local noise as follows.
    \begin{itemize}
        \item $\KeyGen:$
        Run $\sk'\gets\PRFC.\KeyGen(1^\secp)$ and choose $k,k'\gets\bit^\secp$.
        Run $G\gets\GraphSample(1^{n_\phys/2})$.
        Choose $P\gets\cP_{n_\phys}$.
        Output $\sk\coloneqq(\sk',k,k',G,P)$.
        
        \item $\Enc:$
        Parse $\sk=(\sk_{\mathrm{in}},G,P)$, where $\sk_{\mathrm{in}}\coloneqq(\sk', k,k')$.
        Let $\regA$, $\regB$, and $\regC$ be $n_\logic$-qubit, $(n_\logic+\ell+r)$-qubit, and $m$-qubit registers, respectively.
        Given $\sk$, $\Enc$ implements the following isometry $\Enc_\sk$:\footnote{For the definition of the unitary $U_G$, see \cref{def:CWS_encoding_unitary}.}
        \begin{align}
            \Enc_{\sk,\regA\to\regB\regC}
            \coloneqq
            (PU_{G})_{\regB\regC}\circ V_{\sk_{\mathrm{in}},\regA\to\regB\regC}.
        \end{align}
        Here $V_{\sk_{\mathrm{in}}}$ is the isometry
        \begin{align}
            V_{\sk_{\mathrm{in}}}
            :\ket{x}_\regA\mapsto
             \frac{1}{\sqrt{2^\ell}}\sum_{y\in\bit^\ell}(-1)^{f_{k}(x\|y)}\ket{\pi_{k'}(x\|y\|0^r)}_\regB\ket{g_{\sk'}(x\|y)}_\regC,
        \end{align}
        where $f_{k}(\cdot)\coloneqq F(k,\cdot)$, $\pi_{k'}\coloneqq \Pi(k',\cdot)$, and $g_{\sk'}(\cdot)\coloneqq\PRFC.\Enc(\sk',\cdot)$.
        $V_{\sk_{\mathrm{in}}}$ is implemented as follows:
        \begin{align}
            \ket{x}_\regA
            \XMAPSTO{}&
            \ket{x\|0^{\ell}\|0^r}_\regB\ket{0^m}_\regC
            \\
            \XMAPSTO{H}&
            \frac{1}{\sqrt{2^\ell}}\sum_{y\in\bit^\ell}\ket{x\|y\|0^{r}}_\regB\ket{0^m}_\regC
            \\
            \XMAPSTO{f_{k}}&
            \frac{1}{\sqrt{2^\ell}}\sum_{y\in\bit^\ell}(-1)^{f_{k}(x\|y)}\ket{x\|y\|0^{r}}_\regB\ket{0^m}_\regC
            \\
            \XMAPSTO{g_{\sk'}}&
            \frac{1}{\sqrt{2^\ell}}\sum_{y\in\bit^\ell}(-1)^{f_{k}(x\|y)}\ket{x\|y\|0^{r}}_\regB\ket{g_{\sk'}(x\|y)}_\regC\\
            \XMAPSTO{\pi_{k'}}&
            \frac{1}{\sqrt{2^\ell}}\sum_{y\in\bit^\ell}(-1)^{f_{k}(x\|y)}\ket{\pi_{k'}(x\|y\|0^r)}_\regB\ket{g_{\sk'}(x\|y)}_\regC.
        \end{align}
        
        \item $\Dec:$
        Parse $\sk=(\sk_{\mathrm{in}},G,P)$, where $\sk_{\mathrm{in}}\coloneqq(\sk', k,k')$.
        Given an input $\rho_{\regB\regC}$, $\Dec(\sk,\cdot)$ applies the following operations.
        \begin{enumerate}
            \item Apply $U_G^\dag\cdot P^\dag$ to $\regB\regC$.
            \item 
            Apply the isometry $\BitDec_{\sk_{\mathrm{in}},\regB\regC\to\regB\regC\regD}$ which is implemented as follows.
            \begin{enumerate}
                \item Apply the following isometry to the register $\regC$ using $\sk'$:
                \begin{align}
                    \ket{c}_\regC\mapsto\ket{c}_\regC\ket{\PRFC.\Dec(\sk',c)}_\regD.
                    \label{eq:BitDec_step1}
                \end{align}

                \item 
                Apply the following unitary to $\regB\regC\regD$ using $\sk'$, $k$, and $k'$:
                \begin{align}
                    \ket{b}_\regB\ket{c}_\regC\ket{z}_\regD\mapsto
                    (-1)^{f_{k}(z)}\ket{b\oplus \pi_{k'}(z\|0^r)}_\regB\ket{c\oplus g_{\sk'}(z)}_\regC\ket{z}_\regD.
                    \label{eq:BitDec_step2}
                \end{align}
            \end{enumerate}
            
            \item 
            Apply the CPTP map $\PhaseDec_{\sk_{\mathrm{in}},G,\regB\regC\regD\to\regA}$ which is implemented as follows.
            \begin{enumerate}

               \item 
               Measure $\regB\regC$ in the computational basis to get $e\in\bit^{n_\phys}$.

               \item 
               Run $\Recover(G,e)$ to get $u\in\bit^{n_\phys}$.

               \item 
               Apply the following unitary to the register $\regD$ using $\sk$, $k'$, and $u$.
                \begin{align}
                    \ket{z}_\regD\mapsto (-1)^{u\cdot(\pi_{k'}(z\|0^r)\| g_{\sk'}(z))}\ket{z}_\regD.
                    \label{eq:PhaseDec_step3}
                \end{align}

                \item 
                Output the first $n_\logic$ qubits of $\regD$ as the $n_\logic$-qubit register $\regA$.
            \end{enumerate}
        \end{enumerate}
    \end{itemize}
\end{construction}

We prove the following theorem.

\begin{theorem}\label{thm:PRIC}
    Let $w:\N\to\N$ be a function satisfying $w=o(n_\phys\log\log n_\phys/\log n_\phys)$.
    If $n_\phys(\secp)\ge\secp^\kappa$ for some constant $\kappa>0$, $r\ge n_\logic+\ell+\omega(\log\secp)$, $\ell\ge\omega(\log\secp)$, and $n_\phys=O(m)$, then \cref{const:PRIC}, with the graph parameters chosen above, is a PRIC robust to any $w$-local noise.
\end{theorem}

Before proving \cref{thm:PRIC}, we mention some corollaries.
Since PRFCs imply PRFs and PRPs by \cref{thm:PRFC_imply_PRF/PRP} and secret-key PRCs robust to $p$-bounded noise imply PRFCs robust to $p$-bounded noise for any constant $0<p<1/2$ by \cref{thm:PRFC}, we immediately obtain the following corollary.

\begin{corollary}\label{coro:PRIC_from_SKPRC}
    Let $n_\logic,\ell,r,m:\N\to\N$ be polynomials, and set $n_\phys=n_\logic+\ell+r+m$.
    Assume that $n_\phys\ge\secp^\kappa$ for some constant $\kappa>0$, $r\ge n_\logic+\ell+\omega(\log\secp)$, $\ell\ge\omega(\log\secp)$, and $n_\phys=O(m)$.
    If $(n_\logic+\ell,m)$ secret-key PRCs robust to $p$-bounded noise exist for some constant $0<p<1$, then, for every function $w:\N\to\N$ satisfying $w=o(n_\phys\log\log n_\phys/\log n_\phys)$, there exist PRICs encoding $n_\logic$ logical qubits into $n_\phys$ physical qubits and robust to any $w$-local noise.
\end{corollary}

In addition, by \cref{thm:post-quantum_PRC}, \cref{assumption:for_PRC} implies the existence of $(n,O(n))$ secret-key PRCs robust to $p$-bounded noise for any $0<p<1/2$.
Thus, we also obtain the following corollary.

\begin{corollary}\label{coro:PRIC_from_assumption}
    There exists a constant $\gamma>0$ such that the following holds.
    Let $n_\logic,n_\phys:\N\to\N$ be polynomials satisfying $n_\phys\ge\gamma n_\logic$ and $n_\phys(\secp)\ge\secp^\kappa$ for some constant $\kappa>0$.
    If \cref{assumption:for_PRC} holds, then, for every function $w:\N\to\N$ satisfying $w=o(n_\phys\log\log n_\phys/\log n_\phys)$, there exist PRICs encoding $n_\logic$ logical qubits into $n_\phys$ physical qubits and robust to any $w$-local noise.
\end{corollary}

To prove \cref{thm:PRIC}, we prove the robustness and pseudorandomness of \cref{const:PRIC} in \cref{subsec:PRIC_robustness,subsec:PRIC_pseudorandomness}, respectively.

\subsection{Robustness}
\label{subsec:PRIC_robustness}

In this subsection, we prove that \cref{const:PRIC} satisfies robustness.

\begin{lemma}\label{lem:robustness_of_PRIC}
    Let $w:\N\to\N$ be a function satisfying $w=o(n_\phys\log\log n_\phys/\log n_\phys)$.
    If $n_\phys(\secp)\ge\secp^\kappa$ for some constant $\kappa>0$ and $n_\phys=O(m)$, then \cref{const:PRIC}, with the graph parameters chosen above, satisfies robustness for any $w$-local error channel.
\end{lemma}

\begin{proof}[Proof of \cref{lem:robustness_of_PRIC}]
        By the definition of $w$-local noise, choose an exact-local channel
        $\widetilde{\cE}_\secp$ with
        $\|\cE_\secp-\widetilde{\cE}_\secp\|_\diamond\le\negl(\secp)$.
        We prove the claim for $\widetilde{\cE}_\secp$ and write it as
        $\cE_\secp$ below; contractivity adds the approximation error back at
        the end.
        Let $\regA$ and $\regA'$ be $n_\logic$-qubit registers.
        Our goal is to show
        \begin{align}
            \frac{1}{2}
            \bigg\|
             \bigg(
              \identitymap_{\regA'}\otimes
              \Exp_{\sk\gets\KeyGen(1^\secp)}(\Dec_\sk\circ\cE_\secp\circ\Enc_{\sk})_\regA
             \bigg)
             (\ketbra{\psi}{\psi}_{\regA'\regA})
             -\ketbra{\psi}{\psi}_{\regA'\regA}
            \bigg\|_1
            \le\negl(\secp)
        \end{align}
        for any $2n_\logic$-qubit pure state $\ket{\psi}$ and any $w$-local error channel $\cE_\secp$.
        In \cref{const:PRIC}, the secret key $\sk$ is parsed as $\sk = (P, G, \sk_{\mathrm{in}})$, where $P \gets \cP_{n_\phys}$ is a uniformly random Pauli for twirling, $G \gets \GraphSample(1^{n_\phys/2})$ is a random bipartite graph for the underlying CWS code, and $\sk_{\mathrm{in}} = (\sk', k, k')$ is the inner key.
        We decompose $\Enc_\sk$ as follows:
        \begin{align}
            \Enc_{\sk,\regA\to\regB\regC}= (PU_{G})_{\regB\regC}\circ \Enc_{\sk_{\mathrm{in}},\regA\to\regB\regC},
        \end{align}
        where
        \begin{align}
            \Enc_{\sk_{\mathrm{in}},\regA\to\regB\regC}:
            \ket{x}_\regA
            \mapsto
            \frac{1}{\sqrt{2^\ell}}\sum_{y\in\bit^\ell}(-1)^{f_{k}(x\|y)}\ket{\pi_{k'}(x\|y\|0^r)}_\regB\ket{g_{\sk'}(x\|y)}_\regC.
        \end{align}
        We also decompose $\Dec_\sk$ as follows:
        \begin{align}
            \Dec_{\sk,\regB\regC\to\regA}= \PhaseDec_{\sk_{\mathrm{in}},G,\regB\regC\regD\to\regA} \circ \BitDec_{\sk_{\mathrm{in}},\regB\regC\to\regB\regC\regD}\circ (PU_{G})_{\regB\regC}^\dag
        \end{align}
        
        By taking the expectation over the random Pauli operator
        $P \in \cP_{n_\phys}$, the arbitrary $w$-local error channel
        $\cE_\secp$ is twirled into a probabilistic mixture of $w$-local Pauli
        errors from \cref{lem:Pauli_twirling_for_local_noise}:
        \begin{align}
            \Exp_{P\gets\cP_{n_\phys}}
            \left[
             P^\dagger \cE_\secp (P \rho P^\dagger) P
            \right]
            = \sum_{E \in \cP_{n_\phys} : \weight(E) \le w} p_E E \rho E^\dagger.
        \end{align}
        Thus, under the Pauli-noise mixture, the expectation over $\sk$ reduces to the expectation over $\sk_{\mathrm{in}}$ and $G$.
        By convexity of the trace distance, it suffices to establish the following bound, uniformly over all fixed Pauli errors $E=Z^vX^u$ with $\weight(E)\le w$.
        \begin{align}
            &\frac{1}{2}
            \bigg\|
             \bigg(
              \identitymap_{\regA'}\otimes
              \Exp_{\sk_{\mathrm{in}},G}(\PhaseDec_{(\sk_{\mathrm{in}},G)} \circ \BitDec_{\sk_{\mathrm{in}}}\circ (U_G^\dag E U_G) \circ\Enc_{\sk_{\mathrm{in}}})_\regA
             \bigg)
             (\ketbra{\psi}{\psi}_{\regA'\regA})
             -\ketbra{\psi}{\psi}_{\regA'\regA}
            \bigg\|_1
            \notag\\
            \le&\negl(\secp).
            \label{eq:goal_of_robustness}
        \end{align}
        In the following, fix a $2n_\logic$-qubit pure state $\ket{\psi}$ and a Pauli error $E = Z^v X^u$ satisfying $\weight(E)\le w$.
        To show \cref{eq:goal_of_robustness}, we need the following claim.
        \begin{claim}\label{claim:after_1st_step_in_Dec_PRIC}
            For each $G$, define $e \coloneqq v\oplus Au$, where $A$ denotes the adjacency matrix of $G$, and the projection
            \begin{align}
                    \Pi^{\good}_{\regB\regC\regD}\coloneqq
                    &
                    \ketbra{e}{e}_{\regB\regC}\otimes
                    \sum_{x\in\bit^{n_\logic},y\in\bit^\ell}\ketbra{(x\|y)}{(x\|y)}_\regD.
            \end{align}
            For each $(\sk_{\mathrm{in}}, G)$, define the states. If the
            projection below has nonzero norm, normalize it; if its norm is
            zero, define $\ket{\phi_{(\sk_{\mathrm{in}},G)}^{\good}}$ to be a
            fixed computational-basis state of the same register dimension.
            \begin{align}
                    \ket{\phi_{(\sk_{\mathrm{in}}, G)}}_{\regA'\regB\regC\regD}
                     \coloneqq
                     \BitDec_{\sk_{\mathrm{in}}, \regB\regC \to \regB\regC\regD} \cdot 
                     (U_G^\dag Z^vX^u U_G)_{\regB\regC} \cdot
                     \Enc_{\sk_{\mathrm{in}}, \regA \to \regB\regC} \cdot
                     \ket{\psi}_{\regA'\regA}
            \end{align}
            and
            \begin{align}
                    \ket{\phi_{(\sk_{\mathrm{in}}, G)}^{\good}}_{\regA'\regB\regC\regD}
                    \coloneqq
                    \frac{\Pi^{\good}_{\regB\regC\regD}\ket{\phi_{(\sk_{\mathrm{in}}, G)}}_{\regA'\regB\regC\regD}}{
                    \left\|
                    \Pi^{\good}_{\regB\regC\regD}\ket{\phi_{(\sk_{\mathrm{in}}, G)}}_{\regA'\regB\regC\regD}
                    \right\|_2}.
            \end{align}
            Then, there exists a negligible function $\epsilon$ such that
            \begin{align}
                \Pr_{(\sk_{\mathrm{in}}, G)}
                \bigg[
                 \bigg\|
                  \ket{\phi_{(\sk_{\mathrm{in}}, G)}}-\ket{\phi_{(\sk_{\mathrm{in}}, G)}^{\good}}
                 \bigg\|_2\le\epsilon(\secp)
                 \wedge
                 \Recover(G,e)=u
                \bigg]
                \ge1-\negl(\secp).
            \end{align}
        \end{claim}

        We prove \cref{claim:after_1st_step_in_Dec_PRIC} below.
        Let $\epsilon$ be the negligible function in \cref{claim:after_1st_step_in_Dec_PRIC}.
        Define $\cK_\good$ to be the following set of keys:
        \begin{align}
            \cK_\good\coloneqq
            \bigg\{(\sk_{\mathrm{in}}, G):
            \bigg\|
             \ket{\phi_{(\sk_{\mathrm{in}}, G)}}-
             \ket{\phi_{(\sk_{\mathrm{in}}, G)}^{\good}}
            \bigg\|_2\le\epsilon(\secp)
            \wedge
            \Recover(G,e)=u
            \bigg\}.
        \end{align}
        We will show that for any $(\sk_{\mathrm{in}}, G)\in\cK_\good$,
        \begin{align}
            &\frac{1}{2}
            \bigg\|
             \bigg(
              \identitymap_{\regA'}\otimes
              (\PhaseDec_{(\sk_{\mathrm{in}},G)} \circ \BitDec_{\sk_{\mathrm{in}}}\circ (U_G^\dag E U_G) \circ\Enc_{\sk_{\mathrm{in}}})_\regA
             \bigg)
             (\ketbra{\psi}{\psi}_{\regA'\regA})
             -\ketbra{\psi}{\psi}_{\regA'\regA}
            \bigg\|_1
            \notag\\
            \le&\negl(\secp)
            \label{eq:closeness_for_good_sk}
        \end{align}
        which implies \cref{eq:goal_of_robustness} from \cref{claim:after_1st_step_in_Dec_PRIC}.

        Now we show \cref{eq:closeness_for_good_sk}.
        In the following, fix a key $(\sk_{\mathrm{in}}, G) \in \cK_\good$. We use a hybrid argument.
        Define the following states:
        \begin{align}
            &\rho^{(0)}_{\regA'\regA}\coloneqq
                \bigg(
                 \identitymap_{\regA'}\otimes
                 (\PhaseDec_{(\sk_{\mathrm{in}},G)} \circ \BitDec_{\sk_{\mathrm{in}}}\circ (U_G^\dag E U_G) \circ\Enc_{\sk_{\mathrm{in}}})_\regA
                \bigg)
                (\ketbra{\psi}{\psi}_{\regA'\regA}),
            \\
            &\rho^{(1)}_{\regA'\regA}\coloneqq
                \bigg(
                 \identitymap_{\regA'}\otimes
                 \PhaseDec_{(\sk_{\mathrm{in}},G),\regB\regC\regD\to\regA}
                \bigg)
                (\ketbra{\phi_{(\sk_{\mathrm{in}}, G)}}{\phi_{(\sk_{\mathrm{in}}, G)}}_{\regA'\regB\regC\regD}),
            \\
            &\rho^{(2)}_{\regA'\regA}\coloneqq
                \bigg(
                 \identitymap_{\regA'}\otimes
                 \PhaseDec_{(\sk_{\mathrm{in}},G),\regB\regC\regD\to\regA}
                \bigg)
                (\ketbra{\phi^{\good}_{(\sk_{\mathrm{in}}, G)}}{\phi^{\good}_{(\sk_{\mathrm{in}}, G)}}_{\regA'\regB\regC\regD}),
            \\
            &\rho^{(3)}_{\regA\regA'}\coloneqq
            \ketbra{\psi}{\psi}_{\regA'\regA}
        \end{align}
        For each $i\in[2]$, we show that $\rho^{(i-1)}$ is statistically close to $\rho^{(i)}$ as follows.

        \begin{itemize}
            \item $\rho^{(0)}$ and $\rho^{(1)}$:
            By the definition of $\ket{\phi_{(\sk_{\mathrm{in}},G)}}$, we have $\rho^{(0)}=\rho^{(1)}$.

            \item $\rho^{(1)}$ and $\rho^{(2)}$:
            Since $(\sk_{\mathrm{in}}, G)\in\cK_\good$, we have $\|\ket{\phi_{(\sk_{\mathrm{in}}, G)}}-\ket{\phi^{\good}_{(\sk_{\mathrm{in}}, G)}}\|_2\le\epsilon(\secp)\le\negl(\secp)$. 
            Thus, we obtain $\frac{1}{2}\|\rho^{(1)}-\rho^{(2)}\|_1 \le \negl(\secp)$.
        \end{itemize}
        
        For $\rho^{(2)}$ and $\rho^{(3)}$, we prove the following bound below:
        \begin{align}
           \frac{1}{2}\|\rho^{(2)}-\rho^{(3)}\|_1\le\negl(\secp).
            \label{eq:rho2_and_rho3}
        \end{align}
        From the above and the triangle inequality, we obtain \cref{eq:closeness_for_good_sk}, which concludes the main proof.
    \end{proof}

To complete the proof, we show \cref{claim:after_1st_step_in_Dec_PRIC,eq:rho2_and_rho3}.
For the proofs, we use the following claim.

\begin{claim}\label{claim:good_projection_after_BitDec}
    For any $x\in\bit^{n_\logic}$,
    \begin{align}
        &\Pi^\good_{\regB\regC\regD}\cdot
            \BitDec_{\sk_{\mathrm{in}},\regB\regC\to\regB\regC\regD}\cdot
            (U_G^\dag Z^vX^u U_G)\cdot
            \Enc_{\sk_{\mathrm{in}},\regA\to\regB\regC}\cdot
            \ket{x}_\regA
        \notag\\
        =&\ket{e}_{\regB\regC}
            \otimes
            \bigg(
             \frac{1}{\sqrt{2^\ell}}\sum_{y:\PRFC.\Dec(\sk', g_{\sk'}(x\|y) \oplus e_1)=x\|y}
            (-1)^{u\cdot c(x\|y)}\ket{x\|y}_\regD
            \bigg),
    \end{align}
    where $c(z)\coloneqq \pi_{k'}(z\|0^r) \| g_{\sk'}(z)$ and $e_1$ denotes the last $m$ bits of $e=v\oplus Au$.
\end{claim}

\begin{proof}[Proof of \cref{claim:good_projection_after_BitDec}]
    A straightforward calculation gives
        \begin{align}
            (U_G^\dag Z^vX^u U_G)_{\regB\regC}\cdot
            \Enc_{\sk_{\mathrm{in}},\regA\to\regB\regC}\ket{x}_\regA
            =&
            (X^{e}Z^u)_{\regB\regC}\cdot
            \Enc_{\sk_{\mathrm{in}},\regA\to\regB\regC}\ket{x}_\regA
            \tag{From \cref{lem:error_propagation}}\\
            =&
             \frac{1}{\sqrt{2^\ell}} \sum_{y} (-1)^{f_{k}(x\|y) + u \cdot c(x\| y)}
            \ket{c(x\| y)\oplus e}_{\regB\regC}.
        \end{align}
    Decompose $\BitDec_{\sk_{\mathrm{in}},\regB\regC\to\regB\regC\regD}=\BitDec_{\sk_{\mathrm{in}},\regB\regC\regD}^{(b)}\circ\BitDec_{\sk_{\mathrm{in}},\regC\to\regC\regD}^{(a)}$, where $\BitDec_{\sk_{\mathrm{in}}}^{(a)}$ is the isometry defined as
        \begin{align}
            \BitDec_{\sk_{\mathrm{in}},\regC\to\regC\regD}^{(a)}:
            \ket{c}_\regC\mapsto\ket{c}_\regC\ket{\PRFC.\Dec(\sk',c)}_\regD,
        \end{align}
        and $\BitDec_{\sk_{\mathrm{in}}}^{(b)}$ is the unitary defined as
        \begin{align}
            \BitDec_{\sk_{\mathrm{in}},\regB\regC\regD}^{(b)}:
            \ket{b}_\regB\ket{c}_\regC\ket{z}_\regD\mapsto
                    (-1)^{f_{k}(z)}\ket{b\oplus \pi_{k'}(z\|0^r)}_\regB\ket{c\oplus g_{\sk'}(z)}_\regC\ket{z}_\regD.
        \end{align}
        Then, we have
        \begin{align}
            &\BitDec_{\sk_{\mathrm{in}},\regC\to\regC\regD}^{(a)} \cdot
            (U_G^\dag Z^vX^u U_G)_{\regB\regC}\cdot
            \Enc_{\sk_{\mathrm{in}},\regA\to\regB\regC}\ket{x}_\regA
            \notag\\
            =&
             \frac{1}{\sqrt{2^\ell}} \sum_{y} (-1)^{f_{k}(x\|y) + u \cdot c(x\| y)}
            \ket{c(x\| y)\oplus e}_{\regB\regC}
            \ket{\PRFC.\Dec(\sk', g_{\sk'}(x\|y)\oplus e_1)}_\regD,
            \label{eq:eq1_in_claim:good_projection_after_BitDec}
        \end{align}
        where $e_1$ denotes the last $m$ bits of $e=v\oplus Au$.
        Since $c(z)=\pi_{k'}(z\|0^r)\|g_{\sk'}(z)$ and $\pi_{k'}$ is a permutation, $c(x\|y)\oplus c(z)\oplus e=e$ if and only if $x\|y=z$.
        Thus, \cref{eq:eq1_in_claim:good_projection_after_BitDec} implies that
        \begin{align}
            &\Pi^\good_{\regB\regC\regD}\cdot
            \BitDec_{\sk_{\mathrm{in}},\regB\regC\to\regB\regC\regD}\cdot
            (U_G^\dag Z^vX^u U_G)\cdot
            \Enc_{\sk_{\mathrm{in}},\regA\to\regB\regC}\cdot
            \ket{x}_\regA
            \notag\\
            =&\ket{e}_{\regB\regC}
            \otimes
            \bigg(
             \frac{1}{\sqrt{2^\ell}}\sum_{y:\PRFC.\Dec(\sk', g_{\sk'}(x\|y) \oplus e_1)=x\|y}
            (-1)^{u\cdot c(x\| y)}\ket{x\|y}_\regD
            \bigg),
        \end{align}
        which concludes the proof.
\end{proof}

We are ready to prove \cref{claim:after_1st_step_in_Dec_PRIC,eq:rho2_and_rho3}.

    \begin{proof}[Proof of \cref{claim:after_1st_step_in_Dec_PRIC}]
        We show that there exists a negligible function $\delta$ such that
        \begin{align}
            \Exp_{\sk_{\mathrm{in}}}
             \bigg[
              \bigg\|
                  (I-\Pi^{\good})_{\regB\regC\regD}\cdot
                  \BitDec_{\sk_{\mathrm{in}},\regB\regC\to\regB\regC\regD}\cdot 
                  (U_G^\dag E U_G)_{\regB\regC}\cdot
                  \Enc_{\sk_{\mathrm{in}},\regA\to\regB\regC}\cdot
                  \ket{\xi}_\regA
              \bigg\|_2^2
             \bigg]
             \le\delta(\secp) 
             \label{eq:goal_of_claim:after_1st_step_in_Dec_PRIC}
        \end{align}
        for any $n_\logic$-qubit state $\ket{\xi}$ and any simple graph $G$ such that $\weight(e)\le pm$.
        Fix a simple graph $G$ satisfying the above condition.
        Before proving \cref{eq:goal_of_claim:after_1st_step_in_Dec_PRIC}, we show \cref{claim:after_1st_step_in_Dec_PRIC} assuming it.
        From \cref{eq:goal_of_claim:after_1st_step_in_Dec_PRIC}, we have
        \begin{align}
            \Exp_{\sk_{\mathrm{in}}}
             \bigg[
              \bigg\|
                  (I-\Pi^{\good})\cdot
                  \ket{\phi_{(\sk_{\mathrm{in}}, G)}}
              \bigg\|_2^2
             \bigg]
             \le\delta(\secp).
        \end{align}
        Thus, applying Markov's inequality (\cref{lem:Markov}), we obtain
        \begin{align}
            \Pr_{\sk_{\mathrm{in}}}
                \bigg[
                 \bigg\|
                  (I-\Pi^{\good})\cdot
                  \ket{\phi_{(\sk_{\mathrm{in}}, G)}}
                 \bigg\|_2^2\le\sqrt{\delta(\secp)}
                \bigg]
                \ge1-\sqrt{\delta(\secp)}.
                \label{eq:good_sk_in_whp0}
        \end{align}
        Fix $\sk_{\mathrm{in}}$ satisfying 
        $\|(I-\Pi^{\good})\cdot\ket{\phi_{(\sk_{\mathrm{in}}, G)}}\|_2^2\le\sqrt{\delta(\secp)}$.
        The triangle inequality implies 
        \begin{align}
            \bigg\|
             \ket{\phi_{(\sk_{\mathrm{in}}, G)}}-\ket{\phi_{(\sk_{\mathrm{in}}, G)}^{\good}}
            \bigg\|_2
            \le&
            \bigg\|\ket{\phi_{(\sk_{\mathrm{in}}, G)}}-\Pi^{\good}\ket{\phi_{(\sk_{\mathrm{in}}, G)}}\bigg\|_2
            +\bigg\|\Pi^{\good}\ket{\phi_{(\sk_{\mathrm{in}}, G)}}-\ket{\phi_{(\sk_{\mathrm{in}}, G)}^{\good}}\bigg\|_2
            \notag\\
            \le&\delta^{1/4}(\secp)
            +\bigg\|\Pi^{\good}\ket{\phi_{(\sk_{\mathrm{in}}, G)}}-\ket{\phi_{(\sk_{\mathrm{in}}, G)}^{\good}}\bigg\|_2.
            \label{eq:phi_vs_phi_good1}
        \end{align}
        From the definition of $\ket{\phi_{(\sk_{\mathrm{in}}, G)}^{\good}}$, the second term in \cref{eq:phi_vs_phi_good1} becomes
        \begin{align}
            \bigg\|\Pi^{\good}\ket{\phi_{(\sk_{\mathrm{in}}, G)}}-\ket{\phi_{(\sk_{\mathrm{in}}, G)}^{\good}}\bigg\|_2
            =&\left(
             \frac{1}{\|\Pi^\good\ket{\phi_{(\sk_{\mathrm{in}}, G)}}\|_2}-1
            \right)\cdot
            \|\Pi^\good\ket{\phi_{(\sk_{\mathrm{in}}, G)}}\|_2
            \notag\\
            =&1-\|\Pi^\good\ket{\phi_{(\sk_{\mathrm{in}}, G)}}\|_2.
            \label{eq:phi_vs_phi_good2}
        \end{align}
        On the other hand, $\|(I-\Pi^{\good})\cdot\ket{\phi_{(\sk_{\mathrm{in}}, G)}}\|_2^2\le\sqrt{\delta(\secp)}$ implies
        \begin{align}
            \|\Pi^\good\ket{\phi_{(\sk_{\mathrm{in}}, G)}}\|_2
            =\sqrt{
            1-
            \bigg\|
            (I-\Pi^{\good})\cdot
            \ket{\phi_{(\sk_{\mathrm{in}}, G)}}
            \bigg\|_2^2
            }
            \ge&\sqrt{1-\sqrt{\delta(\secp)}}
            \ge1-\sqrt{\delta(\secp)}.
            \label{eq:phi_vs_phi_good3}
        \end{align}
        From \cref{eq:phi_vs_phi_good1,eq:phi_vs_phi_good2,eq:phi_vs_phi_good3}, we obtain
        \begin{align}
            \bigg\|
             \ket{\phi_{(\sk_{\mathrm{in}}, G)}}-\ket{\phi_{(\sk_{\mathrm{in}}, G)}^{\good}}
            \bigg\|_2
             \le&\delta^{1/4}(\secp)+\sqrt{\delta(\secp)}.
        \end{align}
        Thus, by defining a negligible function $\epsilon(\secp) \coloneqq \delta^{1/4}(\secp)+\sqrt{\delta(\secp)}$, we have
        \begin{align}
            \Pr_{\sk_{\mathrm{in}}}
                \bigg[
                 \bigg\|
                  \ket{\phi_{(\sk_{\mathrm{in}}, G)}}-\ket{\phi_{(\sk_{\mathrm{in}}, G)}^{\good}}
                 \bigg\|_2\le \epsilon(\secp)
                \bigg]
                \ge1-\sqrt{\delta(\secp)}\ge1-\negl(\secp)
                \label{eq:good_sk_in_whp}
        \end{align}
        For the fixed Pauli error $E=Z^vX^u$, \cref{thm:better-graph} gives $\weight(e)\le p_{\mathrm g}n_\phys/2\le pm$ for every output graph and $\Pr_G[\Recover(G,e)=u]\ge1-\negl(n_\phys)$. Since $n_\phys(\secp)\ge\secp^\kappa$, the latter failure probability is negligible in $\secp$. Combining this bound with \cref{eq:good_sk_in_whp} gives $\Pr_{\sk_{\mathrm{in}},G}[\|\ket{\phi_{(\sk_{\mathrm{in}},G)}}-\ket{\phi^{\good}_{(\sk_{\mathrm{in}},G)}}\|_2\le\epsilon(\secp)\ \wedge\ \Recover(G,e)=u]\ge1-\sqrt{\delta(\secp)}-\negl(\secp)$, which proves the claim.

        It remains to prove \cref{eq:goal_of_claim:after_1st_step_in_Dec_PRIC}.
        To this end, fix an $n_\logic$-qubit state $\ket{\xi}_\regA = \sum_{x\in\bit^{n_\logic}}\xi_x\ket{x}_\regA$.
        From \cref{claim:good_projection_after_BitDec}, we have
        \begin{align}
            &\Pi^\good_{\regB\regC\regD}\cdot
            \BitDec_{\sk_{\mathrm{in}},\regB\regC\to\regB\regC\regD}\cdot
            (U_G^\dag Z^vX^u U_G)\cdot
            \Enc_{\sk_{\mathrm{in}},\regA\to\regB\regC}\cdot
            \ket{\xi}_\regA
            \notag\\
            =&\ket{e}_{\regB\regC}
            \otimes
            \bigg(
             \frac{1}{\sqrt{2^\ell}}\sum_{x,y:\PRFC.\Dec(\sk', g_{\sk'}(x\|y) \oplus e_1)=x\|y}
            (-1)^{u\cdot c(x\| y)}\xi_x\ket{x\|y}_\regD
            \bigg).
        \end{align}
        For the fixed graph under consideration, we have $\weight(e_1)\le\weight(e)\le pm$.
        For fixed $E$ and $G$, the error $e_1$ is independent of the inner key, so the additive channel $c\mapsto c\oplus e_1$ is $p$-bounded.
        The robustness of the inner PRFC therefore gives the following estimate.
    \begin{align}
        &\Exp_{\sk_{\mathrm{in}}}
             \bigg[
              \bigg\|
                  \Pi^{\good,\regB\regC\regD}\cdot
                  \BitDec_{\sk_{\mathrm{in}},\regB\regC\to\regB\regC\regD}\cdot 
                  (U_G^\dag Z^vX^u U_G)_{\regB\regC}\cdot
                  \Enc_{\sk_{\mathrm{in}},\regA\to\regB\regC}\cdot
                  \ket{\xi}_\regA
              \bigg\|_2^2
             \bigg]
        \notag\\
        =&\sum_{x}|\xi_x|^2
        \Pr_{
        \substack{\sk'\gets\PRFC.\KeyGen(1^\secp)\\ y\gets\bit^\ell}
        }
        \bigg[
         \PRFC.\Dec(\sk', g_{\sk'}(x\|y)\oplus e_1)=x\|y
        \bigg]
        \notag\\
        \ge&\sum_{x}|\xi_x|^2(1-\negl(\secp))
        \tag{By robustness of PRFC and $\weight(e_1)\le\weight(e)\le pm$}\\
        =&1-\negl(\secp),
        \label{eq:exp_good}
    \end{align}
    where we have used $\sum_{x}|\xi_x|^2=\|\ket{\xi}\|^2_2=1$ in the last line.
    This implies \cref{eq:goal_of_claim:after_1st_step_in_Dec_PRIC}, which concludes the proof.
\end{proof}

\begin{proof}[Proof of \cref{eq:rho2_and_rho3}]
        Recall that $(\sk_{\mathrm{in}},G)\in\cK_\good$ is fixed; thus, we have $\Recover(G,e)=u$ and
        \begin{align}
            \bigg\|
             \ket{\phi_{(\sk_{\mathrm{in}}, G)}}-
             \ket{\phi_{(\sk_{\mathrm{in}}, G)}^{\good}}
            \bigg\|_2\le\epsilon(\secp).
            \label{eq:assumption_eq:rho2_and_rho3}
        \end{align}
        Define $\cZ^\good_{\sk_{\mathrm{in}},G}$ as follows:
        \begin{align}
            \cZ^\good_{\sk_{\mathrm{in}}, G}
            &\coloneqq
            \bigg\{
             (x,y)\in\bit^{n_\logic}\times\bit^\ell: \PRFC.\Dec(\sk', g_{\sk'}(x\|y) \oplus e_1)=x\|y
            \bigg\},
        \end{align}
        where $e_1$ denotes the last $m$ bits of $e$.
        Let $\ket{\psi}_{\regA'\regA}=\sum_{x\in\bit^{n_\logic}}\ket{\psi_x}_{\regA'}\ket{x}_\regA$.
        From \cref{claim:good_projection_after_BitDec}, we have
        \begin{align}
            \Pi_{(\sk_{\mathrm{in}}, G)}^\good\ket{\phi_{(\sk_{\mathrm{in}}, G)}}_{\regA'\regB\regC\regD}
            =\ket{e}_{\regB\regC}
            \bigg(
             \frac{1}{\sqrt{2^\ell}}\sum_{(x,y)\in\cZ^\good_{\sk_{\mathrm{in}}, G}}(-1)^{u\cdot c(x\| y)}\ket{\psi_x}_{\regA'}\ket{x\|y}_\regD
            \bigg),
        \end{align}
        where $c(x\| y)=\pi_{k'}(x\|y\|0^r) \| g_{\sk'}(x\|y)$.
        Since $\ket{\phi^{\good}_{(\sk_{\mathrm{in}}, G)}}_{\regA'\regB\regC\regD}\propto\Pi_{(\sk_{\mathrm{in}}, G)}^\good\ket{\phi^{\good}_{(\sk_{\mathrm{in}}, G)}}_{\regA'\regB\regC\regD}$, we obtain $e$ with probability one when we measure $\regB\regC$ in the computational basis.
        Moreover, from the definition of $\cK_\good$, we have $\Recover(G,e)=u$.
        From the above observations, when we run $\PhaseDec_{(\sk_{\mathrm{in}},G)}$ on input $\ket{\phi^{\good}_{(\sk_{\mathrm{in}}, G)}}$, the state just before the final step is
        \begin{align}
            \frac{1}{\|\Pi_{(\sk_{\mathrm{in}}, G)}^\good\ket{\phi_{(\sk_{\mathrm{in}}, G)}}\|_2}\ket{e}_{\regB\regC}
            \underbrace{
            \bigg(
             \frac{1}{\sqrt{2^\ell}}\sum_{(x,y)\in\cZ^\good_{\sk_{\mathrm{in}}, G}}\ket{\psi_x}_{\regA'}\ket{x}_\regA\ket{y}_\regE
            \bigg)
            }_{
            \coloneqq \ket{\vartheta_{(\sk_{\mathrm{in}}, G)}^\good}_{\regA'\regA\regE}
            }.
        \end{align}
        Here we write $\regD=\regA\regE$, where $\regA$ consists of the first $n_\logic$ qubits of $\regD$ and $\regE$ consists of the last $\ell$ qubits.
        Thus, we have
        \begin{align}
            \rho^{(2)}_{\regA'\regA}=
            \frac{1}{\|\Pi_{(\sk_{\mathrm{in}}, G)}^\good\ket{\phi_{(\sk_{\mathrm{in}}, G)}}\|_2^2}
            \Tr_{\regE}
            \bigg[
            \ketbra{\vartheta_{(\sk_{\mathrm{in}}, G)}^\good}{\vartheta_{(\sk_{\mathrm{in}}, G)}^\good}_{\regA'\regA\regE}
            \bigg].
        \end{align}
        Therefore, we have
        \begin{align}
            \frac{1}{2}\|\rho^{(2)}-\rho^{(3)}\|_1
            =&\frac{1}{2}\|\rho^{(2)}-\ketbra{\psi}{\psi}\|_1
            \notag\\
            \le&\sqrt{1-
            \frac{1}{\|\Pi_{(\sk_{\mathrm{in}}, G)}^\good\ket{\phi_{(\sk_{\mathrm{in}}, G)}}\|_2^2}\bigg\|
             (\bra{\psi}_{\regA'\regA}\otimes I_\regE)
             \ket{\vartheta^\good_{(\sk_{\mathrm{in}}, G)}}_{\regA'\regA\regE}
             \bigg\|_2^2},
            \notag\\
            \le&\sqrt{1-
            \bigg\|
             (\bra{\psi}_{\regA'\regA}\otimes I_\regE)
             \ket{\vartheta^\good_{(\sk_{\mathrm{in}}, G)}}_{\regA'\regA\regE}
             \bigg\|_2^2},
            \label{eq:rho3_vs_rho4}
        \end{align}
        where we have used the Fuchs--van de Graaf bound in the first inequality and $\|\Pi^{\good}\ket{\phi_{(\sk_{\mathrm{in}}, G)}}\|_2^2\le 1$ in the last line.
        From a straightforward calculation, we have
        \begin{align}
            \bigg\|
             (\bra{\psi}_{\regA'\regA}\otimes I_\regE)
             \ket{\vartheta^\good_{(\sk_{\mathrm{in}}, G)}}_{\regA'\regA\regE}
            \bigg\|_2^2
            =&
            \frac{1}{2^\ell}\sum_{y\in\bit^\ell}
            \bigg(
             \sum_{x\in\cX^\good_{\sk_{\mathrm{in}}, G ,y}}\| \ket{\psi_x} \|^2_2
            \bigg)^2,
            \label{eq:norm_of_psi_xiGood}
        \end{align}
        where $\cX^\good_{\sk_{\mathrm{in}}, G ,y} \coloneqq \{ x\in\bit^{n_\logic}: (x,y)\in\cZ^\good_{\sk_{\mathrm{in}}, G} \}$.
        If we have
        \begin{align}
            \frac{1}{2^\ell}\sum_{y\in\bit^\ell}
            \bigg(
             \sum_{x\in\cX^\good_{\sk_{\mathrm{in}}, G ,y}}\| \ket{\psi_x} \|^2_2
            \bigg)^2
            \ge 1-\negl(\secp)
            \label{eq:exp_good_x_for_y},
        \end{align}
        then we obtain
        \begin{align}
            \frac{1}{2}\|\rho^{(2)}-\rho^{(3)}\|_1
            \le&
            \sqrt{1-
            \frac{1}{2^\ell}\sum_{y\in\bit^\ell}
             \bigg(
              \sum_{x\in\cX^\good_{\sk_{\mathrm{in}}, G ,y}}\| \ket{\psi_x} \|^2_2
             \bigg)^2},
            \tag{By \cref{eq:rho3_vs_rho4,eq:norm_of_psi_xiGood}}\\
            \le&\negl(\secp).
        \end{align}
        
        To complete the proof, it remains to show \cref{eq:exp_good_x_for_y}.
        To this end, for each $y\in\bit^\ell$, define
        \begin{align}
            X_y\coloneqq
            \sum_{x\in\cX^\good_{\sk_{\mathrm{in}}, G ,y}}\| \ket{\psi_x} \|^2_2.
        \end{align}
        The expectation of $X_y$ for uniformly random $y$ is
        \begin{align}
            \Exp_{y\gets\bit^\ell}[X_y]
            =&\frac{1}{2^\ell}\sum_{y\in\bit^\ell}\sum_{x\in\cX^\good_{\sk_{\mathrm{in}}, G ,y}}\| \ket{\psi_x} \|^2_2
            \notag\\
            =&\frac{1}{2^\ell}\sum_{(x,y)\in\cZ^\good_{\sk_{\mathrm{in}}, G}}\| \ket{\psi_x} \|^2_2
            \notag\\
            =&\|\Pi^{\good}\cdot
            \ket{\phi_{(\sk_{\mathrm{in}}, G)}}\|_2^2.
        \end{align}
        Since  $\|\Pi^{\good}\cdot\ket{\phi_{(\sk_{\mathrm{in}}, G)}}\|_2^2\ge 1-\epsilon(\secp)^2$ from \cref{eq:assumption_eq:rho2_and_rho3,lem:gentle_measurement}, this implies
        \begin{align}
            \Exp_{y\gets\bit^\ell}[X_y]\ge1-\epsilon(\secp)^2.
            \label{eq:exp_X_y}
        \end{align}
        By Markov's inequality (\cref{lem:Markov}), we have
        \begin{align}
            \Pr_{y\gets\bit^\ell}[X_y<1-\epsilon(\secp)]
            =&\Pr_{y\gets\bit^\ell}[1-X_y>\epsilon(\secp)]
            \notag\\
            \le&\frac{\Exp_{y\gets\bit^\ell}[1-X_y]}{\epsilon(\secp)}
            \le\epsilon(\secp).
            \label{eq:X_y_is_large_whp}
        \end{align}
        Now we evaluate the sum:
        \begin{align}
            \frac{1}{2^\ell}\sum_{y\in\bit^\ell}
            \bigg(
             \sum_{x\in\cX^\good_{\sk_{\mathrm{in}}, G ,y}}\| \ket{\psi_x} \|^2_2
            \bigg)^2
            =&\Exp_{y\gets\bit^\ell}[X_y^2]
            \notag\\
            \ge&(1-\epsilon(\secp))^2\cdot
            \Pr_{y\gets\bit^\ell}[X_y\ge 1-\epsilon(\secp)]
            \notag\\
            \ge&(1-\epsilon(\secp))^3.
        \end{align}
        Since $\epsilon(\secp)$ is negligible, this implies \cref{eq:exp_good_x_for_y}, which completes the proof.
\end{proof}

\subsection{Pseudorandomness}
\label{subsec:PRIC_pseudorandomness}

In this subsection, we prove the pseudorandomness of \cref{const:PRIC}.
Our proof relies on two facts.
The first fact is that PF isometries are indistinguishable from Haar-random isometries, as shown in \cite{FOCS:MPSY24}.

\begin{theorem}[Adapted from Theorem 6.2 in \cite{FOCS:MPSY24}]\label{thm:PF_isometry}
    Let $n_\logic,\ell,$ and $s$ be integers, and define $n_\phys\coloneqq n_\logic+\ell+s$.
    For a permutation $\pi'$ over $n_\phys$ bits and a function $f':\bit^{n_\logic+\ell}\to\bit$, define an isometry from $n_\logic$ qubits to $n_\phys$ qubits by
    \begin{align}
        V_{\pi',f'}:\ket{x}\mapsto
        \frac{1}{\sqrt{2^\ell}}\sum_{y\in\bit^\ell}(-1)^{f'(x\|y)}\ket{\pi'(x\|y\|0^{s})}.
    \end{align}
    Then, for any $t$-query adversary $\cA$,
    \begin{align}
        \bigg|
         \Pr_{\pi',f'}[1\gets\cA^{V_{\pi',f'}}]-\Pr_{V\gets\mu_{n_\logic\to n_\phys}}[1\gets\cA^{V}]
        \bigg|
        \le
        O\bigg(
        \frac{t^2}{2^{\ell}}
        \bigg),
    \end{align}
    where $\pi'$ and $f'$ are chosen uniformly at random.
\end{theorem}
% The quantitative PF-isometry bound is recorded in the decision ledger.

The second fact is the following lemma, whose proof is given below.

\begin{lemma}\label{lem:pi||f_vs_pi'}
    Let $n,r,m\in\N$.
    Define $\cD$ and $\cD'$ to be the following distributions over injective functions $\bit^n\to\bit^{n+r+m}$.
    \begin{itemize}
        \item $\cD:$ 
        Choose a uniformly random function $f:\bit^n\to\bit^m$ and a uniformly random permutation $\pi:\bit^{n+r}\to\bit^{n+r}$.
        Define $h:\bit^{n}\to\bit^{n+r+m}$ as the injective function $h(z)\coloneqq\pi(z\|0^r)\|f(z)$.
        Output $h$.
        
        \item $\cD':$
        Choose a uniformly random permutation $\pi':\bit^{n+r+m}\to\bit^{n+r+m}$.
        Define $h:\bit^{n}\to\bit^{n+r+m}$ as the injective function $h(z)\coloneqq\pi'(z\|0^{r+m})$.
        Output $h$.
    \end{itemize}
    Then, we have
    \begin{align}
        {\rm SD}(\cD,\cD')\le 2^{n-r-1}.
    \end{align}
\end{lemma}

Now we are ready to prove pseudorandomness.

\begin{lemma}\label{lem:pseudorandomness_of_PRIC}
    If $n_\phys(\secp)\ge\secp^\kappa$ for some constant $\kappa>0$,
    $r(\secp)\ge n_\logic(\secp)+\ell(\secp)+\omega(\log\secp)$, and
    $\ell(\secp)\ge\omega(\log\secp)$, then \cref{const:PRIC} satisfies
    pseudorandomness.
    % The blocklength and dimension conditions are stated in the lemma and hybrids below.
\end{lemma}

\begin{proof}[Proof of \cref{lem:pseudorandomness_of_PRIC}]
    Let $\cA$ be a QPT algorithm.
    Our goal is to show that
    \begin{align}
        \bigg|
         \Pr_{\sk\gets\KeyGen(1^\secp)}[1\gets\cA^{\Enc_\sk}]-
         \Pr_{V\gets\mu_{n_\logic\to n_\phys}}[1\gets\cA^V]
        \bigg|
        \le\negl(\secp).
    \end{align}
    We use a hybrid argument.
    Define the following hybrids of distributions over isometries $V$.
    % The red text in H_1--H_3 records the adjacent-hybrid differences intentionally.
    In the following, we omit the $\secp$-dependence when it is clear from the context.
    \begin{itemize}
        \item ${\rm H}_0:$
        Run $\sk'\gets\PRFC.\KeyGen(1^\secp)$, and choose $k,k'\gets\bit^\secp$.
        Run $G\gets\GraphSample(1^{n_\phys/2})$ and choose $P\gets\cP_{n_\phys}$.
        $V$ is defined as
        \begin{align}
            \ket{x}\mapsto
            P\cdot U_G \cdot
            \bigg(
            \frac{1}{\sqrt{2^\ell}}\sum_{y\in\bit^\ell}(-1)^{f_{k}(x\|y)}\ket{\pi_{k'}(x\|y\|0^r)}\ket{g_{\sk'}(x\|y)}
            \bigg).
        \end{align}

        \item ${\rm H}_1:$
        {\color{red}Choose uniformly random functions $f:\bit^{n_\logic+\ell}\to\bit$ and $g:\bit^{n_\logic+\ell}\to\bit^m$, and a uniformly random permutation $\pi:\bit^{n_\logic+\ell+r}\to\bit^{n_\logic+\ell+r}$.}
        Run $G\gets\GraphSample(1^{n_\phys/2})$ and choose $P\gets\cP_{n_\phys}$.
        $V$ is defined as 
        \begin{align}
            \ket{x}\mapsto
            P\cdot U_G \cdot
            \bigg(
            \frac{1}{\sqrt{2^\ell}}\sum_{y\in\bit^\ell}(-1)^{\color{red}f(x\|y)}\ket{\color{red}\pi(x\|y\|0^r)}\ket{\color{red}g(x\|y)}
            \bigg).
        \end{align}

        \item ${\rm H}_2:$
        Choose a uniformly random function $f:\bit^{n_\logic+\ell}\to\bit$ and {\color{red}a uniformly random permutation $\pi':\bit^{n_\phys}\to\bit^{n_\phys}$.}
        Run $G\gets\GraphSample(1^{n_\phys/2})$ and choose $P\gets\cP_{n_\phys}$.
        $V$ is defined as
        \begin{align}
            \ket{x}\mapsto
            P\cdot U_G \cdot
            \bigg(
            \frac{1}{\sqrt{2^\ell}}\sum_{y\in\bit^\ell}(-1)^{f(x\|y)}\ket{\color{red}\pi'(x\|y\|0^{r+m})}
            \bigg).
        \end{align}

        \item ${\rm H}_3:$
        {\color{red}Choose a Haar random isometry $W:\C^{2^{n_\logic}}\to\C^{2^{n_\phys}}$.}
        Run $G\gets\GraphSample(1^{n_\phys/2})$ and choose $P\gets\cP_{n_\phys}$.
        $V$ is defined as $V\coloneqq P\cdot U_G \cdot {\color{red}W}$.

        \item ${\rm H}_4:$
        Choose a Haar random isometry $W:\C^{2^{n_\logic}}\to\C^{2^{n_\phys}}$.
        \Erase{Run $G\gets\GraphSample(1^{n_\phys/2})$ and choose $P\gets\cP_{n_\phys}$.}
        $V$ is defined as $V\coloneqq W$.
    \end{itemize}
    
    It is clear that $\Pr_{\sk\gets\KeyGen(1^\secp)}[1\gets\cA^{\Enc_\sk}]=\Pr_{V\gets{\rm{H}}_0}[1\gets\cA^V]$ and $\Pr_{V\gets\mu_{n_\logic\to n_\phys}}[1\gets\cA^V]=\Pr_{V\gets{\rm{H}}_4}[1\gets\cA^V]$.
    Thus, it suffices to argue that ${\rm H}_0$ is indistinguishable from ${\rm H}_4$.
    From the security of PRFC, PRF, and PRP, we have
    \begin{align}
        \bigg|
             \Pr_{V\gets {\rm H}_0}[1\gets\cA^V]
             -\Pr_{V\gets {\rm H}_1}[1\gets\cA^V]
        \bigg|
        \le\negl(\secp).
        \label{eq:H0_vs_H1}
    \end{align}
    For ${\rm{H}}_1$ and ${\rm{H}}_2$, from \cref{lem:pi||f_vs_pi'} with $n=n_\logic+\ell$, we have
    \begin{align}
        \bigg|
             \Pr_{V\gets {\rm H}_1}[1\gets\cA^V]
             -\Pr_{V\gets {\rm H}_2}[1\gets\cA^V]
        \bigg|
        \le 2^{n_\logic+\ell-r-1}
        \le\negl(\secp),
        \label{eq:H1_vs_H2}
    \end{align}
    where we have used $r\ge n_\logic+\ell+\omega(\log\secp)$.
    For ${\rm H}_2$ and ${\rm H}_3$, by \cref{thm:PF_isometry} with $\ell(\secp)\ge\omega(\log\secp)$ and $s=r+m$, we have
    \begin{align}
        \bigg|
             \Pr_{V\gets {\rm H}_2}[1\gets\cA^V]
             -\Pr_{V\gets {\rm H}_3}[1\gets\cA^V]
        \bigg|
        \le\negl(\secp).
        \label{eq:H2_vs_H3}
    \end{align}
    By the left invariance of Haar-random isometries, we have
        \begin{align}
            \Pr_{V\gets {\rm H}_3}[1\gets\cA^V]
         =\Pr_{V\gets {\rm H}_4}[1\gets\cA^V].
        \label{eq:H3_vs_H4}
    \end{align}

    From \cref{eq:H0_vs_H1,eq:H1_vs_H2,eq:H2_vs_H3,eq:H3_vs_H4}, we obtain

     \begin{align}
        \bigg|
             \Pr_{V\gets {\rm H}_0}[1\gets\cA^V]
             -\Pr_{V\gets {\rm H}_4}[1\gets\cA^V]
        \bigg|
        \le\negl(\secp),
    \end{align}
    which concludes the proof.
\end{proof}

\paragraph{Proof of \cref{lem:pi||f_vs_pi'}.}

To complete the proof, we show \cref{lem:pi||f_vs_pi'}.

\begin{proof}[Proof of \cref{lem:pi||f_vs_pi'}]
    For an injective function $h:\bit^n\to\bit^{n+r+m}$, define $U_h:\bit^n\to\bit^{n+r}$ and $W_h:\bit^n\to\bit^m$ by the relation $h(z)=U_h(z)\|W_h(z)$.
    Define $E$ as the event that $U_h(z)\neq U_h(z')$ for all distinct $z,z'\in\bit^n$.
    Then the following holds.

    \begin{claim}\label{claim:D'_on_E_is_D}
        $\cD'$ conditioned on $E$ is the same distribution as $\cD$.
        In other words, for any injective function $h':\bit^{n}\to\bit^{n+r+m}$, $\Pr_{h\gets\cD}[h=h']=\Pr_{h\gets\cD'}[h=h'|E]$.
    \end{claim}

    \begin{claim}\label{claim:E_whp}
        \begin{align}
            \Pr_{h\gets\cD'}[E]\ge 1-2^{n-r-1}.
        \end{align}
    \end{claim}
    We prove these claims below.
    From the above two claims,
    \begin{align}
        {\rm SD}(\cD,\cD')=1-\Pr_{h\gets\cD'}[E]\le 2^{n-r-1},
    \end{align}
    which concludes the proof.
\end{proof}

To complete the proof, we show \cref{claim:D'_on_E_is_D,claim:E_whp}.
To this end, we need the following claim.

\begin{claim}\label{claim:number_of_E}
    Define $N\coloneqq 2^n$, $R\coloneqq 2^r$, and $M\coloneqq 2^m$.
    The number of injective functions that satisfy the event $E$ is $NR(NR-1)\cdots(NR-N+1)\cdot M^N$.
\end{claim}

\begin{proof}
    A function $h:\bit^{n}\to\bit^{n+r+m}$ satisfies the event $E$ if and only if $U_h:\bit^{n}\to\bit^{n+r}$ is injective.
    Thus, it is sufficient to count the number of injective functions $\bit^n\to\bit^{n+r}$ and the number of functions $\bit^n\to\bit^m$.
    The number of injective functions $\bit^{n}\to\bit^{n+r}$ is $NR(NR-1)\cdots(NR-N+1)$.
    On the other hand, the number of functions $\bit^n\to\bit^m$ is $M^N$.
    Therefore, the number of injective functions that satisfy the event $E$ is $NR(NR-1)\cdots(NR-N+1)\cdot M^N$.
\end{proof}

First, we prove \cref{claim:D'_on_E_is_D}.

\begin{proof}[Proof of \cref{claim:D'_on_E_is_D}]
    For notational simplicity, define $N\coloneqq 2^n$, $R\coloneqq 2^r$, and $M\coloneqq 2^m$.
    Fix a $h':\bit^{n}\to\bit^{n+r+m}$.
    If $U_{h'}(z)= U_{h'}(z')$ for some distinct $z,z'\in\bit^{n}$, then
    \begin{align}
        \Pr_{h\gets\cD}[h=h']=\Pr_{h\gets\cD'}[h=h'|E]=0
    \end{align}
    from the definition of $\cD$ and $E$.
    Thus, without loss of generality, we can assume that $U_{h'}(z)\neq U_{h'}(z')$ for all distinct $z,z'\in\bit^{n}$.

    First, let us consider the case $h\gets\cD$.
    From the definition of $\cD$, $U_h$ and $W_h$ are independent.
    Since $U_h(\cdot)=\pi((\cdot)\|0^r)$ for a uniformly random permutation $\pi$, $U_h$ is a uniformly random injective function.
    Thus, we have
    \begin{align}
        \Pr_{h\gets\cD}[U_h=U_{h'}]=\frac{1}{NR(NR-1)\cdots(NR-N+1)}.
    \end{align}
    On the other hand, $W_h$ is a uniformly random function from the definition of $\cD$,
    \begin{align}
        \Pr_{h\gets\cD}[W_h=W_{h'}]=\frac{1}{M^N}.
    \end{align}
    Since $U_h$ and $W_h$ are independent, we have
    \begin{align}
        \Pr_{h\gets\cD}[h=h']
        =\Pr_{h\gets\cD}[U_h=U_{h'}]\Pr_{h\gets\cD}[W_h=W_{h'}]
        =\frac{1}{NR(NR-1)\cdots(NR-N+1)\cdot M^N}.
        \label{eq:prob_h=h'_wrt_D}
    \end{align}

    Next, let us consider the case $h\gets\cD'$.
    Since $h(\cdot)=\pi'((\cdot)\|0^{r+m})$ for a uniformly random permutation $\pi'$, $h$ is a uniformly random injective function.
    Thus, we have
    \begin{align}
        \Pr_{h\gets\cD'}[h=h']=\frac{1}{NRM(NRM-1)\cdots (NRM-N+1)}.
    \end{align}
    In addition, our assumption that $U_{h'}(z)\neq U_{h'}(z')$ for all distinct $z,z'\in\bit^{n}$ implies that $h$ satisfies the event $E$ if $h=h'$.
    Hence, we have
    \begin{align}
        \Pr_{h\gets\cD'}[h=h'\wedge E]=\Pr_{h\gets\cD'}[h=h']=\frac{1}{NRM(NRM-1)\cdots (NRM-N+1)}.
    \end{align}
    On the other hand, from \cref{claim:number_of_E}, we have
    \begin{align}
        \Pr_{h\gets\cD'}[E]=\frac{NR(NR-1)\cdots(NR-N+1)\cdot M^N}{NRM(NRM-1)\cdots (NRM-N+1)}.
    \end{align}
    Therefore,
    \begin{align}
        \Pr_{h\gets\cD'}[h=h'|E]
        =\frac{\Pr_{h\gets\cD'}[h=h'\wedge E]}{\Pr_{h\gets\cD'}[E]}
        =\frac{1}{NR(NR-1)\cdots(NR-N+1)\cdot M^N}.
        \label{eq:prob_h=h'_wrt_D'}
    \end{align}
    From \cref{eq:prob_h=h'_wrt_D,eq:prob_h=h'_wrt_D'}, we have $\Pr_{h\gets\cD}[h=h']=\Pr_{h\gets\cD'}[h=h'\mid E]$, which concludes the proof.
\end{proof}

Next, we prove \cref{claim:E_whp}.

\begin{proof}[Proof of \cref{claim:E_whp}]
    For notational simplicity, define $N\coloneqq 2^n$, $R\coloneqq 2^r$, and $M\coloneqq 2^m$.
    Suppose that $h\gets\cD'$.
    Since $h(\cdot)=\pi'((\cdot)\|0^{r+m})$ for a uniformly random permutation $\pi'$, $h$ is a uniformly random injective function.
    Thus, from \cref{claim:number_of_E}, we have
    \begin{align}
        \Pr_{h\gets\cD'}[E]
        &=\frac{NR(NR-1)\cdots(NR-N+1)\cdot M^N}{NRM(NRM-1)\cdots (NRM-N+1)}
        \notag\\
        &=\prod_{i\in[N-1]}\frac{M(NR-i)}{NRM-i}
        \notag\\
        &=\prod_{i\in[N-1]}\frac{1-\frac{i}{NR}}{1-\frac{i}{NRM}}
        \notag\\
        &\ge\prod_{i\in[N-1]}
        \bigg(
         1-\frac{i}{NR}
        \bigg)
        \tag{By $(1-x)^{-1}\ge1$ for any $0<x<1$}\\
        &\ge1-\frac{N}{2R},
    \end{align}
    where, in the last line, we used $\prod_i(1-a_i)\ge1-\sum_i a_i$ for $a_i>0$.
    Since $N=2^{n}$ and $R=2^r$, the proof is complete.
\end{proof}

\subsection{Remark on Constructing PRICs Secure Against Complex-Conjugate Queries}

Our security definition of PRICs in \cref{def:PRIC} considers only adversaries that query $\Enc_\sk$.
On the other hand, one can consider a stronger security notion against adversaries that also query the complex conjugate of $\Enc_\sk$.
Here, we sketch how to construct PRICs satisfying this stronger security notion by modifying \cref{const:PRIC} as follows.
Instead of using a PRF $F:\bit^\secp\times\bit^{n_\logic+\ell}\to\bit$, we use a PRF $F':\bit^\secp\times\bit^{n_\logic+\ell}\to\trit$.
We define $\KeyGen$ in the same way as in \cref{const:PRIC}.
For $\sk_{\mathrm{in}}=(\sk',k,k')$, we define an isometry
\begin{align}
    V'_{\sk_{\rm{in}}}
    :\ket{x}_\regA\mapsto
             \frac{1}{\sqrt{2^\ell}}\sum_{y\in\bit^\ell}\omega_3^{f_{k}(x\|y)}\ket{\pi_{k'}(x\|y\|0^r)}_\regB\ket{g_{\sk'}(x\|y)}_\regC,
\end{align}
where $\omega_3\coloneqq\exp(2\pi i/3)$ is a primitive third root of unity.
For a secret key $\sk=(\sk_{\rm{in}},G,P)$, define the encoding isometry by $\Enc_\sk\coloneqq PU_G V'_{\sk_{\rm{in}}}$.
$\Dec_\sk$ is defined in the same way as in \cref{const:PRIC}, except that it uses the following $\BitDec'_{\sk_{\rm{in}}}$ in place of $\BitDec_{\sk_{\rm{in}}}$:
$\BitDec'_{\sk_{\rm{in}},\regB\regC\to\regB\regC\regD}$ applies the following operations.
\begin{enumerate}
                \item Apply the following isometry to the register $\regC$ using $\sk'$:
                \begin{align}
                    \ket{c}_\regC\mapsto\ket{c}_\regC\ket{\PRFC.\Dec(\sk',c)}_\regD.
                \end{align}

                \item 
                Apply the following unitary to the register $\regB\regC\regD$ using $\sk'$, $k$, and $k'$:
                \begin{align}
                    \ket{b}_\regB\ket{c}_\regC\ket{z}_\regD\mapsto
                    \omega_3^{-f_{k}(z)}\ket{b\oplus \pi_{k'}(z\|0^r)}_\regB\ket{c\oplus g_{\sk'}(z)}_\regC\ket{z}_\regD.
                \end{align}
            \end{enumerate}

The proof of robustness is the same as the proof of \cref{lem:robustness_of_PRIC}.
For pseudorandomness, it suffices to prove the following analogue of \cref{thm:PF_isometry}: no $t$-query algorithm $\cA$ can distinguish Haar-random isometries from the following isometries, even when given oracle access to their complex conjugates:
\begin{align}
    V'_{\pi',f'}:\ket{x}
    \to\frac{1}{\sqrt{2^\ell}}\sum_{y\in\bit^\ell}\omega_3^{f'(x\|y)}\ket{\pi'(x\|y\|0^s)},
\end{align}
where $f':\bit^{n_\logic+\ell}\to\trit$ is a uniformly random function and $\pi'$ is a uniformly random permutation over $\bit^{n_\phys}$, with $n_\phys=n_\logic+\ell+s$.
We expect that this claim can be shown using the generalized path-recording framework of \cite{SMLBH25,FLMNW26}.
We leave this direction for future work.

\section{Graph Sampling and Efficient Phase Recovery}
\label{sec:better_graph}
In this section, we prove \cref{thm:better-graph}.
In \cref{subsec:better-good-conditions}, we define the good graph conditions under which recovery succeeds.
In \cref{subsec:better-recovery}, we construct $\Recover$ and prove its correctness for graphs satisfying these conditions.
In \cref{subsec:better-sampling}, we construct $\GraphSample$ and show that its output satisfies these conditions with overwhelming probability for each fixed perturbation.
In \cref{subsec:better-main-proof}, we combine these results to prove \cref{thm:better-graph}.
Finally, in \cref{subsec:better-application}, we apply the theorem to construct keyed CWS codes from classical codes that need not be linear.

\subsection{Good Graph Conditions}
\label{subsec:better-good-conditions}
This subsection formulates the conditions under which recovery succeeds.
We use a bipartite graph so that the equation $e=v\oplus Au$ separates into two blocks.
\begin{definition}[Bipartite graph]\label{def:better-bipartite}
Let $G=([2n],E)$ be simple with $L\coloneqq[n]$ and $R\coloneqq[2n]\setminus[n]$.
We say that $G$ is a bipartite graph if its adjacency matrix $A$ can be written as
\begin{align}
 A=\begin{pmatrix}0&B\\B^\top&0\end{pmatrix}
 \label{eq:better-bipartite-matrix}
\end{align}
for some $B\in\mathbb F_2^{n\times n}$.
Equivalently, every edge of $G$ has one endpoint in $L$ and the other endpoint in $R$, and $\{x,y\}\in E$ if and only if $B_{x,y-n}=1$ for all $x\in L$ and $y\in R$.
We call $B$ the biadjacency matrix of $G$.
\end{definition}

For the condition of the recovery algorithm, we define a unique-neighbor set.

\begin{definition}[Unique-neighbor sets]\label{def:better-unique-neighbors}
Let $G$ be a bipartite graph with biadjacency matrix $B\in\mathbb F_2^{n\times n}$.
For $S\subseteq L$ and $T\subseteq R$, define
$\Gamma_R(S)\coloneqq\{y\in R:|\{x\in S:B_{x,y-n}=1\}|=1\}$ and
$\Gamma_L(T)\coloneqq\{x\in L:|\{y\in T:B_{x,y-n}=1\}|=1\}$.
\end{definition}

Now we introduce the condition for the recovery algorithm.

\begin{definition}[Graph good for a fixed perturbation]\label{def:better-good-for-v}
For positive integers $t,d$ and $v\in\bit^{2n}$, call $G$ $(t,d)$-good for $v$ if it is a bipartite graph and the following conditions hold.
\begin{enumerate}
\item[\textnormal{(i) Degree bound.}] 
Every vertex has degree at most $d$, i.e., $\max_{i\in[2n]}\deg_G(i)\le d$.

\item[\textnormal{(ii) Unique-neighbor bound.}] 
For every nonempty $S\subseteq L$ and $T\subseteq R$ with $|S|,|T|\le2t$, $|\Gamma_R(S)|\ge\frac34d|S|$ and $|\Gamma_L(T)|\ge\frac34d|T|$.

\item[\textnormal{(iii) Local-overlap bound.}] 
For every $i\in[2n]$, $|\{j\in[2n]:A_{i,j}=1,\ v_j=1\}|\le d/16$.
\end{enumerate}
\end{definition}

\begin{remark}
The numerical constants in the definition are chosen for notational simplicity rather than because they are essential.
\end{remark}

\subsection{Deterministic Phase Recovery}
\label{subsec:better-recovery}
This subsection constructs the deterministic recovery algorithm for graphs satisfying the good condition and proves its running time and correctness.
To this end, write $u=(u_L,u_R)$, $v=(v_L,v_R)$, and $e=(e_L,e_R)$.
Then we have
\begin{align}
 e=v\oplus Au,
 \qquad e_L=Bu_R\oplus v_L,
 \qquad e_R=B^\top u_L\oplus v_R.
 \label{eq:better-syndrome-decomposition}
\end{align}
The left and right blocks are symmetric.
Thus, an algorithm that recovers $u_R$ from $e_L=Bu_R\oplus v_L$ also recovers $u_L$ after replacing the input by $(B^\top,e_R)$.
We give this block algorithm below.
In the following, fix $t$ and $d$.

\begin{construction}[Block recovery]\label{const:better-block-recovery}
We construct an algorithm $\algo{BlockRecover}^{t,d}$ as follows.
\begin{enumerate}
\item\label{step:better-block-initialize} 
On input $M\in\mathbb F_2^{n\times n}$ and $e\in\bit^n$, initialize $\widehat x^{(0)}=0^n$, $r^{(0)}=e$, and $k=0$.
\item\label{step:better-block-repeat} While $k<8t$, repeat the following steps.
\begin{enumerate}
\item\label{step:better-block-score} For every $j\in[n]$, define and compute the score
\begin{align}
 \sigma_k(j)\coloneqq\left|\left\{i\in[n]:M_{i,j}=1,\ r_i^{(k)}=1\right\}\right|
 \label{eq:better-score}
\end{align}
which counts the coordinates that would change from $1$ to $0$ when column $j$ is flipped.
\item\label{step:better-block-stop} If no $j$ satisfies $\sigma_k(j)>5d/8$, stop the repetition and go to Step~\ref{step:better-block-output}.
\item\label{step:better-block-flip} Otherwise choose the smallest eligible $j$ and set
\begin{align}
 \widehat x^{(k+1)}&\coloneqq\widehat x^{(k)}\oplus\varepsilon_j
 \qquad\text{and}\qquad
 r^{(k+1)}\coloneqq r^{(k)}\oplus M\varepsilon_j,
\end{align}
 where $\varepsilon_j$ denotes the $j$-th standard basis vector of $\bit^n$.
Increase $k$ by one.
\end{enumerate}
\item\label{step:better-block-output} Output $\widehat x^{(k)}$.
\end{enumerate}
\end{construction}

We construct the full recovery algorithm by applying the block recovery algorithm to the two blocks.

\begin{construction}[Full phase recovery]\label{const:better-phase-recovery}
We construct an algorithm $\Recover$ as follows.
\begin{enumerate}
    \item On input a graph $G$ on $[2n]$ and $e\in\bit^{2n}$, if $G$ is not a bipartite graph as in \cref{def:better-bipartite} such that $\max_{i\in[2n]}\deg_G(i)\le d$, output $0^{2n}$.

    \item 
    Otherwise, parse $e=(e_L,e_R)$ and compute
\begin{align}
 \widehat u_R&\coloneqq\algo{BlockRecover}^{t,d}(B,e_L)
 \qquad\text{and}\qquad
 \widehat u_L\coloneqq\algo{BlockRecover}^{t,d}(B^\top,e_R).
\end{align}
Then, compute the two estimates and output
\begin{align}
 \Recover(G,e)\coloneqq(\widehat u_L,\widehat u_R),
 \label{eq:better-full-recovery}
\end{align}
\end{enumerate}
\end{construction}

From \cref{const:better-phase-recovery,const:better-block-recovery}, it is easy to see the following lemma.

\begin{lemma}[Recovery runtime]\label{lem:better-recovery-runtime}
If $t$ is polynomially bounded in $n$, then $\Recover$ runs in deterministic polynomial time.
\end{lemma}

The following lemma gives the three score properties used in the correctness proof: stopping at the correct estimate, finding a candidate before recovery, and decreasing the signal weight after each flip.

\begin{lemma}[Score properties]\label{lem:better-score-properties}
Assume that $G$ is $(t,d)$-good for $v$.
For the left block, set $M=B$, $\widehat x^{(k)}=\widehat u_R^{(k)}$, and $r^{(k)}=r_L^{(k)}$.
Let $z_R^{(k)}=u_R\oplus\widehat u_R^{(k)}$ denote the estimation error after $k$ flips.
The residual is $r_L^{(k)}=e_L\oplus B\widehat u_R^{(k)}=Bz_R^{(k)}\oplus v_L$.
\begin{itemize}
\item[\textnormal{(i) Stopping at the correct estimate.}] If $z_R^{(k)}=0$, then $\sigma_k(j)\le d/16$ for every $j\in[n]$.
\item[\textnormal{(ii) Existence of a candidate.}] If $0<\weight(z_R^{(k)})\le2t$, then $\max_{j\in[n]}\sigma_k(j)>5d/8$.
\item[\textnormal{(iii) Decrease in signal weight.}] For every $j\in[n]$, if $\sigma_k(j)>5d/8$, then
\begin{align}
 \weight\bigl(B(z_R^{(k)}\oplus\varepsilon_j)\bigr)-\weight\bigl(Bz_R^{(k)}\bigr)<-d/8.
 \label{eq:better-potential-decrease}
\end{align}
\end{itemize}
The same three properties hold for the right block after replacing $(B,u_R,v_L)$ with $(B^\top,u_L,v_R)$.
\end{lemma}

\begin{proof}[Proof of \cref{lem:better-score-properties}]
We prove only the left-block properties because the right-block proof is symmetric.
For item (i), we use the local-overlap bound.
If $z_R^{(k)}=0$, then $r_L^{(k)}=v_L$.
Fix $j\in[n]$.
Since $A_{n+j,x}=B_{x,j}$ for every $x\in[n]$, the local-overlap bound at $i=n+j$ gives
\begin{align}
 \sigma_k(j)=\left|\left\{x\in[n]:B_{x,j}=1,\ (v_L)_x=1\right\}\right|\le d/16.
\end{align}

For item (ii), we use the unique-neighbor and local-overlap bounds.
Assume $0<\weight(z_R^{(k)})\le2t$ and let $T\coloneqq\{n+j:j\in[n],\ (z_R^{(k)})_j=1\}$, which represents the right vertices where the current estimate differs from $u_R$.
Then $T\subseteq R$ and $0<|T|=\weight(z_R^{(k)})\le2t$.
Thus, the unique-neighbor bound in \cref{def:better-good-for-v} gives $|\Gamma_L(T)|\ge3d|T|/4$.
By the definition of $\Gamma_L(T)$, each $x\in\Gamma_L(T)$ is adjacent to exactly one vertex of $T$.
Thus, we have
\begin{align}
 \sum_{y\in T}|\{x\in\Gamma_L(T):B_{x,y-n}=1\}|=|\Gamma_L(T)|\ge3d|T|/4.
\end{align}
From this and the standard average argument, we can show that there exists some $n+j\in T$ such that $|W_j|\ge3d/4$, where 
\begin{align}
 W_j\coloneqq\{x\in\Gamma_L(T):B_{x,j}=1\}
\end{align}
is the set of left vertices whose only neighbor in $T$ is $n+j$.
For every $x\in W_j$, the definition of $\Gamma_L(T)$ and $n+j\in T$ imply
\begin{align}
 (Bz_R^{(k)})_x=1,
 \qquad\text{and}\qquad
 (r_L^{(k)})_x=1\oplus(v_L)_x.
\end{align}
Consequently,
\begin{align}
 W_j\setminus\Support(v_L)
 \subseteq\{x\in[n]:B_{x,j}=1,\ (r_L^{(k)})_x=1\}.
\end{align}
The set on the right has cardinality $\sigma_k(j)$.
Applying the local-overlap bound at the right vertex $n+j$ gives $|W_j\cap\Support(v_L)|\le d/16$.
Therefore
\begin{align}
 \sigma_k(j)\ge|W_j\setminus\Support(v_L)|
 =|W_j|-|W_j\cap\Support(v_L)|
 \ge3d/4-d/16>5d/8.
\end{align}

For item (iii), we use the degree and local-overlap bounds.
Fix $j\in[n]$ with $\sigma_k(j)>5d/8$ and write $C_j\coloneqq\{x\in[n]:B_{x,j}=1\}$.
Since $r_L^{(k)}=Bz_R^{(k)}\oplus v_L$, we have
\begin{align}
 (C_j\cap\Support(r_L^{(k)}))\setminus\Support(v_L)
 \subseteq C_j\cap\Support(Bz_R^{(k)}).
\end{align}
The score definition with $M=B$ and $r=r_L^{(k)}$ gives $\sigma_k(j)=|C_j\cap\Support(r_L^{(k)})|.$
Thus, we have
\begin{align}
 |C_j\cap\Support(Bz_R^{(k)})|
 \ge|(C_j\cap\Support(r_L^{(k)}))\setminus\Support(v_L)|,
\end{align}
which implies
\begin{align}
    |C_j\cap\Support(Bz_R^{(k)})|
    &\ge\sigma_k(j)-|C_j\cap\Support(r_L^{(k)})\cap\Support(v_L)|\\
    &\ge\sigma_k(j)-|C_j\cap\Support(v_L)|.
\end{align}
Since $\sigma_k(j)>5d/8$ from the assumption and $|C_j\cap\Support(v_L)|\le d/16$ from the local-overlap bound at $i=n+j$, we have
\begin{align}
    |C_j\cap\Support(Bz_R^{(k)})|
    >\frac{5d}{8}-\frac{d}{16}=\frac{9d}{16}.
\end{align}
The degree bound in \cref{def:better-good-for-v} gives $|C_j|\le d$.
Flipping $j$ changes only the coordinates in $C_j$, changing each $1$ to $0$ and each $0$ to $1$.
Therefore
\begin{align}
 \weight\bigl(B(z_R^{(k)}\oplus\varepsilon_j)\bigr)-\weight\bigl(Bz_R^{(k)}\bigr)
 =|C_j|-2|C_j\cap\Support(Bz_R^{(k)})|
 <d-\frac{9d}{8}
 =-\frac d8.
\end{align}
\end{proof}

The following gives the correctness of the recovery algorithm.

\begin{theorem}[Deterministic full recovery]\label{thm:better-full-recovery}
If $G$ is $(t,d)$-good for $v$, then for every $u\in\bit^{2n}$ with $\weight(u)\le t$,
\begin{align}
 \Recover(G,v\oplus Au)=u.
 \label{eq:better-full-correctness}
\end{align}
\end{theorem}

Before proving \cref{thm:better-full-recovery}, we outline the argument.
Flipping the selected $j$-th bit changes $Bz_R^{(k)}$ to $Bz_R^{(k)}\oplus B\varepsilon_j$, and item (iii) of \cref{lem:better-score-properties} shows that its weight decreases.
This decrease bounds the number of flips and keeps the difference support below $2t$.
Within this range, \cref{lem:better-score-properties} supplies a candidate while the difference is nonzero and forces the algorithm to stop when it is zero.

\begin{proof}[Proof of \cref{thm:better-full-recovery}]
We prove the left block and use symmetry for the right block.
Fix $u=(u_L,u_R)$ with $\weight(u)\le t$.
We analyze $\algo{BlockRecover}^{t,d}(B,e_L)$, writing $\widehat x^{(k)}=\widehat u_R^{(k)}$ and $r^{(k)}=r_L^{(k)}$, where $\widehat u_R^{(k)}$ is the current estimate of $u_R$ after $k$ flips.
Let $z_R^{(k)}=u_R\oplus\widehat u_R^{(k)}$.
At Step~\ref{step:better-block-initialize}, $\widehat u_R^{(0)}=0^n$, so $z_R^{(0)}=u_R$, $r_L^{(0)}=e_L=Bz_R^{(0)}\oplus v_L$ and
\begin{align}
 \weight(Bz_R^{(0)})\le d\weight(u_R)\le dt
\end{align}
from the degree bound in \cref{def:better-good-for-v}.
We show that $r_L^{(k)}=Bz_R^{(k)}\oplus v_L$ holds after every number $k$ of flips performed by the algorithm.
The initialization above establishes this relation for $k=0$.
Suppose that it holds after $k$ flips and that the algorithm performs another flip.
Let $j$ be the coordinate selected at Step~\ref{step:better-block-flip}, so that $\sigma_k(j)>5d/8$.
The update of the estimate gives
\begin{align}
 z_R^{(k+1)}
 =u_R\oplus\widehat u_R^{(k+1)}
 =u_R\oplus\widehat u_R^{(k)}\oplus\varepsilon_j
 =z_R^{(k)}\oplus\varepsilon_j.
\end{align}
Using the residual update and the induction hypothesis, we obtain
\begin{align}
 r_L^{(k+1)}
 =r_L^{(k)}\oplus B\varepsilon_j
 =Bz_R^{(k)}\oplus v_L\oplus B\varepsilon_j
 =B(z_R^{(k)}\oplus\varepsilon_j)\oplus v_L
 =Bz_R^{(k+1)}\oplus v_L.
\end{align}
Thus the relation also holds after $k+1$ flips, completing the induction.

We show that $\weight(z_R^{(k)})<2t$ at every state reached by the algorithm.
To prove this bound, we show that $\weight(Bz_R^{(k)})$ decreases at each flip.
For any performed flip, the selected coordinate $j$ satisfies $\sigma_k(j)>5d/8$, so item (iii) of \cref{lem:better-score-properties} gives
\begin{align}
 \weight(Bz_R^{(k+1)})
 <\weight(Bz_R^{(k)})-\frac d8.
 \label{eq:better-recovery-weight-decrease}
\end{align}
For a contradiction, assume that $\weight(z_R^{(k)})\ge 2t$ for some $k$.
Each flip changes $\weight(z_R^{(k)})$ by exactly one, and the initial value is at most $t$.
Let $k$ be the first iteration at which $\weight(z_R^{(k)})=2t$, and let
$T=\{n+j:j\in[n],\ (z_R^{(k)})_j=1\}$ at that iteration.
Each $x\in\Gamma_L(T)$ has exactly one edge from the error support, so $(Bz_R^{(k)})_x=1$.
The unique-neighbor bound in \cref{def:better-good-for-v} would therefore imply
\begin{align}
    \weight(Bz_R^{(k)})
    \ge|\Gamma_L(T)|
    \ge\frac34d|T|
    =\frac32dt
    >dt
    \ge\weight(Bz_R^{(0)}),
\end{align}
which contradicts the decrease in \cref{eq:better-recovery-weight-decrease}.

Next, we show that the algorithm reaches $z_R^{(k)}=0$ for some $k<8t$.
When $z_R^{(k)}\ne0$, item (ii) of \cref{lem:better-score-properties} supplies a candidate, so Step~\ref{step:better-block-stop} cannot occur.
These bounds also rule out $8t$ flips: if the algorithm performed $8t$ flips, we would have
\begin{align}
 \weight(Bz_R^{(8t)})
 < \weight(Bz_R^{(0)})-8t\cdot\frac d8
 \le dt-dt=0,
\end{align}
contradicting the nonnegativity of Hamming weight.
Hence it reaches $z_R^{(k)}=0$, equivalently $\widehat u_R^{(k)}=u_R$, before reaching the repetition limit in Step~\ref{step:better-block-repeat}.

When $z_R^{(k)}=0$, item (i) of \cref{lem:better-score-properties} makes every score at most $d/16$, and the output at Step~\ref{step:better-block-output} returns $\widehat u_R=u_R$.
The right block is symmetric under replacing $(B,u_R,v_L)$ with $(B^\top,u_L,v_R)$.
Therefore $\widehat u_L=u_L$ and the output is $(\widehat u_L,\widehat u_R)=u$.
\end{proof}

\subsection{Construction of $\GraphSample$}
\label{subsec:better-sampling}
We construct $\GraphSample$ and prove that its output satisfies the good condition with overwhelming probability.
Fix a positive rational constant $\gamma\le\min\{2^{-32},p/2\}$ and choose before sampling
\begin{align}
 d(n)&\coloneqq16\left\lfloor\frac1{16}\min\left\{\lceil\log_2(2n)\rceil^2,\frac{\gamma n}{t(n)}\right\}\right\rfloor.
 \label{eq:better-asymptotic-degree}
\end{align}
We choose this degree so that the output of the graph sampler constructed below satisfies $(d(n)+1)t(n)\le pn$ for sufficiently large $n$, which yields the induced-weight bound because every sampled graph has degree at most $d(n)$.

\begin{construction}[Graph sampler]\label{const:better-graph-sampler}
On input $1^n$, $\GraphSample$ samples a graph as follows.
\begin{enumerate}
\item Choose independent uniform permutations $\pi_1,\ldots,\pi_d:[n]\to[n]$.
\item Let $(P_a)_{i,j}=1$ exactly when $j=\pi_a(i)$ and define
\begin{align}
 B&\coloneqq P_1\oplus\cdots\oplus P_d,
 \qquad A\coloneqq\begin{pmatrix}0&B\\B^\top&0\end{pmatrix}.
 \label{eq:better-matching-sampler}
\end{align}
\item Output the graph $G$ with adjacency matrix $A$.
\end{enumerate}
\end{construction}

By construction, we immediately obtain the following lemma.
Since the proof is straightforward, we omit it.

\begin{lemma}\label{lem:better-sampler-runtime}
    $\GraphSample$ is PPT.
\end{lemma}

We also obtain the following lemma.

\begin{lemma}[Sampler degree and induced-weight bound]\label{lem:better-sampler-degree}
Every output $G$ of $\GraphSample(1^n)$ satisfies
\begin{align}
 \max_{i\in[2n]}\deg_G(i)\le d(n).
\end{align}
For all sufficiently large $n$ and all $u,v\in\bit^{2n}$ satisfying $|\Support(u)\cup\Support(v)|\le t(n)$, every output graph satisfies
\begin{align}
 \weight(v\oplus Au)\le pn.
\end{align}
\end{lemma}

\begin{proof}[Proof of \cref{lem:better-sampler-degree}]
Each permutation matrix $P_a$ has one $1$ in every row and column.
Therefore every row and column of $B=P_1\oplus\cdots\oplus P_d$ has weight at most $d(n)$, so every output graph has maximum degree at most $d(n)$.
For the induced-weight bound, if $|\Support(u)\cup\Support(v)|\le t(n)$, then
\begin{align}
 \weight(v\oplus Au)
 \le\weight(v)+\max_{i\in[2n]}\deg_G(i)\weight(u)
 \le(d(n)+1)t(n)\le pn,
 \label{eq:better-degree-budget}
\end{align}
where the last inequality follows for sufficiently large $n$ from $d(n)t(n)\le\gamma n$, $t(n)=o(n)$, and $\gamma\le p/2$.
\end{proof}

Next, we show that this outputs a good graph with high probability.

\begin{theorem}[Sampling a graph good for a fixed perturbation]\label{thm:better-sampler-good}
For all sufficiently large $n$ and every $v\in\bit^{2n}$ with $\weight(v)\le t(n)$,
\begin{align}
 \Pr_{G\gets\GraphSample(1^n)}[G\text{ is }(t(n),d(n))\text{-good for }v]\ge1-\negl(n).
 \label{eq:better-good-sampling}
\end{align}
\end{theorem}

\begin{proof}[Proof of \cref{thm:better-sampler-good}]
Throughout this proof, probabilities are over $G\gets\GraphSample(1^n)$, and we omit this subscript.
Note that, for sufficiently large $n$, $d$ is a multiple of $16$, $d\ge128$, and $2dt\le2^{-31}n$.
For each $a\in[d]$, the edges $\{x,n+\pi_a(x)\}$ for $x\in[n]$ form a perfect matching: every vertex belongs to exactly one of these edges.
Keeping one copy of each edge for every permutation that generates it gives a multigraph, in which two vertices may be joined by more than one edge.
We bound the failure probabilities of the unique-neighbor and local-overlap conditions separately, and then take a union bound.

First, we show
\begin{align}
 \Pr[\text{unique-neighbor condition fails}]
 \le 4t(d/n)^{d/16}.
 \label{eq:better-expansion-error}
\end{align}
To prove \cref{eq:better-expansion-error}, we first establish the implication from unique-neighbor failure to neighborhood concentration: for each nonempty set $S\subseteq L$ with $s=|S|\le2t$, we show that $|\Gamma_R(S)|<\frac34ds$ implies $|\mathcal N(S)|<\frac78ds$, where
\begin{align}
 \mathcal N(S)\coloneqq\{n+\pi_a(x):x\in S,\ a\in[d]\}
\end{align}
is the neighborhood in the multigraph obtained by superposing the $d$ matchings.
Define
\begin{align}
 U_1(S)\coloneqq
 \left\{y\in R:
 \left|\left\{(x,a)\in S\times[d]:n+\pi_a(x)=y\right\}\right|=1
 \right\}.
\end{align}
Each pair $(x,a)$ specifies one edge, so the inner cardinality counts parallel edges with their multiplicity.
Each of the $s$ vertices in $S$ contributes one edge for each of the $d$ permutations, giving $ds$ edges counted with multiplicity.
Counting these edges by their endpoints in $R$, each vertex in $U_1(S)$ receives exactly one edge from $S$, whereas each vertex in $\mathcal N(S)\setminus U_1(S)$ receives at least two: it receives at least one by the definition of $\mathcal N(S)$, but not exactly one by the definition of $U_1(S)$.
Therefore, we have
\begin{align}
 ds\ge|U_1(S)|+2\bigl(|\mathcal N(S)|-|U_1(S)|\bigr)
\end{align}
which implies
\begin{align}
    |U_1(S)|\ge2|\mathcal N(S)|-ds.
\end{align}
For each $y\in U_1(S)$, exactly one edge connects a vertex in $S$ to $y$.
This edge occurs in exactly one of the $d$ matchings, so it remains after the XOR reduction.
Therefore $U_1(S)\subseteq\Gamma_R(S)$ in the resulting simple graph, and
\begin{align}
 |\Gamma_R(S)|
 \ge |U_1(S)|
 \ge 2|\mathcal N(S)|-ds,
\end{align}
which gives $|\mathcal N(S)|<\frac78ds$ if $|\Gamma_R(S)|<\frac34ds$.

The following claim bounds the probability of this concentration.

\begin{claim}\label{claim:better-neighborhood-concentration}
For $1\le s\le2t$,
\begin{align}
 \Pr[\exists S\subseteq L:|S|=s,\ |\mathcal N(S)|<7ds/8]
 \le\left(\frac dn\right)^{d/16}.
 \label{eq:better-expansion-union}
\end{align}
\end{claim}

For each $1\le s\le2t$, by \cref{claim:better-neighborhood-concentration}, we have
\begin{align}
 \Pr
 \bigg[\exists S\subseteq L:|S|=s,|\Gamma_R(S)|<\frac34ds
 \bigg]
 \le\Pr
 \bigg[\exists S\subseteq L:|S|=s,\ |\mathcal N(S)|<\frac78ds
 \bigg]
 \le\left(\frac dn\right)^{d/16}.
\end{align}
Thus, by the union bound, 
\begin{align}
    \Pr[\exists S\subseteq L:0<|S|\le2t,\ |\Gamma_R(S)|<\tfrac34d|S|]
    \le 2t\left(\frac dn\right)^{d/16}.
\end{align}
Since the same bound also holds for the right case, the union bound gives \cref{eq:better-expansion-error}.

For the local-overlap condition, fix $i\in L$ and set $k=d/16$.
The edge generated by $\pi_a$ connects $i$ to $n+\pi_a(i)$.
Define
\begin{align}
 X\coloneqq\bigl|\{a\in[d]:v_{n+\pi_a(i)}=1\}\bigr|.
\end{align}
The variable $X$ counts permutation indices $a$, so repeated occurrences of the same neighbor are counted with multiplicity.
We first show that a violation of the local-overlap condition at $i$ implies $X>k$.
If $B_{i,j}=1$, then the number of $a$ satisfying $\pi_a(i)=j$ is odd and therefore at least one.
Consequently,
\begin{align}
 \bigl|\{j\in[n]:B_{i,j}=1,\ v_{n+j}=1\}\bigr|
 \le\sum_{\substack{j\in[n]\\v_{n+j}=1}}
       \bigl|\{a\in[d]:\pi_a(i)=j\}\bigr|\notag
 =X.
\end{align}
Thus, if the local-overlap condition fails at $i$, then $X>k$.

The following lemma gives the upper bound of the probability $X>k$.
\begin{claim}\label{claim:better-local-overlap-tail}
For $X$ and $k$ defined above,
\begin{align}
 \Pr[X>k]\le(16\mathrm e\,t/n)^{d/16}.
\end{align}
\end{claim}
Therefore, \cref{claim:better-local-overlap-tail} upper-bounds the probability of a local-overlap violation for the fixed left vertex $i$.
For a right vertex $n+j$, the neighbor generated by $\pi_a$ is $\pi_a^{-1}(j)$.
The inverse permutations are independent and uniform, and we obtain the same upper bound.
Thus, a union bound over all $2n$ vertices gives
\begin{align}
 \Pr[\text{local-overlap condition fails}]
 \le 2n(16\mathrm e\,t/n)^{d/16}.
 \label{eq:better-local-error}
\end{align}

Combining the two failure bounds and the union bound gives
\begin{align}
 &\Pr[G\text{ is }(t(n),d(n))\text{-good for }v]
 \ge1-4t(d/n)^{d/16}-2n(16\mathrm e\,t/n)^{d/16}.
 \label{eq:better-finite-error}
\end{align}
By \cref{eq:better-radius-growth,eq:better-asymptotic-degree}, we have $\omega(\log n/\log\log n)\le d \le O(\log^2 n)$, and $\log(n/t)=\Omega(\log\log n)$.
Consequently, we have $d\log(n/d)=\omega(\log n)$ and $d\log\bigl(n/(16\mathrm e\,t)\bigr)=\omega(\log n)$.
Both error terms in \cref{eq:better-finite-error} are therefore $n^{-\omega(1)}$.
Together with the degree bound in \cref{lem:better-sampler-degree}, this proves \cref{eq:better-good-sampling}.
\end{proof}

To complete the proof, we prove \cref{claim:better-neighborhood-concentration,claim:better-local-overlap-tail}.

\begin{proof}[Proof of \cref{claim:better-neighborhood-concentration}]
Fix $1\le s\le2t$ and set $\ell=7ds/8$.
Since $d$ is a multiple of $16$, $\ell$ is an integer, and the condition $2dt\le2^{-31}n$ gives $s\le\ell<n$.

For fixed sets $S\subseteq L$ and $W\subseteq R$ of sizes $s$ and $\ell$, respectively, we first show that $\Pr[\mathcal N(S)\subseteq W]\le(\ell/n)^{ds}$.
Write $W-n=\{y-n:y\in W\}\subseteq[n]$ for the corresponding set of right-vertex indices.
For one uniform permutation $\pi_a$,
\begin{align}
 \Pr[\pi_a(S)\subseteq W-n]
 =\frac{\binom{\ell}{s}}{\binom{n}{s}}
 =\prod_{h=0}^{s-1}\frac{\ell-h}{n-h}
 \le\left(\frac \ell n\right)^s.
\end{align}
The $d$ permutations are independent, so
\begin{align}
\Pr[\mathcal N(S)\subseteq W]
 =\Pr[\pi_a(S)\subseteq W-n\text{ for every }a\in[d]]
 \le (\ell/n)^{ds}.
\end{align}

We now show that allowing all choices of $S$ and $W$ increases this bound by at most the factor $\binom ns\binom n\ell$.
Thus, taking a union bound over all pairs $(S,W)$ gives
\begin{align}
 \Pr[\exists S\subseteq L:|S|=s,\ |\mathcal N(S)|<\ell]
 \le
 \sum_{\substack{S\subseteq L\\|S|=s}}
 \sum_{\substack{W\subseteq R\\|W|=\ell}}
 \Pr[\mathcal N(S)\subseteq W]
 \le\binom ns\binom n\ell
 \left(\frac{\ell}{n}\right)^{ds}.
 \label{eq:better-expansion-counting}
\end{align}
The last inequality uses the bound for each fixed pair: there are $\binom ns$ ways to choose $s$ vertices from the $n$ vertices of $L$, and $\binom n\ell$ ways to choose $\ell$ vertices from the $n$ vertices of $R$.

We next show that this upper bound is at most $\exp[-(ds/16)\log(n/(ds))]$.
Set $a=\log(n/(ds))\ge31\ln2>20$.
Using $\binom Nm\le(\mathrm eN/m)^m$ and $\ell=7ds/8$, we obtain
\begin{align}
 \log\left[\binom ns\binom n\ell\left(\frac \ell n\right)^{ds}\right]
 \le ds\left[\frac{1+\log d+a}{d}+\frac78-\frac a8\right]
 \le ds\left(1-\frac{15a}{128}\right)
 \le-\frac{ds}{16}\log\frac{n}{ds},
 \label{eq:better-expansion-exponent}
\end{align}
The first inequality drops the negative term $\frac18\log\frac78$.
Here the second inequality uses $(1+\log d)/d\le1/8$ and $1/d\le1/128$, while the last uses $a>20$.

Finally, we show that $\exp[-(ds/16)\log(n/(ds))]\le(d/n)^{d/16}$ uniformly for $1\le s\le2t$.
The function $f(s)=s\log(n/(ds))$ is increasing on $1\le s\le2t$ because $f'(s)=\log(n/(ds))-1>0$.
Thus $f(s)\ge f(1)$, so
\begin{align}
 \exp\left[-\frac{ds}{16}\log\frac{n}{ds}\right]
 \le\exp\left[-\frac d{16}\log\frac nd\right]
 =\left(\frac dn\right)^{d/16}.
 \label{eq:better-expansion-uniform}
\end{align}

From \cref{eq:better-expansion-counting,eq:better-expansion-exponent,eq:better-expansion-uniform}, we obtain \cref{eq:better-expansion-union}, which concludes the proof.
\end{proof}

\begin{proof}[Proof of \cref{claim:better-local-overlap-tail}]
Fix the left vertex $i$ and the variables $X$ and $k=d/16$ defined above.
It suffices to bound $\Pr[X\ge k]$.
Since $\pi_a(i)$ is uniform on $[n]$,
\begin{align}
 \Pr[v_{n+\pi_a(i)}=1]
 =\frac{|\Support(v)\cap R|}{n}
 \le\frac tn,
\end{align}
where $R$ denotes the set of right vertices.
The event $X\ge k$ means that one can choose $k$ successful permutation indices.
More precisely,
\begin{align}
 \{X\ge k\}
 =\bigcup_{\substack{J\subseteq[d]\\|J|=k}}
   \bigcap_{a\in J}\{v_{n+\pi_a(i)}=1\}.
\end{align}
Here $J$ is a $k$-element set of successful permutation indices, not a set of graph vertices.
For fixed $v$ and $i$, the events in the intersection depend on distinct independent permutations, so for every fixed $J$,
\begin{align}
 \Pr\left[\bigcap_{a\in J}\{v_{n+\pi_a(i)}=1\}\right]
 =\prod_{a\in J}\Pr[v_{n+\pi_a(i)}=1]
 \le\left(\frac tn\right)^k.
\end{align}
Since there are $\binom dk$ choices for $J$, taking a union bound and using $\binom dk\le(\mathrm e d/k)^k$, with $k=d/16$, gives
\begin{align}
 \Pr[X>k]\le\Pr[X\ge k]
 \le\binom dk(t/n)^k
 \le(16\mathrm e\,t/n)^{d/16},
 \label{eq:better-local-tail}
\end{align}
which completes the proof.
\end{proof}

\subsection{Proof of the Main Theorem}
\label{subsec:better-main-proof}
This subsection combines the sampling guarantee with deterministic recovery and verifies the induced-weight bound.

\SampleRecover*
\begin{proof}[Proof of \cref{thm:better-graph}]
For any admissible pair, both $u$ and $v$ have weight at most $t(n)$.
By \cref{lem:better-recovery-runtime,lem:better-sampler-runtime}, the algorithms in \cref{const:better-graph-sampler,const:better-phase-recovery} run in polynomial time.
They satisfy the induced-weight bound and recovery guarantee by \cref{lem:better-sampler-degree,thm:better-sampler-good}.
This completes the proof.
\end{proof}

\subsection{Application}
\label{subsec:better-application}

We apply \cref{thm:better-graph} to construct efficiently decodable keyed CWS codes from classical codes, without assuming linearity.
Fix a constant $p>0$.
Let $n_\phys$ range over even positive integers, and let $n_\logic=n_\logic(n_\phys)$.
Let $t:\N\to\N$ satisfy
\begin{align}
 t(n_\phys)=o\left(\frac{n_\phys\log\log n_\phys}{\log n_\phys}\right).
 \label{eq:better-cws-radius}
\end{align}
Let $C:\bit^{n_\logic}\to\bit^{n_\phys}$ be an injective classical encoder with a deterministic decoder $D:\bit^{n_\phys}\to\bit^{n_\logic}$, both computable in time polynomial in $n_\phys$, such that
\begin{align}
 D(C(x)\oplus e)=x
 \quad\text{for all }x\in\bit^{n_\logic}
 \text{ and }e\in\bit^{n_\phys}\text{ with }\weight(e)\le p n_\phys/2.
 \label{eq:better-classical-decoding}
\end{align}
We suppress the dependence of $C$ and $D$ on $n_\phys$.
If the original decoder has a failure output, replace it by $0^{n_\logic}$ to define $D$ on all inputs.
Use $\GraphSample$ and $\Recover$ from \cref{thm:better-graph} with constant $p$ and radius function $m\mapsto t(2m)$.
First, we give the construction.

\begin{construction}[Keyed CWS code]\label{const:better-keyed-cws}
Define a tuple of QPT algorithms $(\KeyGen,\Enc,\Dec)$ as follows.
\begin{itemize}
\item $\KeyGen(1^{n_\phys})$: Sample $G\gets\GraphSample(1^{n_\phys/2})$ and an independent $P\gets\cP_{n_\phys}$, and output $\sk=(G,P)$.
\item $\Enc_\sk$: Apply the isometry
\begin{align}
 V_C:\ket{x}\mapsto\ket{C(x)}
\end{align}
and then apply $PU_G$, where $U_G$ is defined in \cref{def:CWS_encoding_unitary}.
Thus $\Enc_\sk=PU_GV_C$.
\item $\Dec_\sk$: Perform the following steps.
\begin{enumerate}
\item Apply $U_G^\dagger P^\dagger$ to the received register $\regB$.
\item\label{step:better-cws-bit-decode} Append a register $\regA$ of $n_\logic$ qubits initialized to zero and coherently apply
\begin{align}
 \ket{y}_\regB\ket{0^{n_\logic}}_\regA
 \longmapsto
 \ket{y\oplus C(D(y))}_\regB\ket{D(y)}_\regA.
 \label{eq:better-cws-bit-decoding}
\end{align}
\item Measure $\regB$ in the computational basis to obtain $e$, and compute $\widehat u=\Recover(G,e)$.
\item\label{step:better-cws-phase-correct} Apply $\ket{x}_\regA\mapsto(-1)^{\widehat u\cdot C(x)}\ket{x}_\regA$ and output $\regA$.
\end{enumerate}
\end{itemize}
\end{construction}

We show that this construction satisfies the following robustness guarantee.

\begin{theorem}\label{thm:better-keyed-cws}
Under the assumptions above, \cref{const:better-keyed-cws} encodes $n_\logic$ qubits into $n_\phys$ qubits and corrects $t(n_\phys)$-local quantum noise on average over the key.
More precisely, for every family $\{\cE_{n_\phys}\}$ of $t(n_\phys)$-local noise channels independent of the key,
\begin{align}
 \left\|
 \Exp_{\sk\gets\KeyGen(1^{n_\phys})}
 \bigl[\Dec_\sk\circ\cE_{n_\phys}\circ\Enc_\sk\bigr]
 -\identitymap_{n_\logic}
 \right\|_\diamond
 \le\negl(n_\phys).
 \label{eq:better-cws-average-recovery}
\end{align}
\end{theorem}

\begin{proof}[Proof of \cref{thm:better-keyed-cws}]
Fix a sufficiently large even $n_\phys$.
We first reduce to fixed Pauli errors.
Suppose initially that $\cE_{n_\phys}$ has a Kraus representation whose operators have weight at most $t(n_\phys)$.
By \cref{lem:Pauli_twirling_for_local_noise}, averaging over $P$ replaces $\cE_{n_\phys}$ by a mixture of Pauli errors $Z^vX^u$ of weight at most $t(n_\phys)$.
This mixture is independent of $G$ because the noise is independent of the key.
It therefore suffices to bound the decoding error uniformly for each fixed $Z^vX^u$.

We next show exact recovery for a fixed $Z^vX^u$ whenever $\Recover(G,v\oplus Au)=u$.
Set $e=v\oplus Au$.
By \cref{thm:better-graph}, applied at input length $n_\phys/2$, every sampled graph satisfies $\weight(e)\le p n_\phys/2$.
For an arbitrary logical state $\ket{\psi}=\sum_x\alpha_x\ket{x}$, \cref{lem:error_propagation} gives, up to a global phase,
\begin{align}
 U_G^\dagger Z^vX^u U_G V_C\ket{\psi}
 =\sum_x\alpha_x(-1)^{u\cdot C(x)}\ket{C(x)\oplus e}.
\end{align}
By \cref{eq:better-classical-decoding,eq:better-cws-bit-decoding}, coherent classical decoding maps this state to
\begin{align}
 \ket{e}_\regB\otimes
 \sum_x\alpha_x(-1)^{u\cdot C(x)}\ket{x}_\regA.
\end{align}
Thus measuring $\regB$ reveals only $e$ and does not disturb the logical superposition.
If $\widehat u=u$, the final phase correction cancels every factor $(-1)^{u\cdot C(x)}$ and returns $\ket{\psi}$.
The same calculation holds with an arbitrary untouched reference register, so the recovered channel is exactly the identity for every such graph.

We now average over $G$.
Let $\cD_G$ denote Steps~\ref{step:better-cws-bit-decode} to \ref{step:better-cws-phase-correct} of the decoder, and write
$\cT_{G,u,v}=\cD_G\circ(U_G^\dagger Z^vX^u U_G)\circ V_C$,
identifying isometries and unitaries with their induced channels.
By \cref{thm:better-graph}, applied at input length $n_\phys/2$, the probability that $\Recover(G,e)\ne u$ is at most $\epsilon_G(n_\phys)=\negl(n_\phys)$, uniformly over the fixed errors under consideration.
On the complementary event $\cT_{G,u,v}=\identitymap_{n_\logic}$, and otherwise its diamond distance from the identity is at most one under our convention.
Hence convexity gives
\begin{align}
  \left\|\Exp_G\cT_{G,u,v}-\identitymap_{n_\logic}\right\|_\diamond
  \le\epsilon_G(n_\phys).
\end{align}
Averaging this bound over the Pauli mixture proves \cref{eq:better-cws-average-recovery} for the exact local channel.

Finally, for local noise as defined in \cref{def:local_noise}, choose an exact local channel $\widetilde{\cE}_{n_\phys}$ within diamond distance $\eta(n_\phys)=\negl(n_\phys)$ of $\cE_{n_\phys}$.
Contractivity under encoding and decoding, followed by the triangle inequality, bounds the left-hand side of \cref{eq:better-cws-average-recovery} by $\epsilon_G(n_\phys)+\eta(n_\phys)=\negl(n_\phys)$.
\end{proof}

\begin{remark}
We assume that the classical decoder in \cref{eq:better-classical-decoding} recovers every message from every bit error within the stated radius with certainty.
This exact guarantee simplifies the proof.
If classical decoding were guaranteed only with high probability for each fixed message, we would need more careful analysis, as in the robustness proof of PRICs in \cref{subsec:PRIC_robustness}.
\end{remark}

\ifnum\anonymous=1
\else
\paragraph{Acknowledgments.}
The authors thank Tomoyuki Morimae for helpful comments on the early draft.
%TM is supported by JST CREST JPMJCR23I3, JST Moonshot JPMJMS2061-5-1-1, JST FOREST, MEXT QLEAP, the Grant-in Aid for Transformative Research Areas (A) 21H05183, and the Grant-in-Aid for Scientific Research (A) No.22H00522. 
%Part of this work was done by TM at Columbia University as a visiting scientist.
\fi

\paragraph{AI usage statement.}
The authors developed the main ideas behind the constructions and their security and robustness proofs.
A weaker form of the graph-theoretic result was obtained through interactions with ChatGPT 5.5 Thinking, and its present form was obtained through interactions with ChatGPT 6 Astra.
Codex and Claude Code assisted with drafting and editing the manuscript.
The authors take full responsibility for the correctness of this work.

\ifnum\submission=0
\bibliographystyle{alpha} 
\else
\bibliographystyle{splncs04}
\fi
\bibliography{abbrev3,crypto,reference}

\appendix
\crefalias{section}{appendix}
\crefname{appendix}{Appendix}{Appendices}
\Crefname{appendix}{Appendix}{Appendices}

\section{Omitted Proofs}
\label{sec:omitted_proof}

In this section, we provide the proofs omitted in the main body.

\paragraph{Proof of \cref{lem:Pauli_twirling_for_local_noise}.}

\PauliTwirling*

\begin{proof}[Proof of \cref{lem:Pauli_twirling_for_local_noise}]
    By assumption, $\cE$ has a Kraus representation
    \begin{align}
        \cE(\cdot)=\sum_i E_i(\cdot)E_i^\dag
    \end{align}
    such that $\weight(E_i)\le t$ for all $i$.
    Fix $i$.
    Within this proof, let $\cP_m^0\coloneqq\{X^xZ^z:x,z\in\bit^m\}$ be the set of phase-free Pauli representatives.
    By the definition of $\weight(E_i)$, the Pauli expansion of $E_i$ contains only Pauli operators of weight at most $t$.
    Hence, we can write
    \begin{align}
        E_i=\sum_{Q\in\cP_m^0:\,\weight(Q)\le t}\alpha_{i,Q}Q,
    \end{align}
    for some coefficients $\alpha_{i,Q}\in\C$.
    For any $P\in\cP_m$, conjugation by a Pauli operator maps each
    canonical Pauli operator to itself up to a phase.
    This motivates us to define, for each $P\in\cP_m$, a map $\chi_P:\cP_m^0\to\{\pm1\}$ such that
    \begin{align}
        \chi_P(Q)\,Q\coloneqq P^\dag QP.
    \end{align}
    Therefore,
    \begin{align}
        P^\dag E_i P
        =
        \sum_{Q\in\cP_m^0:\,\weight(Q)\le t}\alpha_{i,Q}\chi_P(Q)\,Q.
    \end{align}
    It follows that
    \begin{align}
        (P^\dag\circ\cE\circ P)(\cdot)
        &=
        \sum_i (P^\dag E_i P)(\cdot)(P^\dag E_i P)^\dag \notag\\
        &=
        \sum_i
        \sum_{\substack{Q,R\in\cP_m^0\\ \weight(Q),\weight(R)\le t}}
        \alpha_{i,Q}\overline{\alpha_{i,R}}\,
        \chi_P(Q)\overline{\chi_P(R)}\,
        Q(\cdot)R^\dag.
    \end{align}
    Averaging over $P\gets\cP_m$, we obtain
    \begin{align}
        \cE^{\mathrm{tw}}(\cdot)
        &=
        \sum_i
        \sum_{\substack{Q,R\in\cP_m^0\\ \weight(Q),\weight(R)\le t}}
        \alpha_{i,Q}\overline{\alpha_{i,R}}\,
        \Exp_{P\gets\cP_m}\!\big[\chi_P(Q)\overline{\chi_P(R)}\big]\,
        Q(\cdot)R^\dag.
    \end{align}
    We claim that the following identity holds for $Q,R\in\cP_m^0$:
    \begin{align}
        \Exp_{P\gets\cP_m}\!\big[\chi_P(Q)\overline{\chi_P(R)}\big]
        =
        \begin{cases}
            1 & \text{if } Q=R,\\
            0 & \text{if } Q\neq R.
        \end{cases}\label{eq:Pauli_cahracter}
    \end{align}
    We prove this below.
    From \cref{eq:Pauli_cahracter}, we obtain
    \begin{align}
        \cE^{\rm tw}(\cdot)
        &=
        \sum_{Q\in\cP_m^0:\,\weight(Q)\le t}
        \left(\sum_i |\alpha_{i,Q}|^2\right) Q(\cdot)Q^\dag.
    \end{align}
    By defining $p_Q\coloneqq\sum_i|\alpha_{i,Q}|^2$ for $Q\in\cP_m^0$, we have
    \begin{align}
        \cE^{\rm tw}(\cdot)
        &=
        \sum_{Q\in\cP_m^0:\,\weight(Q)\le t} p_Q\, Q(\cdot) Q^\dag.
    \end{align}
    Since $\cE^{\rm tw}$ is trace preserving and each $Q\in\cP_m^0$ is unitary, the coefficients $p_Q$ form a probability distribution supported on $\cP_m^0\subseteq\cP_m$.
    
    To complete the proof, we prove \cref{eq:Pauli_cahracter}.
    If $Q=R$, then
    \begin{align}
        \chi_P(Q)\overline{\chi_P(R)}
        =
        \chi_P(Q)\overline{\chi_P(Q)}
        =
        1
    \end{align}
    for every $P\in\cP_m$, and hence the expectation is $1$.
    Now assume that $Q\neq R$. Then $QR^\dag$ is a non-identity Pauli operator. 
    Hence, there exists
    a Pauli operator $S\in\cP_m$ such that $S$ anticommutes with $QR^\dag$, i.e.,
    \begin{align}
        S^\dag QR^\dag S=-QR^\dag.
        \label{eq:S_anticomute_w_QR}
    \end{align}
    We argue that
    \begin{align}
        \chi_{SP}(Q)\overline{\chi_{SP}(R)}=-\chi_P(Q)\overline{\chi_P(R)}
        \label{eq:relation_of_chi}
    \end{align}
    for all $P\in\cP_m$.
    To see this, recall that
    \begin{align}
        (SP)^\dag Q (SP)=\chi_{SP}(Q)Q,
        \text{ and }
        (SP)^\dag R^\dag (SP)=\overline{\chi_{SP}(R)}R^\dag.
    \end{align}
    Thus, we have
    \begin{align}
        (SP)^\dag QR^\dag (SP)=\chi_{SP}(Q)\overline{\chi_{SP}(R)}QR^\dag.
        \label{eq:QR_conjugated_by_SP:1}
    \end{align}
    On the other hand,
    \begin{align}
        (SP)^\dag QR^\dag (SP)
        &=P^\dag (S^\dag QR^\dag S) P
        \notag\\
        &=-P^\dag QR^\dag P
        \tag{By \cref{eq:S_anticomute_w_QR}}\\
        &=-\chi_P(Q)\overline{\chi_P(R)}QR^\dag,
        \label{eq:QR_conjugated_by_SP:2}
    \end{align}
    where we used $P^\dag QP=\chi_P(Q)Q$ and $P^\dag R P=\chi_P(R)R$ in the last line.
    From \cref{eq:QR_conjugated_by_SP:1,eq:QR_conjugated_by_SP:2}, we have
    \begin{align}
        \chi_{SP}(Q)\overline{\chi_{SP}(R)}QR^\dag=-\chi_P(Q)\overline{\chi_P(R)}QR^\dag,
    \end{align}
    which shows \cref{eq:relation_of_chi}.
    Since left multiplication by $S$ permutes $\cP_m$, the terms in the average cancel in pairs by \cref{eq:relation_of_chi}.
    Therefore, we obtain
    \begin{align}
        \Exp_{P\gets\cP_m}\!\big[\chi_P(Q)\overline{\chi_P(R)}\big]
        =
        0,
    \end{align}
    which shows \cref{eq:Pauli_cahracter}.
    This concludes the proof.
\end{proof}

\paragraph{Proof of \cref{lem:error_propagation}.}

\ErrorPropagation*

\begin{proof}[Proof of \cref{lem:error_propagation}]
    Write $C_G \coloneqq \prod_{\{i,j\}\in E} \CZ_{ij}$, so that $U_G = C_G H^{\otimes n}$ and $U_G^\dagger = H^{\otimes n} C_G$.
    Since each $\CZ_{ij}$ is diagonal, it commutes with every $Z_k$. Hence
    \begin{align}
        C_G Z^v C_G = Z^v.
    \end{align}
    Now let $r_i\in\mathbb{F}_2^n$ denote the $i$th row of $A$. For each $i\in[n]$, conjugation by the
    controlled-$Z$ gates incident to $i$ gives
    \begin{align}
        C_G X_i C_G = X_i Z^{r_i}.
    \end{align}
    Therefore, since conjugation is multiplicative,
    \begin{align}
        C_G X^u C_G
        =
        \prod_{i:u_i=1} \bigl(C_G X_i C_G\bigr)
        =
        \prod_{i:u_i=1} X_i Z^{r_i}.
    \end{align}
    The factors on the right-hand side commute with each other, because for $i\neq j$ the possible sign
    from commuting $X_i$ past $Z^{r_j}$ is canceled by the sign from commuting $X_j$ past $Z^{r_i}$.
    Hence
    \begin{align}
        C_G X^u C_G
        =
        (-1)^{q_G(u)}X^u Z^{\sum_{i:u_i=1} r_i}
        =
        (-1)^{q_G(u)}X^u Z^{Au}.
    \end{align}
    Combining the above,
    \begin{align}
        C_G P C_G
        =
        C_G Z^v X^u C_G
        =
        Z^v X^u Z^{Au}.
    \end{align}
    Since $A$ is symmetric with zero diagonal, we have $u\cdot Au=0 \pmod 2$, and therefore $X^u Z^{Au} = Z^{Au} X^u.$
    Thus
    \begin{align}
    C_G P C_G = (-1)^{q_G(u)}Z^{v\oplus Au} X^u.
    \end{align}
    Finally, since $H^{\otimes n}$ swaps $X$ and $Z$, we have
    $U_G^\dagger P U_G=(-1)^{q_G(u)}X^{v\oplus Au} Z^u$.
    The factor $(-1)^{q_G(u)}$ is a global phase and therefore disappears at
    the channel level.
    This concludes the proof.
\end{proof}

\paragraph{Proof of \cref{lem:cptp_encoding_no_redundancy}}

\ImpossibilityOfCPTPEncode*

\begin{proof}[Proof of \cref{lem:cptp_encoding_no_redundancy}]
In the following, we omit $\secp$ dependence for simplicity.
Let
$\ket{\Phi}\coloneqq 2^{-n_\logic/2}\sum_{x\in\bit^{n_\logic}}\ket{x}\ket{x}$ and
$\Phi\coloneqq\ketbra{\Phi}{\Phi}$.
Fix $\sk$ and define
\begin{align}
    \mathcal B_\sk(\cdot)&\coloneqq(\Tr_{[t]}\circ\Enc_\sk)(\cdot), &
    \mathcal C_\sk(\cdot)&\coloneqq\Dec_\sk\left(\frac{I^{\otimes t}}{2^t}\otimes(\cdot)\right).
\end{align}
The map $\mathcal B_\sk$ maps $n_\logic$ qubits to
$n_\phys-t$ qubits, and $\mathcal C_\sk$ maps $n_\phys-t$ qubits
to $n_\logic$ qubits.
Moreover,
$\mathcal C_\sk\circ\mathcal B_\sk
=\Dec_\sk\circ\Delta_{[t]}\circ\Enc_\sk$.

Define the (unnormalized) Choi operators
\begin{align}
    X_\sk&\coloneqq 2^{n_\logic}(\mathcal B_\sk\otimes\identitymap)(\Phi), &
    Y_\sk&\coloneqq 2^{n_\logic}(\mathcal C_\sk^\dagger\otimes\identitymap)(\Phi),
\end{align}
where $\mathcal C_\sk^\dagger$ denotes the Hilbert--Schmidt adjoint of
$\mathcal C_\sk$, which is uniquely characterized by
$\Tr[A\mathcal C_\sk(B)]=\Tr[\mathcal C_\sk^\dagger(A)B]$
for every operator $A$ on $n_\logic$ qubits and every operator $B$ on
$n_\phys-t$ qubits.
Since $\mathcal B_\sk$ is trace preserving,
\begin{align}
    \Tr_{[n_\phys-t]}[X_\sk]&=I^{\otimes n_\logic},
\end{align}
where $\Tr_{[n_\phys-t]}$ denotes the partial trace over the first $n_\phys-t$ qubits.
Since $\Tr[A\mathcal C_\sk(B)]=\Tr[\mathcal C_\sk^\dagger(A)B]$
for every operator $A$ on $n_\logic$ qubits and every operator $B$ on
$n_\phys-t$ qubits, we have
\begin{align}
    \Tr[Y_\sk]
    =\Tr\left[\mathcal C_\sk^\dagger(I^{\otimes n_\logic})\right]
    =\Tr\left[\mathcal C_\sk(I^{\otimes (n_\phys-t)})\right]
    =2^{n_\phys-t}.
\end{align}
We claim that $\|X_\sk\|_\infty\le 2^{n_\phys-t}$. If
$\lambda=\|X_\sk\|_\infty$ and $\ket{v}$ is a normalized eigenvector for
$\lambda$, then we have $X_\sk\ge\lambda\ketbra{v}{v}$.
Since $\Tr_{[n_\phys-t]}[X_\sk]=I^{\otimes n_\logic}$, we have
\begin{align}
    I^{\otimes n_\logic}
    \ge\lambda\,\Tr_{[n_\phys-t]}[\ketbra{v}{v}].
\end{align}
The reduced state on the right has rank at most $2^{n_\phys-t}$ and trace
one, so its largest eigenvalue is at least $2^{-(n_\phys-t)}$. Hence
$\lambda\le 2^{n_\phys-t}$, which means $\|X_\sk\|_\infty\le 2^{n_\phys-t}$.

The Choi identity and the definition of the adjoint map give
\begin{align}
    \bra{\Phi}
    ((\mathcal C_\sk\circ\mathcal B_\sk)\otimes\identitymap)(\Phi)
    \ket{\Phi}
    =\Tr\left[
       (\mathcal C_\sk^\dagger\otimes\identitymap)(\Phi)
       (\mathcal B_\sk\otimes\identitymap)(\Phi)
       \right]
    =\frac{\Tr[X_\sk Y_\sk]}{2^{2n_\logic}}
\end{align}
By Hölder's inequality and the bounds established above,
$\Tr[X_\sk Y_\sk]\le\|X_\sk\|_\infty\Tr[Y_\sk]\le 2^{2(n_\phys-t)}$, which implies
\begin{align}
    \bra{\Phi}
    ((\mathcal C_\sk\circ\mathcal B_\sk)\otimes\identitymap)(\Phi)
    \ket{\Phi}
      \le\frac{2^{2(n_\phys-t)}}{2^{2n_\logic}}
      =4^{n_\phys-n_\logic-t}.
\end{align}
Set $\mathcal F\coloneqq\Exp_\sk[\Dec_\sk\circ\Delta_{[t]}\circ\Enc_\sk]$.
By linearity and $\mathcal C_\sk\circ\mathcal B_\sk
=\Dec_\sk\circ\Delta_{[t]}\circ\Enc_\sk$ for each $\sk$,
\begin{align}
    \bra{\Phi}(\mathcal F\otimes\identitymap)(\Phi)\ket{\Phi}
    =\Exp_\sk[\bra{\Phi}((\mathcal C_\sk\circ\mathcal B_\sk)\otimes\identitymap)(\Phi)\ket{\Phi}]\le\min\{1,4^{n_\phys-n_\logic-t}\}.
\end{align}
Since $\frac{1}{2}\|\rho-\sigma\|_1=\max_{P}|\Tr[P(\rho-\sigma)]|$, where the maximization is taken over all projections, we have
\begin{align}
    \frac12\left\|(\mathcal F\otimes\identitymap)(\Phi)-\Phi\right\|_1
    \ge\left|\Tr\left[\Phi\left((\mathcal F\otimes\identitymap)(\Phi)-\Phi\right)\right]\right|
    =1-\bra{\Phi}(\mathcal F\otimes\identitymap)(\Phi)\ket{\Phi}.
\end{align}
By the definition of the diamond norm and the above two inequalities, we have
\begin{align}
    \|\mathcal F-\identitymap\|_\diamond
    \ge\frac12\left\|(\mathcal F\otimes\identitymap)(\Phi)-\Phi\right\|_1
    \ge\max\{0,1-4^{n_\phys-n_\logic-t}\},
\end{align}
which completes the proof.
\end{proof}

\iffalse
\paragraph{Proof of \cref{lem:induced_error_weight}.}

\DegreeVSWeight*

\begin{proof}[Proof of \cref{lem:induced_error_weight}]
    Let $P = Z^v X^u$ be a Pauli operator with $\weight(P) \le w$.
    The support of $P$ is
    \begin{align}
        S \coloneqq \Support(u) \cup \Support(v),
    \end{align}
    where we recall that $\Support(u)=\{i:u_i=1\}$ and $\Support(v)=\{i:v_i=1\}$.
    Since $\weight(P)\le w$, we have $|S| \le w$.
    By \cref{lem:error_propagation}, conjugating $P$ by the graph-basis unitary $U_G$ yields $U_G^\dagger P U_G = (-1)^{q_G(u)}X^{v\oplus Au} Z^u$; the global phase is irrelevant for the weight bound.
    By the triangle inequality for the Hamming weight,
    \begin{align}
        \weight(v\oplus Au)
        \le
        \weight(Au) + \weight(v).
    \end{align}
    Since the maximum degree of $G$ is $\Delta$, each vertex in $\Support(u)$ has at most $\Delta$ neighbors.
    Therefore,
    \begin{align}
        \weight(Au) \le \Delta\,\weight(u).
    \end{align}
    Because both $\Support(u)$ and $\Support(v)$ are subsets of $S$, their individual weights are bounded by $|S| \le w$.
    Thus,
    \begin{align}
        \weight(v\oplus Au)
        \le
        \Delta |S| + |S|
        =
        (\Delta+1)|S|
        \le
        (\Delta+1)w,
    \end{align}
    which concludes the proof.
\end{proof}
\fi

\section{Post-Quantum Security of Classical Constructions}
\label{sec:post_quantum_CG24}

In this section, we show that post-quantum PRCs exist under \cref{assumption:for_PRC}.
In the remainder of this section, we omit the qualifier ``post-quantum'' when referring to PRCs.
The constructions in this section are exactly the constructions of~\cite{C:ChrGun24}.  Our additional argument is only the
lifting of their classical reductions to QPT adversaries.  As in
\cref{def:PKPRC,def:SKPRC}, the adversary is allowed to make only adaptive
classical queries to the encoding oracle.  
Following~\cite{C:ChrGun24}, we first give zero-bit PRCs, then lift them to multi-bit PRCs, and finally make them constant rate. 

\paragraph{Zero-bit construction and security.}
We first recall the zero-bit construction.

\begin{construction}[Zero-bit public-key PRC (\cite{C:ChrGun24}, Definition~3 and Construction~2)]
    \label{const:CG24_zero_bit}
    Fix efficiently computable functions
    $n,g,t,r:\N\to\N$ and $\eta,\zeta:\N\to[0,1/2)$.  On security
    parameter $\secp$, sample a matrix
    $P\in\F_2^{r(\secp)\times n(\secp)}$ whose rows are independent and
    uniform among vectors of Hamming weight $t(\secp)$, and then sample
    $G\in\F_2^{n(\secp)\times g(\secp)}$ uniformly subject to $PG=0$.
    Sample $z\gets\F_2^{n(\secp)}$ and output
    \begin{align*}
        \KeyGen(1^\secp)&:\quad \sk=(P,z),\quad \pk=(G,z),\\
        \Enc(\pk)&:\quad u\gets\F_2^{g(\secp)},\ e\gets\Ber(n(\secp),\eta(\secp)),
        \quad \Enc(\pk)=Gu\oplus z\oplus e,\\
        \Dec(\sk,c)&:\quad
        \begin{cases}
        1 & \text{if }\weight(Pc\oplus Pz)<(\frac12-\zeta(\secp))r(\secp),\\
        \bot & \text{otherwise.}
        \end{cases}
    \end{align*}
    The message space is the singleton zero-bit message space.
\end{construction}

For the planted-XOR case, we use the following lemma.
Since the proof is the same as that of~\cite{C:ChrGun24} except using the post-quantum XOR assumption, we omit the proof.

\begin{lemma}[Generator-matrix pseudorandomness (\cite{C:ChrGun24}, Lemma~8)]
    \label{lem:CG24_generator_pseudorandomness}
    If $g+r\le n-\omega(\log n)$, then under the post-quantum $\XOR_{g+r,t}$ assumption, the marginal distribution on $G$ in \cref{const:CG24_zero_bit} is computationally indistinguishable from the uniform distribution on $\F_2^{n\times g}$ against QPT distinguishers.
\end{lemma}

We use the following statistical fact for security.

\begin{lemma}[Generator-matrix transition (\cite{C:ChrGun24}, Lemma~9)]
    \label{lem:CG24_generator_transitions}
    If $r\le0.99n$ and $\omega(\sqrt{\log n})\le t\le O(\log n)$, then for some $g=\Omega(t^2)$ the marginal distribution on $G$ in \cref{const:CG24_zero_bit} is statistically $\negl(n)$-close to uniform in $\F_2^{n\times g}$.
\end{lemma}

Combining these facts gives the following.

\begin{lemma}[Zero-bit PRCs]
    \label{lem:CG24_zero_bit_guarantee}
    Under \cref{assumption:for_PRC}, \cref{const:CG24_zero_bit} is a $(0,n)$ public-key PRC robust to $p$-bounded channels for the following parameters and assumptions.
    \begin{enumerate}
        \item For constants
        $p,\eta\in(0,1/2)$ and $\varepsilon\in(0,1)$, take
        $(g,t,r,\zeta)=(n^\varepsilon,\Theta(\log n),n^\varepsilon,
        n^{-\varepsilon/4})$ under the
        $\LPN_{n^\varepsilon,\eta}$ and
        $\XOR_{2n^\varepsilon,t}$ assumptions.  
        \item For constants
        $p,\eta\in(0,1/2)$, take
        $(g,t,r,\zeta)=(\Omega(\log^2 n),\Theta(\log n),0.99n,
        (0.99n)^{-1/4})$ under the
        $\LPN_{g,\eta}$ assumption.  
    \end{enumerate}
\end{lemma}

\begin{proof}[Proof of \cref{lem:CG24_zero_bit_guarantee}]
    Robustness and soundness are given in~\cite{C:ChrGun24}.
    For the first item, \cref{lem:CG24_generator_pseudorandomness} replaces the public matrix $G$ by a uniform matrix under the $\XOR_{2n^\varepsilon,t}$ assumption.  
    The $\LPN_{n^\varepsilon,\eta}$ assumption then replaces the classical encoder responses by independent uniform strings, which implies security against QPT adversaries.
    For the second item, the generator-matrix transition is statistical by \cref{lem:CG24_generator_transitions}, and the same conclusion holds under the $\LPN_{g,\eta}$ assumption.
\end{proof}

\paragraph{Multi-bit conversion.}
We next convert the zero-bit primitive into a public-key PRC with a
positive polynomial message length.

\begin{construction}[Multi-bit public-key PRC (\cite{C:ChrGun24}, Construction~4)]
    \label{const:CG24_multi_bit}
    Let $\mathrm{PRC}_0=(\KeyGen_0,\Enc_0,\Dec_0)$ be a $(0,n)$
    public-key PRC, and let $\ell=\poly(\secp)$.
    On input $1^\secp$, sample $(\sk_0,\pk_0)\gets\KeyGen_0(1^\secp)$ and
    a uniform permutation $\pi:[n(\ell+1)]\to[n(\ell+1)]$, and output
    $\sk=(\sk_0,\pi)$ and $\pk=(\pk_0,\pi)$.

    On message $b=(b_1,\ldots,b_\ell)\in\bit^\ell$, form blocks
    $c_1,\ldots,c_{\ell+1}\in\bit^n$ by
    \begin{align*}
        c_i&\gets\bit^n &&\text{if }i\le\ell\text{ and }b_i=0,\\
        c_i&\gets\Enc_0(\pk_0) &&\text{if }i\le\ell\text{ and }b_i=1,\\
        c_{\ell+1}&\gets\Enc_0(\pk_0).&&
    \end{align*}
    Let $y=c_1\|\cdots\|c_{\ell+1}$ and output the permuted string
    $x=y_{\pi(1)}\|\cdots\|y_{\pi(n(\ell+1))}$.  To decode, apply
    $\pi^{-1}$ and parse the result into blocks
    $\widehat y_1,\ldots,\widehat y_{\ell+1}$ of length $n$.  Set
    $\widehat b_i=1$ exactly when
    $\Dec_0(\sk_0,\widehat y_i)=1$, and set it to $0$ otherwise.  If
    $\widehat b_{\ell+1}\ne1$, output $\bot$; otherwise output
    $(\widehat b_1,\ldots,\widehat b_\ell)$.
\end{construction}

The cited conversion result gives the robustness and soundness
guarantee for this construction.

\begin{lemma}[Multi-bit PRCs]
    \label{lem:CG24_multi_bit_guarantee}
    Let $\mathrm{PRC}_0$ be any $(0,n)$ public-key PRC robust to $p$-bounded channels.
    Then \cref{const:CG24_multi_bit} is an $(\ell,n\cdot(\ell+1))$ public-key PRC robust to $(p-\varepsilon)$-bounded channels for every $\varepsilon\in(0,p)$.
\end{lemma}

\begin{proof}[Proof of \cref{lem:CG24_multi_bit_guarantee}]
    Robustness and soundness are given in~\cite{C:ChrGun24}.
    Let $\cA$ distinguish the multi-bit real oracle from its ideal oracle.
    The two games are:
    \begin{itemize}
        \item $H_0$ is the real multi-bit oracle, with blocks of
        message bit $0$ sampled uniformly and the remaining blocks
        obtained from the real base encoder.
        \item $H_1$ is the ideal multi-bit oracle, in which every
        response is a fresh uniform string.
    \end{itemize}
    The reduction $\cB$ receives the base public key $\pk_0$, samples a
    uniform permutation $\pi$, and gives $(\pk_0,\pi)$ to $\cA$.  For each
    classical query $b=(b_1,\ldots,b_\ell)$, it samples the blocks with
    $b_i=0$ itself, queries the base encoder for every block with $b_i=1$
    and for the sentinel block, concatenates the blocks, and applies
    $\pi$.
    If the base oracle is real, $\cB$ simulates $H_0$ exactly; if it is ideal, all blocks are independent uniform strings, so $\cB$ simulates $H_1$ exactly.
    Since $\cA$ makes only classical queries, $\cB$ runs $\cA$ as a QPT subroutine and makes only polynomially many base-oracle queries.
    Thus a QPT distinguisher between $H_0$ and $H_1$ would distinguish the real and ideal zero-bit oracles.
\end{proof}

\paragraph{Constant-rate construction.}
We finally give PRCs with a constant encoding rate.
To this end, we need pseudorandom generators.

\begin{definition}[Pseudorandom generator]
    Let $\ell,s:\N\to\N$ be polynomially bounded functions with $s(\secp)>\ell(\secp)$.
    A deterministic polynomial-time family $\operatorname{PRG}:\bit^{\ell(\secp)}\to\bit^{s(\secp)}$ is a pseudorandom generator (PRG) if for every QPT distinguisher $\cD$,
    \[
        \left|\Pr_{r\gets\bit^{\ell(\secp)}}[\cD(\operatorname{PRG}(r))=1]
        -\Pr_{y\gets\bit^{s(\secp)}}[\cD(y)=1]\right|
        \le\negl(\secp).
    \]
\end{definition}

The resulting constant-rate public-key construction is as follows.

\begin{construction}[Constant-rate public-key PRC (\cite{C:ChrGun24}, Construction~5)]
    \label{const:CG24_constant_rate}
    Let $\secp$ be the security parameter and $\ell=\ell(\secp)$ the seed length, and let $\mathrm{PRC}_\ell$ be an $(\ell,n')$ public-key PRC.
    Let $(\Enc,\Dec)$ be a classical error-correcting code with $\Enc:\bit^n\to\bit^{n''}$, where $n''>\ell$, and let $\operatorname{PRG}:\bit^\ell\to\bit^{n''}$ be a pseudorandom generator.
    Key generation samples $(\sk',\pk')\gets\mathrm{PRC}_\ell.\KeyGen(1^\secp)$ and a uniform permutation $\pi:[n'+n'']\to[n'+n'']$, and outputs $\sk=(\sk',\pi)$ and $\pk=(\pk',\pi)$.
    On message $x\in\bit^n$, sample $r\gets\bit^\ell$ and set
    \begin{align*}
        y&=\mathrm{PRC}_\ell.\Enc(\pk',r)
        \|\bigl(\operatorname{PRG}(r)\oplus\Enc(x)\bigr).
    \end{align*}
    Output the permutation of $y$ induced by $\pi$.
    To decode, apply $\pi^{-1}$, parse the result as $y_1\|y_2$ with $|y_1|=n'$ and $|y_2|=n''$, and compute $r\gets\mathrm{PRC}_\ell.\Dec(\sk',y_1)$.
    If $r=\bot$, output $\bot$; otherwise output
    $\Dec(\operatorname{PRG}(r)\oplus y_2)$.
    The resulting construction encodes $n$-bit messages into $(n'+n'')$-bit codewords.
\end{construction}

\begin{lemma}[Constant-rate PRCs]
    \label{lem:CG24_constant_rate_guarantee}
    For constants $0<\varepsilon<p<1/2$, if $\mathrm{PRC}_\ell$ and $(\Enc,\Dec)$ are robust to noise that flips at most a $p$ fraction of bits at random locations, then \cref{const:CG24_constant_rate} is an $(n,n'+n'')$ public-key PRC robust to every $(p-\varepsilon)$-bounded channel.
\end{lemma}

\begin{proof}[Proof of \cref{lem:CG24_constant_rate_guarantee}]
    Robustness and soundness are given in~\cite{C:ChrGun24}.
    Let $\cA$ be a QPT distinguisher.
    Define the hybrids as follows:
    \begin{itemize}
        \item $H_0$ is the real encoder oracle of \cref{const:CG24_constant_rate}.
        \item $H_1$ replaces $\mathrm{PRC}_\ell.\Enc(\pk',r)$ by an independent uniform block in each encoding response.
        \item $H_2$ replaces $\operatorname{PRG}(r)$ by an independent uniform string in each encoding response.
    \end{itemize}
    The pseudorandomness of $\mathrm{PRC}_\ell$ against QPT adversaries implies that $H_0$ and $H_1$ are computationally indistinguishable.

    Replace $\operatorname{PRG}(r)$ with an independent uniform string in the first encoding response, then in the second, and so on.
    Each replacement is computationally indistinguishable under the PRG assumption against QPT distinguishers.
    Thus $H_1$ and $H_2$ are computationally indistinguishable.
    At the end, the resulting oracle outputs uniform bit strings for each query, which is the ideal oracle.
\end{proof}

Now we show \cref{thm:post-quantum_PRC}.

\PostQuantumPRC*

\begin{proof}[Proof of \cref{thm:post-quantum_PRC}]
    The case $p=0$ follows from any positive noise rate, so fix $p\in(0,1/2)$ and $\varepsilon>0$ with $p+2\varepsilon<1/2$.
    For sufficiently large $\alpha$, choose the zero-bit block length and the seed length $\ell$ as sufficiently small positive powers of $\secp$ so that $n'=o(m(\secp))$.
    By \cref{lem:CG24_zero_bit_guarantee,lem:CG24_multi_bit_guarantee}, there are $(\ell,n')$ public-key PRCs robust to $(p+\varepsilon)$-bounded noise.
    The same constructions at scale $\ell$, with the public key included in the secret key, imply an $\ell$-bit-seeded PRG by \cref{coro:PRC_imply_PRF} and the implication of PRGs from PRFs.
    Choose $n''$ satisfying $n''=m-n'$.
    Then, \cref{lem:CG24_constant_rate_guarantee} gives $(n,m)$ public-key PRCs robust to $p$-bounded noise.
    These also imply $(n,m)$ secret-key PRCs robust to $p$-bounded noise, which completes the proof.
\end{proof}

\end{CJK}
\end{document}